\documentclass{llncs}
\usepackage[margin=3cm]{geometry}
\usepackage[pdfpagelabels, pagebackref, naturalnames]{hyperref}

\usepackage[utf8]{inputenc}
\usepackage{graphicx}
\usepackage{paralist}
\usepackage{subfigure}
\usepackage{caption}
\usepackage{amssymb}
\usepackage{mathtools}
\usepackage{booktabs}
\usepackage{multirow}
\usepackage{xspace}
\usepackage{tikz}
\usepackage[rightcaption]{sidecap}
\usepackage{wrapfig}
\usepackage{nicefrac}

\DeclareMathAlphabet{\mathcal}{OMS}{cmsy}{m}{nx}

\newcommand{\switches}{\ensuremath{\mathcal{W}}}
\newcommand{\cands}{\ensuremath{\mathcal K}}
\newcommand{\stops}{\ensuremath{\mathcal S}}
\newcommand{\lines}{\ensuremath{\mathcal M}}
\newcommand{\firstStop}{\ensuremath{\underline{s}}}
\newcommand{\lastStop}{\ensuremath{\overline{s}}}
\newcommand{\sgn}{\operatorname{sgn}}

\newcommand{\RR}{\mathrm{R}}
\newcommand{\LL}{\mathrm{L}}

\newcommand{\consSwitches}{\ensuremath{\Gamma}}
\newcommand{\consLabels}{\ensuremath{P}}

\newcommand{\firstLabel}{\ensuremath{h}}
\newcommand{\lastLabel}{\ensuremath{t}}

\newtheorem{assumption}{Assumption}

\newcommand{\DYNALG}{\textsc{DpAlg}\xspace}
\newcommand{\SCALEALG}{\textsc{ScaleAlg}\xspace}
\newcommand{\GREEDYALG}{\textsc{GreedyAlg}\xspace}
\newcommand{\ILPALG}{\textsc{IlpAlg}\xspace}
\newcommand{\PROPERTY}{transitivity property\xspace}

\title{An Algorithmic Framework for Labeling Network Maps\thanks{A preliminary version of
    this paper has appeared in \emph{Proc.\ 21st Int.\ Conf.\ on Computing Combinatorics (COCOON'15)}, volume 9198 of \emph{Lect. Notes
      Comput. Sci.}, pages 689--700, Springer-Verlag.  This research
    was initiated during Dagstuhl Seminar 13151 “Drawing Graphs and
Maps with Curves” in April 2013}}

\author{Jan-Henrik Haunert\inst{1} \and Benjamin Niedermann\inst{2}
}

\institute{University of Osnabr\"{u}ck, Germany \and Karlsruhe Institute of Technology, Germany }

\begin{document}

\maketitle

\begin{abstract}
  Drawing network maps automatically comprises two challenging steps,
  namely laying out the map and placing non-overlapping labels. In
  this paper we tackle the problem of labeling an already existing
  network map considering the application of metro maps.  We present a
  flexible and versatile labeling model that subsumes different
  labeling styles. We show that labeling a single line of the network
  is NP-hard, even if we make very restricting assumptions about the
  labeling style that is used with this model.  For a restricted
  variant of that model, we then introduce an efficient algorithm that
  optimally labels a single line with respect to a given cost
  function. Based on that algorithm, we present a general and
  sophisticated workflow for multiple metro lines, which is
  experimentally evaluated on real-world metro maps.
\end{abstract}

\section{Introduction}\label{sec:introduction}
Label placement and geographic network visualization are classical
problems in cartography, which independently of each other have
received the attention of computer scientists.  Label placement
usually deals with annotating point, line or area features of interest
in a map with text labels such that the associations between the
features and the labels are clear and the map is kept legible
\cite{imhof}.  Geographic network visualization, on the other hand,
often aims at a geometrically distorted representation of reality that
allows information about connectivity, travel times, and required
navigation actions to be retrieved easily.  Computing a good network
visualization is thus related to finding a layout of a graph with
certain favorable properties \cite{wolff2013}.  For example, to avoid
visual clutter in metro maps, an \emph{octilinear} graph layout is
often chosen, in which the orientation of each edge is a multiple of
45$^\circ$ \cite{metroMap1,metroMap2,wang2011}.  Alternatively, one
may choose a \emph{curvilinear} graph layout, that is, to display the
metro lines as curves \cite{fink,goethem}.

Computing a graph layout for a metro map and labeling the stops
have been considered as two different problems that can be solved in
succession \cite{wang2011}, but also integrated solutions have been
suggested \cite{metroMap1,metroMap2}. Nevertheless, in practice,
metro maps are often drawn manually by cartographers or designers, as
the existing algorithms do not achieve results of sufficient quality
in adequate time. For example, N\"ollenburg and Wolff \cite{metroMap1}
report that their method needed 10 hours and 31 minutes to compute a
labeled metro map of Sydney that they present in their article, while
an unlabeled map for the same instance was obtained after 23
minutes---both results were obtained without proof of optimality but
with similar optimality gaps. On the other hand Wang and Chi \cite{wang2011} present an algorithm that creates the graph layout and labeling within one second, but they cannot guarantee that labels do not overlap each other or the metro lines.

An integrated approach to computing a graph layout and labeling the
stops allows consideration to be given to all quality criteria of
the final visualization. On the other hand, treating both problems
separately will probably reduce computation time. Moreover, we
consider the labeling of a metro map as an interesting problem on its
own, since, in some situations, the layout of the network is given as
part of the input and must not be changed. In a semi-automatic
workflow, for example, a cartographer may want to draw or alter a
graph layout manually before using an automatic method to place
labels, probably to test multiple different labeling styles with the
drawing. Hence, a labeling algorithm is needed that is
rather flexible in dealing with different labeling styles.

In this paper, we are given the layout of a metro map consisting of
several metro lines on which stops (also called stations) are located.
For each stop we are further given its name, which should be placed
close to its position. We first introduce a versatile and general
model for labeling metro maps; see Section~\ref{sec:model}.  Like many
labeling algorithms for point sets
\cite{aks-lpmis-98,cms-esapf-95,fw-ppalm-91}, our algorithm uses a
discrete set of candidate labels for each point.  Often, each label is
represented by a rectangle wrapping the text.  Since we also want to
use curved labels, however, we represent a label by a simple polygon
that approximates a \emph{fat curve}, that is, a curve of certain width
reflecting the text height.  We then prove that even in that
simple model labeling a single metro line is NP-hard considering
different labeling styles. Hence, we restrict the set of candidates
satisfying certain properties, which allows us to solve the problem on one metro line~$C$
in~$O(n^2)$ time, where~$n$ is the number of stops of~$C$;
see Section~\ref{sec:single-line}. This algorithm optimizes
the labeling with respect to a cost function that is based on Imhof's \cite{imhof}
classical criteria of cartographic quality.  Utilizing
that algorithm, we present an efficient heuristic for labeling a
metro map consisting of multiple metro lines;
see Section~\ref{sec:multi-lines}. 
Our method is similar to the heuristic presented by Kakoulis and Tollis~\cite{Kakoulis1998},
in the sense that it discards some label candidates 
to establish a set of preconditions that allow for an efficient exact solution.
Our model of quality is more general than the one of Kakoulis and Tollis, however,
as it not only takes the quality of individual labels but also the quality of pairs of labels for consecutive metro stations into account.
Finally, we evaluate our approach
presenting experiments conducted on realistic metro maps;
see Section~\ref{sec:evaluation}. 
Note that
``stops'' on ``metro lines'' can refer more generally to points of
interest on the lines of any kind of a network map. We address
labeling styles for octilinear graph layouts and curvilinear graph
layouts that use B\'ezier curves. The more general model behind our
method, however, subsumes but is not limited to these particular
styles.

\section{Labeling Model}\label{sec:model}



\newcommand{\LS}{\textsc{Ls}}
\newcommand{\MML}{\textsc{Metro\-Map\-Label\-ing}\xspace}
\newcommand{\RMLL}{\textsc{Soft\-Metro\-Li\-ne\-La\-bel\-ing}\xspace} We
assume that the metro lines are given by directed,
non-self-intersecting curves in the plane described by polylines,
which for example have been derived by approximating B\'{e}zier curves. We
denote that set of metro lines by $\lines$. Further, the stops of each
metro line~$C\in \lines$ are given by an ordered set~$\stops_C$ of
points on~$C$ going from the beginning to the end of~$C$.  For two
stops~$s,s'\in \stops_C$ we write~$s<s'$ if $s$ lies before~$s'$.
We denote the union of the stops among all metro lines by $\stops$ and
call the pair~$(\lines,\stops)$ a \emph{metro map}.  

For each stop~$s\in \stops$ we are further given a name that
should be placed close to it. 
In contrast to previous work,
we do not follow traditional map labeling abstracting from the given
text by bounding boxes. Instead we model a \emph{label}~$\ell$ of a stop~$s\in \stops$ as a simple polygon. For
example, a label could have been derived by approximating a fat curve
prescribing the name of the stop; see Fig.~\ref{fig:curved_candidates}.
For each stop~$s$ we are given a set~$\cands_s$ of labels, which we
also call~\emph{candidates} of~$s$. The set $\bigcup_{s\in \stops} \cands_s$ is denoted by~$\cands$.

Since ``names should disturb other map content as little as
possible''\cite{imhof}, we strictly forbid overlaps between labels and lines as
well as label-label overlaps. Further, each stop must be labeled.
Hence, a set~$\mathcal L\subseteq \cands$ is called a
\emph{labeling} if
\begin{inparaenum}[(1)]
\item no two labels of~$\mathcal L$ intersect each other,
\item no label~$\ell\in \mathcal L$ intersects any metro line~$C\in
  \lines$, and
\item for each stop~$s\in \stops$ there is exactly one label~$\ell\in
  \mathcal L\cap \cands_s$.
\end{inparaenum}

\begin{definition}[\MML]

\noindent  \emph{Given:} Metro map~$(\lines,\stops)$, candidates~$\cands$ and cost function~$w\colon 2^{\cands}\to
\mathbb{R}^+$.

\noindent  \emph{Find, if it exists:} Optimal labeling~$\mathcal L$ of
$(\lines,\stops,\cands,w)$, i.e., $w(\mathcal L)\leq w(\mathcal L')$
for any labeling~$\mathcal L'\subseteq \cands$. 
\end{definition}

\newcommand{\segments}{H}
\newcommand{\CurvedLabels}{\textsc{Curved\-Style}\xspace}
\newcommand{\OctiLabels}{\textsc{Octi\-lin\-Style}\xspace}


\begin{figure}[t]
  \centering \subfigure[]{
    \includegraphics[page=1,scale=1]{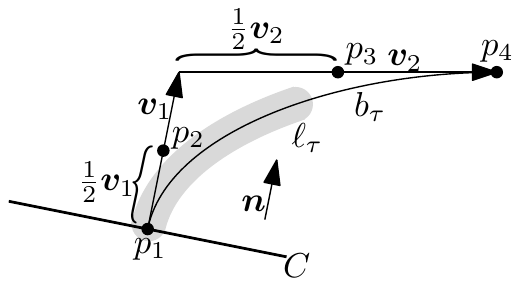}
    \label{fig:curved_candidates:single}
  }
\subfigure[]{
    \includegraphics[page=2,scale=0.8]{candidates}
    \label{fig:curved_candidates:multiple}
  }
  \caption{Construction of curved
    candidates. \protect\subref{fig:curved_candidates:single}
    Construction of a single
    label. \protect\subref{fig:curved_candidates:multiple}
    Candidates~$\cands_s=L_{1}\cup L_{-1}$ for stop~$s$. }
  \label{fig:curved_candidates} 
\end{figure}

The model allows us to create arbitrarily  shaped label candidates for a
metro map. In our evaluation we have considered two different
\emph{labeling styles}. 
The first style, \OctiLabels, creates for each
stop a set of octilinear rectangles as label candidates; see
Fig.~\ref{fig:octilinear_candidates}. We use that style for
octilinear maps.
The second style, \CurvedLabels, creates for each stop a set of fat B\'{e}zier curves as label candidates, which are then approximated by simple polygons; see Fig.~\ref{fig:curved_candidates}. 
We use that style for curvilinear metro maps, in order to adapt the curvilinear style of the metro map. 
The basic idea is that a label perpendicularly emanates from a stop with respect to its metro line and then becomes horizontal to
sustain legibility. In the following section we motivate our choice of candidates based on cartographic criteria and give detailed technical descriptions for both labeling styles.

\subsection{Two Examples of Labeling Styles}

We extracted the rules for
generating label candidates  from Imhof's ``general principles and
requirements'' for map labeling \cite{imhof}.
For schematic network maps, the need for ``legibility'' implies that
we must not destroy the underlying design principle with clutter. To
this end, we generate candidate labels that adhere to the schematics
of the network. That is, we use straight horizontal and diagonal
labels with octilinear layouts and curved labels with curvilinear
layouts. 

We now describe more precisely, how we defined the labeling styles \CurvedLabels and \OctiLabels, which are used for curvilinear layouts and octilinear layouts, respectively.

\textbf{Curvilinear Metro Maps.}
For \CurvedLabels, assume that the given metro map is
curvilinear.
In order to achieve a ``clear graphic association'' between a label
and the corresponding point $p$, we construct the simple polygon
prescribing a candidate label based on a curve $\ell$ (possibly a
straight-line segment) that emanates from $p$. The candidate label
itself is a continuous section of $\ell$ that does not directly start
in $p$ but at a certain configurable distance from it. We define the
end of the candidate label on $\ell$ based on the text length and
assign a non-zero width to the curve section to represent the text
height. In the case that $p$ lies on a single curved line $C$, we
require that $\ell$ and $C$ are perpendicular in $p$ to enhance the
angular resolution of the final drawing. By bending $\ell$ towards the
horizontal direction, we avoid steep labels. We approximate~$\ell$ by a simple polygon consisting of a constant number of line segments.

We now describe the construction of a single candidate more
specifically.  For each stop~$s$ of each metro line~$C$ we create a
constant number of curved labels adapting the curvilinear style of the metro
map. The basic idea is that a label perpendicularly emanates from~$s$
with respect to~$C$ and then becomes horizontal to sustain legibility;
see Fig.~\ref{fig:curved_candidates}. Let~$\vec n=(n_x,n_y)$ be the
normalized normal vector of~$C$ at~$s$.  Further, let $d\in\{-1,1\}$
and $c_1,c_2\in \mathbb{R}^+$ be pre-defined
constants. For~$\tau=(c_1,c_2,d)$ we define the fat cubic B{\'e}zier
curve~$b_{\tau}$ by the following four control points; see
Fig.~\ref{fig:curved_candidates}.
 $ p_1= s,~ 
  p_2= s+ 0.5 \cdot \vec v_1,~
  p_3= s+ \vec v_1 + 0.5 \cdot \vec v_2,~
  p_4= s+ \vec v_1 + \vec v_2$,
  where $\vec v_1=c_1\cdot \vec n$, $\vec v_2=\sgn(\vec n)\cdot
  (d\cdot c_2,0)$, and~$\sgn(\vec n)=1$ if~$n_x>0$ and~$\sgn(\vec
  n)=-1$ otherwise. We define the thickness of~$b_\tau$ to be the
  pre-defined height of a label.  Let~$\ell_\tau$ be the sub-curve
  of~$b_{\tau}$ that starts at $p_1$ and has the length of the name
  of~$s$ and let~$\ell'_\tau$ be the curve when mirroring~$\ell_\tau$
  at~$s$. Further, let~$l_{m}$ be the length of the longest name of a
  stop in~$\stops$ and let~$L_d=\{\ell_{\tau},\ell'_\tau\mid \tau
  \in\{(l_{m},l_{m},d),(\frac{l_{m}}{2},l_{m},d),(\frac{l_{m}}{4},l_{m},d)\}\}$.
  If~$\vec n$ has an orientation less than or equal to~$60^\circ$, we
  set~$\cands_s=L_1$ and otherwise $\cands_s=L_{1}\cup L_{-1}$. Hence,
  if~$\vec n$ is almost vertical and~$C$ is therefore almost
  horizontal at~$s$, we also add the labels~$L_{-1}$ pointing into the
  opposite $x$-direction than~$\vec n$. In our experiments we did not
  let the labels start at~$s$, but with a certain offset to~$s$, in
  order to avoid intersections with~$C$.

  \textbf{Octilinear Metro Maps.}  For \OctiLabels assume that the
  metro map is octilinear. We model the labels as horizontal and
  diagonal rectangles.  Let~$l$ be the line segment of~$C$ on
  which~$s$ lies and let~$R$ be an axis-aligned rectangle that is the
  bounding box of the name of~$s$. Further, let~$c$ be a circle
  around~$s$ with a pre-defined radius. We place the labels such that
  they touch the border of~$c$, but they do not intersect the interior
  of~$c$. Hence, the labels have a pre-defined offset to~$s$.

\newcommand{\octiLabelsScale}{0.99}
\begin{figure*}[t]
  \centering \subfigure[]{
    \includegraphics[page=3,scale=\octiLabelsScale]{candidates}
    \label{fig:octilinear_candidates:horizontal}
  }
  \subfigure[]{
    \includegraphics[page=5,scale=\octiLabelsScale]{candidates}
    \label{fig:octilinear_candidates:diagonal}
  }
  \subfigure[]{
    \includegraphics[page=4,scale=\octiLabelsScale]{candidates}
    \label{fig:octilinear_candidates:vertical}
  }
  \subfigure[]{
    \includegraphics[page=6,scale=\octiLabelsScale]{candidates}
    \label{fig:octilinear_candidates:dcrossing}
  }
  \subfigure[]{
    \includegraphics[page=7,scale=\octiLabelsScale]{candidates}
    \label{fig:octilinear_candidates:hvcrossing}  
  }
  \subfigure[]{
    \includegraphics[page=8,scale=\octiLabelsScale]{candidates}
    \label{fig:octilinear_candidates:hdcrossing}  
  }

  \caption{Construction of octilinear candidates for a
    stop~$s$. \protect\subref{fig:octilinear_candidates:horizontal}~$s$
    lies on a horizontal segment.
    \protect\subref{fig:octilinear_candidates:diagonal}~$s$ lies on a
    diagonal
    segment. \protect\subref{fig:octilinear_candidates:vertical}~$s$
    lies on a vertical
    segment. \protect\subref{fig:octilinear_candidates:dcrossing}~$s$
    lies on a crossing of two diagonal segments~$l$ and
    $l'$. \protect\subref{fig:octilinear_candidates:hvcrossing}~$s$
    lies on a crossing of a vertical and horizontal segment. \protect\subref{fig:octilinear_candidates:hdcrossing}~$s$
    lies on a crossing of a vertical and diagonal segment.}
  \label{fig:octilinear_candidates} 
\end{figure*}

If~$l$ is horizontal, we place five copies~$\ell_1,\dots,\ell_5$
of~$R$ above~$l$ as follows;~see
Fig.~\ref{fig:octilinear_candidates:horizontal}.  We place~$\ell_1$,
$\ell_2$ and $\ell_3$ such that the left-bottom corner of~$\ell_1$,
the midpoint of $\ell_2$'s bottom edge and the right-bottom corner
of~$\ell_3$ coincides with the topmost point of~$c$. We rotate~$\ell_4$ by~$45^\circ$ counterclockwise and place it at~$c$ such that the midpoint of its left side touches~$c$, i.e., that midpoint lies on a diagonal through~$s$.   Finally, $\ell_5$ is obtained by mirroring
$\ell_4$ at the vertical line through~$s$. Mirroring~$\ell_1,\dots,\ell_5$ at the horizontal
line through~$s$, we obtain the rectangles~$\ell'_1,\dots,\ell'_5$,
respectively. We then set~$\cands_s=\{\ell_i,\ell'_i\mid1\leq i \leq
5\}$.

If~$l$ is diagonal, we create the candidates in the same manner as in
the case that~$l$ is horizontal; see
Fig.~\ref{fig:octilinear_candidates:diagonal}. However, we only create the candidates~$\ell_2$ and $\ell'_2$ and the candidates that are horizontally aligned.

If~$l$ is vertical, we place three copies~$\ell_1,\dots,\ell_3$ of~$R$
to the right of~$l$ as follows; see
Fig.~\ref{fig:octilinear_candidates:vertical}. We rotate~$\ell_1$ by
$45^\circ$ counterclockwise and $\ell_2$ by $45^\circ$ clockwise. We
place~$\ell_1$, $\ell_2$, $\ell_3$ at~$c$ such that the midpoints
of their left edges touch~$c$. Mirroring~$\ell_1$, $\ell_2$,
$\ell_3$ at the vertical line through~$s$, defines the
rectangles~$\ell'_1$, $\ell'_2$ and $\ell'_3$. We
set~$\cands_s=\{\ell_i,\ell'_i\mid1\leq i \leq 3\}$.

In case that~$s$ is a crossing of two metro lines, we create the
candidates differently. If~$s$ is the crossing of two diagonals, we
create the candidates as shown in
Fig.~\ref{fig:octilinear_candidates:dcrossing}.  If~$s$ is the
crossing of a horizontal and a vertical segment, we create the labels
as shown in Fig.~\ref{fig:octilinear_candidates:hvcrossing}. If $s$ is
the crossing of a diagonal and horizontal segment, we create the
labels as shown in Fig.~\ref{fig:octilinear_candidates:hdcrossing}. We
analogously create the labels, if $s$ is the crossing of a vertical
and a diagonal segment.

\textbf{Remark:} If a stop lies on multiple metro lines, then we
can apply similar constructions, where the labels are placed on the
angle bisectors of the crossing lines.

\newcommand{\theoremNP}{
\MML is NP-hard, if the labels are based on \OctiLabels or \CurvedLabels, even if the  map has only one metro line.
}
\newcommand{\theoremNPComplete}{
\MML is NP-complete, if the labels are based on~\OctiLabels or~\CurvedLabels, even if the metro map has only one metro line.
}

\section{Computational Complexity}
\label{sec:comp-compl}

We first study the computational complexity of \MML assuming that the
labels are either based on~\OctiLabels or \CurvedLabels. In particular
we show that the problem is NP-hard, if the metro map consists of only
one line. The proof uses a reduction from the NP-complete problem
\emph{monotone planar 3SAT}\cite{l-pftu-82}. Based on the given style,
we create for the set~$\mathcal C$ of 3SAT clauses a metro
map~$(\lines,\stops)$ such that~$(\lines,\stops)$ has a labeling if
and only if~$\mathcal C$ is satisfiable. The proof can be easily
adapted to other labeling styles. Note that the complexity of labeling
points using a finite set of axis-aligned rectangular label candidates
is a well-studied NP-complete problem, e.g.,
see~\cite{Fowler1981133,fw-ppalm-91}. However, since we do not
necessarily use axis-aligned rectangles as labels and since for the
considered labeling styles the labels are placed along metro lines, it
is not obvious how to reduce a point-feature labeling instance on an
instance of \MML. In order to show the NP-hardness, we prove that it
is NP-complete to decide whether a metro map $(\lines,\stops)$ has a
labeling based on the given labeling style.

\begin{theorem}
  \theoremNPComplete
\end{theorem}



\begin{figure}[t]
\centering
\includegraphics{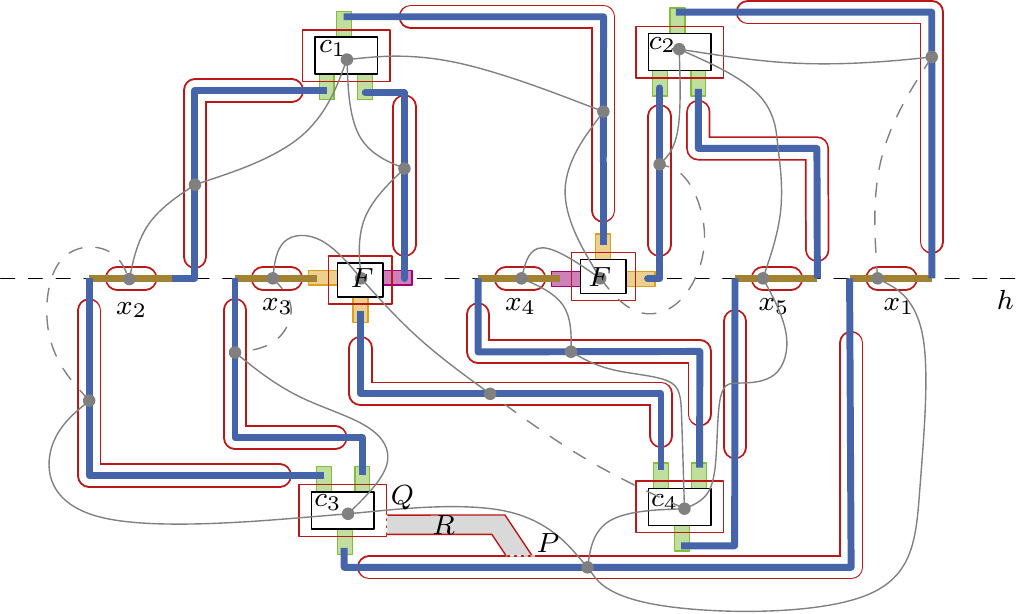}
\caption{Illustration of NP-completeness proof. (a)~3-SAT formula
  $\varphi$ with clauses $c_1=x_4 \vee x_1 \vee x_5$, $c_2=x_2 \vee
  x_4 \vee x_3$, $c_3=\bar x_2 \vee \bar x_1 \vee \bar x_3$ and
  $c_4=\bar x_3 \vee \bar x_5 \vee \bar x_4$ represented as a metro
  map. Truth assignment is $x_1=\mathit{true}$, $x_2=\mathit{true}$,
  $x_3=\mathit{false}$, $x_4=\mathit{false}$ and
  $x_5=\mathit{false}$. The gray graph~$G$ represents the adjacencies
  of the used gadgets. The solid lines represent a spanning tree
  of~$G$ that can be used to merge the polygons to one simple
  polygon. Exemplarily the polygons $P$ and $Q$ are merged into one
  polygon $R$. The single components are illustrated in Fig.~\ref{fig:gadgets}.}
 \label{fig:np-map}
\end{figure}

\begin{figure}[t]
\centering
\includegraphics[page=2]{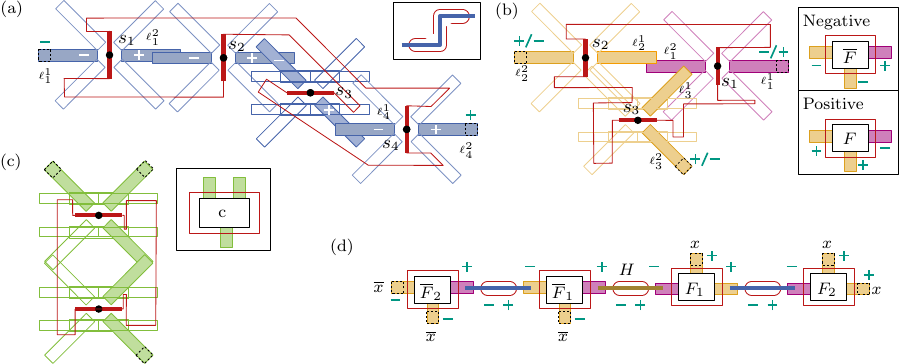}
\caption{Illustration of the gadgets. Selectable labels are filled, while all other labels are not filled. Ports are marked with a dashed square. (a) Chain gadget of length 4. (b) Fork gadget. (c) Clause gadget. (d) Variable gadget with three negative and three positive ports.}
\label{fig:gadgets}
\end{figure}

\begin{proof}
  For the illustrations we use~\OctiLabels, but the same
  constructions can be done based on~\CurvedLabels; see end of proof.
  We first show that the problem deciding whether $(\lines,\stops)$
  has a labeling lies in~NP. We first create for each stop~$s\in
  \stops$ its candidates~$\cands_s$ based on the given labeling
  style. Recall that each candidate has constant size. We then guess for
  each stop~$s\in \mathcal S$ the label~$\ell_s$ that belongs to the
  desired labeling~$\mathcal L$. Obviously, we can decide in
  polynomial time whether~$\{\ell_s\mid s\in \stops\}$ is a labeling
  of~$(\lines,\stops)$ performing basically intersection tests.

  We now perform a reduction from the NP-complete \textsc{Planar
    Monotone 3-Sat} problem~\cite{l-pftu-82}. Let $\varphi$ be a
  Boolean formula in conjunctive normal form such that it consists
  of~$n$ variables and $m$ clauses and, furthermore, each clause
  contains at most three literals. The formula $\varphi$ induces the
  graph~$G_\varphi$ as follows. $G_\varphi$ contains for each variable
  a vertex and it contains for each clause a vertex.  Two vertices~$u$
  and $v$ are connected by an edge $\{u,v\}\in E$ if and only if $u$
  represents a variable~$x$ and $v$ represents a clause~$c$, such that
  $x$ is contained in $c$. We call a clause of~$\varphi$
  \emph{positive} (\emph{negative}) if it contains only positive
  (negative) literals.

  The formula
  $\varphi$ is an instance of \textsc{Planar Monotone 3-Sat} if it
  satisfies the following requirements.
  \begin{compactenum}
  \item $\varphi$ is monotone, i.e., each clause is either positive or
    negative, and
  \item the graph $G_\varphi$ is planar and has a rectilinear plane
    embedding such that
      \begin{compactenum}
      \item the vertices representing variables are placed on a
        horizontal line~$h$,
        \item the vertices representing negative clauses are placed
          below~$h$,
        \item the vertices representing positive clauses are placed
          above~$h$,
        \item the edges are drawn on their respective side of~$h$. 
      \end{compactenum}
  \end{compactenum}
  \textsc{Planar Monotone 3-Sat} then asks whether~$\varphi$ is satisfiable.

  Using only stops lying on single horizontal and
  vertical segments, we construct a metro map~$(\lines,\stops)$ that
  mimics the embedding of~$G_\varphi$. In particular~$\lines$ will
  only consist of one metro line~$C$ that connects all stops such that
  the stops and their candidates simulate the variables and clauses
  of~$\varphi$. We will prove that~$(\lines,\stops)$ has a labeling if
  and only if $\varphi$ is satisfiable. We refer to
  Fig.~\ref{fig:np-map} for a sketch of the construction. We first
  define gadgets simulating variables, clauses and connecting
  structures. Each gadget consists of a set of stops that lie on the
  border of a simple polygon~$P$. Later on we use this polygon~$P$ to
  prescribe the shape of the metro line~$C$.

  \textbf{Chain.} The \emph{chain gadget} represents and transmits
  truth values from variables to clauses mimicking the embeddings of
  the edges in~$G_\varphi$.  A chain consists of an even number of
  stops~$s_1,\dots,s_k$ that lie on vertical and horizontal segments;
  see Fig.~\ref{fig:gadgets}(a). Hence, with respect to the given
  labeling style each stop~$s_i$ has a predefined set of
  candidates~$\cands_{s_i}$. For each stop~$s_i$ there are two
  specially marked candidates $\ell^1_i$ and $\ell^2_i$ that lie on
  opposite sides of $s_i$'s segment; for an example see the filled blue
  labels in Fig~\ref{fig:gadgets}(a). We say that those labels are
  \emph{selectable}, because we define the gadget such that those
  labels are the only labels that can be selected for a labeling. To
  that end, we lay out the metro line such that it does not intersect
  any selectable label, but all labels that are not selectable. The
  stops are placed such that the following conditions are satisfied.
  \begin{compactenum}[(1)]
    \item The label~$\ell^2_i$ intersects the label~$\ell^1_{i+1}$ for $1\leq i <k$. 
    \item Except the intersections mentioned in~(1), there is no
      intersection between selectable labels of different stops.
    \item The segments of the stops are connected by polylines s.t.\
      the result is a simple polygon~$P$ intersecting all labels
      except the selectable labels.
  \end{compactenum}
  The labels~$\ell_1^1$ and $\ell_k^2$ do not intersect any selectable
  labels; we call them the \emph{ports} of the chain. Later on, we use
  the ports to \emph{connect} other gadgets with the chain, i.e., we
  arrange the gadgets such that two of their ports intersect, but no other selectable label.  Further, we
  assign a polarization to each selectable label. The
  labels~$\ell^1_1,\dots,\ell^1_k$ are~\emph{negative} and the
  labels~$\ell^2_1,\dots,\ell^2_k$ are \emph{positive}.
   
  Consider a labeling~$\mathcal L$ of a chain assuming that~$P$ is
  interpreted as a metro line; we can cut~$P$ at some point in order
  to obtain an open curve. By construction of~$P$ only selectable
  labels are contained in~$\mathcal L$. In particular we observe that
  if the negative port~$\ell^1_1$ is not contained in~$\mathcal L$,
  then the positive labels~$\ell^2_1,\dots,\ell^2_k$ belong
  to~$\mathcal L$. Analogously, if the positive port~$\ell^2_k$ is not
  contained in~$\mathcal L$, then the negative
  labels~$\ell^1_1,\dots,\ell^1_k$ belong to~$\mathcal L$. We use this
  behavior to represent and transmit truth values through the chain.

  \textbf{Fork.} The \emph{fork gadget} splits an incoming chain
  into two outgoing chains and transmits the truth value represented
  by the incoming chain into the two outgoing chains. A fork consists
  of three stops~$s_1$, $s_2$ and $s_3$ such that~$s_1$ and $s_2$ are
  placed on vertical segments and~$s_3$ is placed on a horizontal
  segment; see Fig.~\ref{fig:gadgets}(b). Analogously to the chain,
  each stop~$s_i$ with ($1\leq i \leq 3$) has two \emph{selectable}
  labels~$\ell^1_i$ and $\ell^2_i$. We arrange the stops such that the
  following conditions are satisfied.
  \begin{compactenum}
  \item The labels~$\ell^1_2$ and~$\ell^1_3$
    intersect~$\ell^2_1$. Apart from those two intersections no
    selectable label intersects any other selectable label. 
  \item The segments of the stops are connected by polylines s.t.\ the
    result is a simple polygon~$P$ intersecting all labels except the
    selectable labels.
 \end{compactenum}
 The label~$\ell_1^1$ is the \emph{incoming port} and the labels
 $\ell^2_2$ and $\ell_3^2$ are the \emph{outgoing ports} of the
 fork. We distinguish two types of forks by assigning different
 polarizations to the selectable labels. In the \emph{negative}
 (\emph{positive}) fork, the labels~$\ell^1_1$, $\ell^1_2$ and
 $\ell^1_3$ are \emph{positive} (\emph{negative}) and the labels
 are~$\ell^2_1$, $\ell^2_2$ and $\ell^2_3$ are \emph{negative}
 (\emph{positive}). Hence, the incoming port is positive
 (negative) and the outgoings ports are negative
 (positive).

 Consider a labeling~$\mathcal L$ of a fork assuming that $P$ is
 interpreted as a metro line. By construction of~$P$ only selectable
 labels belong to~$\mathcal L$. Further, if the incoming
 port~$\ell^1_1$ does not belong to~$\mathcal L$, then the 
 outgoing ports~$\ell^2_1$ and $\ell^2_1$ belong to~$\mathcal
 L$. Finally, if one outgoing port does not belong to~$\mathcal L$,
 then the incoming port belongs to~$\mathcal L$.
 
 \textbf{Clause.} The \emph{clause gadget} represents a clause~$c$
 of the given instance. It forms a chain of length 2 with the addition
 that it has three ports instead of two ports; see
 Fig~\ref{fig:gadgets}(c). To that end one of both stops has three
 selectable labels; one intersecting a selectable label of the other
 stop, and two lying on the opposite side of the stop's segment
 without intersecting any selectable label of the other stop.  The
 gadget is placed at the position where the vertex of~$c$ is located
 in the drawing of~$G_\varphi$; see Fig.~\ref{fig:np-map}.  We observe
 that a labeling~$\mathcal L$ of a clause gadget always contains at
 least one port. Further, we do not assign any polarization to its
 selectable labels.

 \textbf{Variable.} The variable gadget represents a single
 variable~$x$. It forms a composition of chains and forks that are
 connected by their ports; see Fig~\ref{fig:gadgets}(d). More
 precisely, let $s$ be the number of clauses in which the negative
 literal~$\bar x$ occurs and let~$t$ be the number of clauses in which
 the positive literal~$x$ occurs.  Along the horizontal line~$h$ on
 which the vertex of~$x$ is placed in the drawing of~$G_\varphi$, we
 place a horizontal chain~$H$.  Further, we place a sequence of
 negative forks~$\overline F_1,\dots,\overline F_{s-1}$ to the left
 of~$H$ and a sequence of positive forks $F_1,\dots,F_{t-1}$ to the
 right of~$H$. The negative incoming port of~$F_1$ is connected to the
 positive port of~$H$ by a chain. Two consecutive forks~$F_i$
 and~$F_{i+1}$ are connected by a chain~$H'$ such that~$H'$ connects
 a positive outgoing port of~$F_i$ with the negative incoming port
 of~$F_{i+1}$. Analogously, the positive incoming port of~$\overline
 F_1$ is connected to the negative port of~$H$ by a chain. Two
 consecutive forks~$\overline F_i$ and~$\overline F_{i+1}$ are
 connected by a chain~$H'$ such that~$H'$ connects a negative
 outgoing port of~$\overline F_i$ with the positive incoming port
 of~$\overline F_{i+1}$.

 We observe that the gadget has $s+t$ free ports. Further, we can
 arrange the forks such that the free ports of~$F_{1},\dots,F_{t-1}$
 lie above~$h$ and the free ports of~$\overline F_{1},\dots,\overline
 F_{s-1}$ lie below~$h$. 

 Consider a labeling~$\mathcal L$ of a variable. By construction of
 the forks and chains, if one positive free port is not contained
 in~$\mathcal L$, then all negative free ports must be contained
 in~$\mathcal L$. Analogously, if one negative free port
 is not contained in~$\mathcal L$, then all positive free ports must
 be contained in~$\mathcal L$.

 Using additional chains we connect the positive free ports with the
 positive clauses and the negative ports with the negative clauses
 correspondingly; see Fig.~\ref{fig:np-map}. More precisely, assume
 that the variable~$x$ is contained in the positive clause~$c$;
 negative clauses can be handled analogously. With respect to the
 drawing of~$G_\varphi$, a positive free port of~$x$'s gadget is
 connected with the negative port of a chain whose positive port is
 connected with a free port of $c$'s gadget.  Note that we can easily
 choose the simple polygons enclosing the gadgets such that they do
 not intersect by defining them such that they surround the gadgets
 tightly.

 \textbf{One Metro Line.} We construct the polygons enclosing 
 the single gadgets such that they do not intersect each other.  We
 now sketch how the polygons can be merged to a single simple
 polygon~$P$. Cutting this polygon at some point we obtain a polyline
 prescribing the desired metro line.
 
 We construct a graph $H=(V,E)$ as follows. The polygons of the gadgets
 are the vertices of the graph and an edge $(P,Q)$ is contained in~$E$
 if and only if the corresponding gadgets of the polygons~$P$ and~$Q$
 are connected by their ports; see
 Fig.~\ref{fig:np-map}. Since~$G_\varphi$ is planar and the gadgets
 mimic the embedding of~$G_\varphi$, it is not hard to see that~$H$ is
 also planar. We construct a spanning tree~$T$ of~$H$. If an edge
 $(P,Q)$ of~$H$ is also contained in~$T$, we \emph{merge} $P$ and $Q$
 obtaining a new simple polygon~$R$; see for an
 example~Fig.~\ref{fig:np-map}. To that end we cut~$P$ and $Q$ in
 polylines and connect the four end points by two new polylines such
 that the result is a simple polygon. We in particular ensure that the
 new polygon does not intersect any other polygon and that~$R$
 intersects the same labels as~$P$ and $Q$ together. In~$T$ we
 correspondingly contract the edge. Note that by contracting edges,
 $T$ remains a tree. We repeat that procedure until $T$ consists of a
 single vertex, i.e., only one simple polygon is left.

 \textbf{Soundness.} It is not hard to see that our construction is
 polynomial in the size of the given 3SAT formula~$\varphi$.

 Assume that~$\varphi$ is satisfiable. We show how to construct a
 labeling~$\mathcal L$ of the constructed metro map. For each
 variable~$x$ that is \emph{true} (\emph{false}) in the given truth
 assignment, we put all negative (positive) labels of the
 corresponding variable gadget and its connected chains into~$\mathcal
 L$. By construction those labels do not intersect. It remains to
 select labels for the clause gadgets. Consider a positive clause $c$;
 negative clauses can be handled analogously. Since~$\varphi$ is
 satisfiable, $c$ contains a variable~$x$ that is \emph{true} in the
 given truth assignment of~$\varphi$. The set $\mathcal L$ contains
 only negative labels of the chain connecting the gadget of~$x$ with
 the gadget of~$c$, but no positive labels of that chain. Hence, we
 can add the port of~$c$'s gadget that is connected to that chain
 without creating intersections. For the second stop of the clause we
 put that selectable label into~$\mathcal L$ that is not a port. We
 can apply this procedure to all positive and negative clauses without
 creating intersections, which yields the labeling~$\mathcal L$ of the
 constructed metro map.

 Finally, assume that we are given a labeling~$\mathcal L$ of the
 constructed metro map. Consider the clause gadget of a positive
 clause~$c$; negative clauses can be handled analogously. By
 construction $\mathcal L$ contains at least one port~$\ell$ of that
 gadget. This port is connected to a chain, which is then connected to
 a gadget of a variable~$x$. We set that variable~$x$ \emph{true}. We
 apply this procedure to all clauses; for negative clauses we set the
 corresponding variable to~\emph{false}. Since~$\ell$ is contained
 in~$\mathcal L$, only negative labels of that chain can be contained
 in~$\mathcal L$, but no positive labels. Hence, the  positive ports of the variable gadget are also not contained
 in~$\mathcal L$. By the previous reasoning this implies that all
 negative ports of the gadget are contained in~$\mathcal
 L$. Consequently, by applying a similar procedure to negative
 clauses, it cannot happen that~$x$ is set to \emph{false}.
 Altogether, this implies a valid truth assignment of~$\varphi$.

 \textbf{Remarks:} Fig.~\ref{fig:gadgets-curved} illustrates the
 construction of the gadgets for \CurvedLabels. Note that only the
 fork gadget, the clause gadget and the chain gadget rely on the
 concrete labeling style. Further, using \CurvedLabels, a stop $s$
 lying on a vertical segment~$l_v$ has exactly two different
 distinguish labels; one that lies to the left of~$l_v$ and one that
 lies to the right of $l_v$.\qed

\begin{figure}[t]
\centering
\includegraphics[page=3]{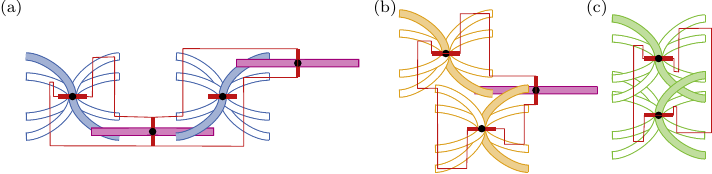}
\caption{Illustration of the gadgets based on \OctiLabels. Selectable labels are filled, while all other labels are not filled. (a) Chain gadget of length 4. (b) Fork gadget. (c) Clause gadget.}
\label{fig:gadgets-curved}
\end{figure}

\end{proof}

\section{Labeling Algorithm for a Single Metro Line}\label{sec:single-line}

\begin{figure}[t]
 \centering
 \includegraphics[page=1]{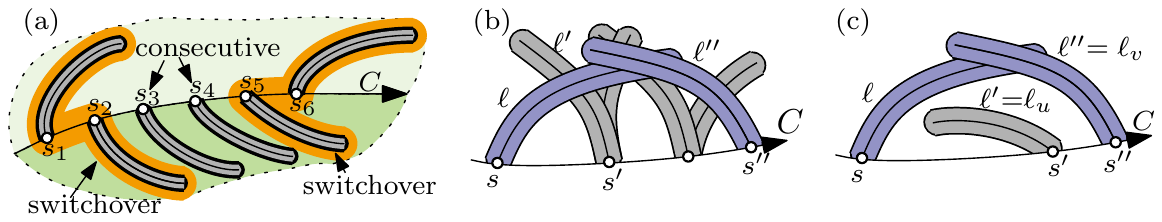}
 \caption{\protect(a) Consecutive stops and switchovers. (b)~The candidates satisfy the \PROPERTY. (c)~The candidates do not satisfy the \PROPERTY.}
 \label{fig:model_small}
\end{figure} 

We now study the case that the given
instance~$I=(\lines,\stops,\cands,w)$ consists only of one metro
line~$C$. Based on cartographic criteria we introduce three additional assumptions on~$I$, which allows us to efficiently solve~\MML.

For each stop~$s \in S$, we assume that each candidate~$\ell\in
\cands_s$ is assigned to one side of~$C$; either $\ell$ is a
\emph{left candidate} assigned to the left side of~$C$, or $\ell$ is a
\emph{right candidate} assigned to the right side of~$C$. For
appropriately defined candidate sets those assignments correspond with
the geometric positions of the candidates, i.e., left (right)
candidates lie on the left (right) hand side of~$C$.
\begin{assumption}[Separated Labels]\label{assumption:separation}
Candidates that are assigned to different sides of~$C$ do not intersect.
\end{assumption}
This assumption is normally not a real restriction, because for
appropriately defined candidate sets and realistic metro lines, the
line~$C$ separates both types of candidates geometrically. We further
require what we call the \emph{transitivity property}.
\begin{assumption}[Transitivity Property]\label{assumption:transitivity-property}
For any three
stops~$s,s',s''\in \stops$ with~$s < s' < s''$ and any three
candidates $\ell \in\cands_s$, $\ell' \in\cands_{s'}$ and $\ell''
\in\cands_{s''}$ assigned to the same side of~$C$, it holds that if
neither $\ell$ and $\ell'$ intersect nor $\ell'$ and $\ell''$
intersect then also $\ell$ and $\ell''$ do not intersect; see also Fig.~\ref{fig:model_small}(b)--(c).
\end{assumption}

In our experiments we established
Assumption~\ref{assumption:separation} and
Assumption~\ref{assumption:transitivity-property} by removing
candidates greedily. In Section~\ref{sec:evaluation} we show that for
real-world metro maps and the considered candidate sets we remove only few labels,
which indicates that those assumptions have only a little influence on
the labelings.

Two stops~$s,s'\in \stops$ with~$s < s'$ are
\emph{consecutive} if there is no other stop~$s''\in \stops$
with~$s<s''<s'$; see Fig.~\ref{fig:model_small}(a).
For two consecutive stops~$s_1,s_2 \in \stops$ we say that each two
candidates~$\ell_1\in \cands_{s_1}$ and~$\ell_2\in \cands_{s_2}$
are~\emph{consecutive} and denote the set that contains each pair of consecutive labels
in~$\mathcal L\subseteq \cands$ by~$\consLabels_{\mathcal L}\subseteq
\mathcal L\times\mathcal L$. Further, two consecutive labels~$\ell_1,\ell_2\in
\cands_C$ form a \emph{switchover}~$(\ell_1,\ell_2)$ if they are assigned to
opposite sides of~$C$, where $(\ell_1,\ell_2)$ denotes an ordered
set indicating the order of the stops of~$\ell_1$ and $\ell_2$. Two
switchovers of~$C$ are \emph{consecutive} in~$\mathcal L\subseteq
\cands$ if there is no switchover in~$\mathcal L$ in between of both.
We define the set of all switchovers in~$\cands$ by~$\switches$ and
the set of consecutive switchovers in $\mathcal L\subseteq \cands$ by
$\consSwitches_{\mathcal L} \subseteq \switches \times \switches$.

Based on cartographic criteria extracted from Imhof's ``general
principles and requirements'' for map labeling~\cite{imhof}, we
require a cost function $w\colon 2^\cands\to \mathbb{R}^+$ of the
following form; see also Section~\ref{sec:weighting-function}
for a detailed motivation of~$w$.
\begin{assumption}[Linear Cost Function]\label{assumption:weighting-function}
For any~$\mathcal L\subseteq K$ we require
\[
w(\mathcal L)=\sum_{\ell \in \mathcal L}w_1(\ell) +
\smashoperator[r]{\sum_{(\ell_1,\ell_2)\in P_{\mathcal L}}}w_2(\ell_1,\ell_2) +
\smashoperator[r]{\sum_{(\sigma_1,\sigma_2)\in \consSwitches_{\mathcal L}}} w_3(\sigma_1,\sigma_2),
\]
where~$w_1\colon \mathcal L \to \mathbb{R}$ rates a single label,
$w_2\colon P_{\mathcal L} \to \mathbb{R}$ rates two consecutive labels
and $w_3\colon \consSwitches_{\mathcal L} \to \mathbb{R}$ rates two
consecutive switchovers. 
\end{assumption}
In particular, we define~$w$ such that it penalizes the following structures to sustain readbility. 
\begin{inparaenum}[(1)]
 \item Steep or highly curved labels.
 \item Consecutive labels that lie on different sides of~$C$, or that are shaped differently.
 \item Consecutive switchovers that are placed close to each other.
\end{inparaenum}

If~$I=(\{C\},\stops,\cands,w)$ satisfies Assumption~\ref{assumption:separation}--\ref{assumption:weighting-function}, we call \MML also~\RMLL. We now introduce an
algorithm that solves this problem in~$O(n^2 k^4)$ time,
where~$n=|\stops|$ and~$k=\max\{|\cands_s|\mid s\in\stops\}$. Note
that~$k$ is typically constant. 
We assume w.l.o.g.\ that~$\cands$ contains only candidates that do not intersect~$C$. 

\begin{figure}[t]
 \centering
 \includegraphics[page=1]{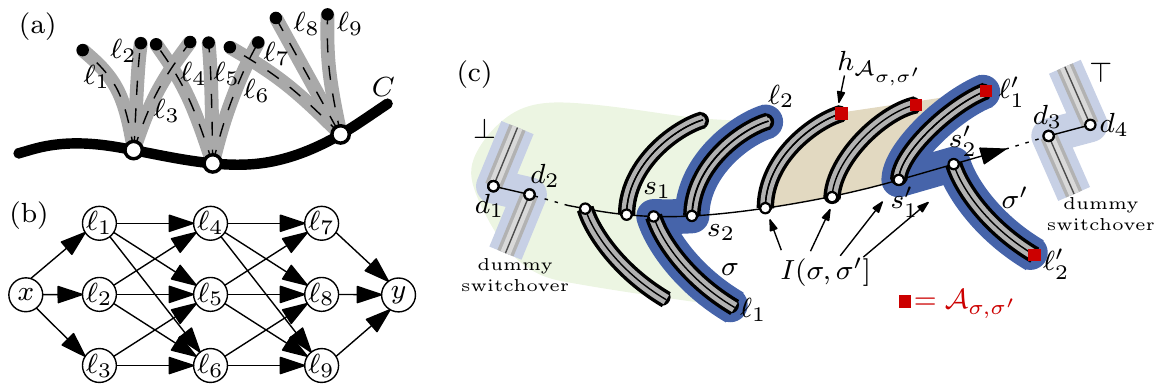}
 \caption{Illustrations for labeling a single metro line. (a) A
   one-sided instance and (b) the acyclic directed graph $G$ based
   on its labels. (c) A two-sided instance with a labeling. The
   switchovers~$\sigma'$ and~$\sigma$ separate the labeling into a
   two-sided and a one-sided instance.}
\label{fig:dyn-prog}
\end{figure}

\paragraph{Labels on One Side.} We first assume that all candidate
labels in $\cands$ are assigned either to the left or to the right
side of~$C$; without loss of generality to the left side of~$C$. For two
stops~$s,s' \in \stops$ we denote the instance restricted to the stops
$\{s,s'\}\cup \{s''\in \stops\mid s<s''<s'\}$ by~$I[s,s']$.  We denote
the first stop of~$C$ by~$\firstStop$ and the last stop
by~$\lastStop$. The transitivity property directly yields the next
lemma.  \newcommand{\contentLemmaOneSidedTransProp}{ Let~$s$, $s'$ and
  $s''$ be stops with~$s<s'<s''$,~$\mathcal L$ be a labeling
  of~$I[\firstStop,s']$, $\ell\in \mathcal L\cap \cands_{s}$ and
  $\ell'\in \mathcal L\cap\cands_{s'}$. Any~$\ell''\in \cands_{s''}$
  intersecting~$\ell$ also intersects~$\ell'$.  }
\begin{lemma}\label{lem:one-sided-trans-prop}
\contentLemmaOneSidedTransProp
\end{lemma}

\begin{proof}
Recall for the proof that we assume that~$I$ satisfies Assumptions~\ref{assumption:separation}--\ref{assumption:weighting-function}.

Assume for the sake of contradiction that there is a candidate~$\ell''\in \cands_{s''}$ such that $\ell''$ intersects~$\ell$ but not~$\ell'$; see Fig.~\ref{fig:model_small}(c). Since~$\mathcal L$ is a labeling, the labels~$\ell$ and~$\ell'$ do not intersect. Hence, neither~$\ell$ and $\ell'$ nor $\ell'$ and~$\ell''$ intersect. Since all three labels are assigned to the same side of~$C$, the transitivity property holds, which directly contradicts that~$\ell$ and~$\ell'$ do not intersect.\qed    
\end{proof}

Hence, the lemma states that~$\ell'$ separates~$\mathcal L$ from the candidates of the stops succeeding~$s'$. We use this observation as follows.
Based on $\cands$ we define a directed acyclic graph $G=(V,E)$;
see Fig.~\ref{fig:dyn-prog}(a)--(b). This graph
contains a vertex~$u$ for each candidate~$\ell\in \cands$ and the two
vertices~$x$ and $y$. We call~$x$ the \emph{source} and $y$ the
\emph{target} of~$G$.
 Let~$\ell_u$ denote the candidate that belongs to the vertex
$u\in V\setminus\{x,y\}$. For each pair $u, v \in V\setminus \{x,y\}$
the graph contains the edge~$(u,v)$ if and only if the stop of
$\ell_u$ lies directly before the stop of~$\ell_v$ and, furthermore,
$\ell_u$ and $\ell_v$ do not intersect. Further, for each vertex~$u$
of any candidate of~$\firstStop$ the graph contains the edge $(x,u)$,
and for each vertex~$u$ of any candidate of~$\lastStop$ the graph
contains the edge~$(u,y)$. For an edge $(u,v)\in E$ we define its
cost~$w_e$ as follows.  For~$u\neq x$ and~$v\neq y$ we
set~$w_e=w_1(\ell_v)+w_2(\ell_u,\ell_v)$. For $x=u$ we
set~$w_e=w_1(\ell_v)$ and for $v=y$ we set $w_e=0$. 

An \emph{$x$-$y$ path}~$P\subseteq E$ in~$G$ is a path in~$G$
that starts at~$x$ and ends at~$y$. Its costs are~$w(P)=\sum_{e\in P}w_e$.
The $x$-$y$ path with minimum costs among all $x$-$y$~paths is the \emph{shortest}~$x$-$y$ path.

\newcommand{\contentLemmaOneSided}{
   For any $x$-$y$ path~$P$ in~$G$ there is a labeling~$\mathcal L$ of~$I$ with~$w(P)=w(\mathcal L)$ and for 
   any labeling~$\mathcal L$ of $I$ there is an $x$-$y$ path~$P$ in~$G$ with~$w(P)=w(\mathcal L)$.  
}
\begin{lemma}\label{lemma-one-sided}
  \contentLemmaOneSided
\end{lemma}

\begin{proof}
  Recall for the proof that we assume that~$I$ satisfies
  Assumptions~\ref{assumption:separation}--\ref{assumption:weighting-function}.

  Let $P=(V_P,E_P)$ be an $x$-$y$ path in $G$ and let $\mathcal
  L=\{\ell_v\in \cands \mid v\in V_P \}$, where $V_P$ denotes the
  vertices of $P$ and $E_P$ the edges of~$P$. We show that~$\mathcal
  L$ is a labeling of~$C$ with~$w(\mathcal L)=w(P)$. Obviously, for
  each stop~$s\in \stops$ the set~$\mathcal L$ contains exactly one
  candidate $\ell \in \cands_s$.  By construction for each edge $(u,v)
  \in E_P$ the labels~$\ell_u$ and $\ell_v$ do not intersect. Hence,
  by Lemma~\ref{lem:one-sided-trans-prop} the label~$\ell_v$ cannot
  intersect any label~$\ell\in\mathcal L$ of any stop that occurs
  before the stop of~$\ell_u$. Hence, the set $\mathcal L$ is a
  labeling.  Let~$\ell_1,\dots,\ell_n$ be the labels in~$\mathcal L$
  in the order of their stops. It holds
  \begin{equation}\label{eqn:one-sided}
    w(P)=\sum_{\mathclap{e\in E_P}}w_e=w_1(\ell_1)+\sum^n_{i=2}(w_1(\ell_i)+w_2(\ell_{i-1},\ell_i))=w(\mathcal L)
  \end{equation}
  Now, let~$\mathcal L$ be an arbitrary labeling of~$C$. We show that
  there is an~$x$-$y$ path~$P$ with~$w(P)=w(\mathcal L)$. Let $s$ and
  $s'$ be two consecutive stops with~$s<s'$ and let~$\ell$ and $\ell'$
  be the corresponding labels in~$\mathcal L$. Since~$\ell$
  and~$\ell'$ do not intersect, the corresponding vertices~$u$ and
  $u'$ of $\ell$ and $\ell'$ are adjacent in~$G$. Hence, the labels in~$\mathcal L$
  induce a path~$P$ in $G$. Let~$\ell_1,\dots,\ell_n$ be the labels
  in~$\mathcal L$ in the order of their stops. Using
  Equation~\ref{eqn:one-sided} we obtain~$w(P)=w(\mathcal
  L)$.\qed
\end{proof}

\newcommand{\SP}{\textsc{MinPath}\xspace}  The lemma in particular proves that a
shortest $x$-$y$~path~$P$ in~$G$ corresponds with an optimal labeling
of~$I$. Due to \cite[Chapter~24]{Cormen},~$P$ can be constructed in~$O(|V|+|E|)$
time using a dynamic programming approach, which we call \SP. In
particular \SP considers each edge only once. There are $O(n\cdot k)$
vertices in~$G$ and each vertex has at most~$k$ incoming edges,
which implies that there are $O(n\cdot k^2)$ edges. Since \SP
considers each edge only once, we compute the edges of~$G$ on
demand, which saves storage.

\begin{theorem}\label{thm:one-sided-case}
If $I$ is one-sided, \RMLL can be optimally solved in~$O(n k^2)$ time and~$O(nk)$ space.
\end{theorem}

\paragraph{Labels on Both Sides.} If candidates lie on
both sides of the metro line, we solve the problem utilizing the
algorithm for the one-sided case.

Consider a labeling~$\mathcal L$ of $I$ and let $\sigma$, $\sigma'$
be two switchovers in~$\mathcal L$ such that~$\sigma$ lies
before~$\sigma'$ and no other switchover lies in between both; see
Fig.~\ref{fig:dyn-prog}(c). Roughly spoken, $\sigma$ and~$\sigma'$ induce a two-sided
instance that lies before~$\sigma$ and a one-sided instance that lies
in between both switchovers $\sigma$
and~$\sigma'$. 
\newcommand{\contentLemmaTwoSidedTransProp}{
 Let~$s$, $s'_1$, $s'_2$ and $s''$ be stops with~$s<s'_1<s'_2<s''$; $s'_1$ and~$s'_2$ are consecutive. Let~$\mathcal L$ be a labeling of~$I[\firstStop,s'_2]$, $\ell\in \mathcal L\cap \cands_s$, $\ell'_1\in \mathcal L\cap\cands_{s'_1}$, $\ell'_2\in \mathcal L\cap \cands_{s'_2}$ s.t.~$(\ell'_1,\ell'_2)$ is a switchover.
Any $\ell''\in \cands_{s''}$ intersecting~$\ell$ intersects~$\ell'_1$ or $\ell'_2$.
}
\begin{lemma}\label{lem:two-sided-trans-prop}
\contentLemmaTwoSidedTransProp
\end{lemma}

\begin{proof}
  Recall for the proof that we assume that~$I$ satisfies
  Assumptions~\ref{assumption:separation}--\ref{assumption:weighting-function}.

  Assume for the sake of contradiction that there is a label~$\ell''
  \in \cands_{s''}$ such that~$\ell''$ intersects~$\ell$ without
  intersecting~$\ell'_1$ and $\ell'_2$. Since~$\ell$ and~$\ell''$
  intersect each other, due to Assumption~\ref{assumption:separation}
  both are assigned to the same side of~$C$; w.l.o.g., let~$\ell$ and
  $\ell''$ be assigned to the left hand side of~$C$. Further,
  w.l.o.g., let~$\ell'_1$ be a left candidate and~$\ell'_2$ a right
  candidate; analogous arguments hold for the opposite
  case. Since~$\mathcal L$ is a labeling, the labels~$\ell$ and
  $\ell'_1$ do not intersect. Hence, neither $\ell$ and $\ell'_1$ nor
  $\ell'_1$ and~$\ell''$ intersect. Since~$\ell$, $\ell'_1$ and
  $\ell''$ are assigned to the same side of~$C$, the transitivity
  property must hold. However, this contradicts that~$\ell$
  and~$\ell''$ intersect.\qed
\end{proof}

Hence, the lemma yields that for the
one-sided instance we can choose any labeling; as long as this
labeling does not intersect any label of~$\sigma$ or~$\sigma'$, it
composes with~$\sigma$, $\sigma'$ and the labeling of the two-sided instance
to one labeling for the instance up to~$\sigma'$. We use that observation as follows.

Let~$\sigma=(\ell_1,\ell_2)$ and $\sigma'=(\ell'_1,\ell'_2)$ be two
switchovers in~$\switches$. Let~$s_1$ and $s_2$ be the stops
of~$\ell_1$ and~$\ell_2$, and let~$s'_1$ and $s'_2$ be the stops
of~$\ell'_1$ and~$\ell'_2$, respectively; see
Fig.~\ref{fig:dyn-prog}(c). We assume that~$\sigma<\sigma'$,
i.e., $s_1<s'_1$. Let~$I(\sigma,\sigma']$ be the instance restricted
to the stops~$\{s\in \stops \mid s_2< s < s'_1\}\cup\{s'_1,s'_2\}$,
where $(\sigma,\sigma']$ indicates that the stops of
$\sigma'$ belong to that instance, while the stops of $\sigma$ do not.

The switchovers~$\sigma$ and $\sigma'$ are~\emph{compatible}
if~$\ell_2$ and $\ell'_1$ are assigned to the same side of~$C$, and there is a
labeling for~$I[s_1,s'_2]$ such that it contains $\ell_1$, $\ell_2$,
$\ell'_1$ and $\ell'_2$ and, furthermore, $\sigma$ and $\sigma'$ are the only
switchovers in that labeling.  Let~$\mathcal L$ be the optimal
labeling among those labelings.  We denote the labeling~$\mathcal
L\setminus\{\ell_1,\ell_2\}$ of~$I(\sigma,\sigma']$ by~$\mathcal
A_{\sigma,\sigma'}$. 
Utilizing Theorem~\ref{thm:one-sided-case}, we obtain~$\mathcal
A_{\sigma,\sigma'}$ in $O(n\cdot k^2)$ time.

For any labeling $\mathcal L$ of an instance $J$
let~$\firstLabel_{\mathcal L}\in \mathcal L$ be the label of the first
stop in~$J$ and let~$\lastLabel_{\mathcal L} \in \mathcal L$ be the
label of the last stop in~$J$; $\firstLabel_{\mathcal L}$ is the
\emph{head} and $\lastLabel_{\mathcal L}$ is the \emph{tail}
of~$\mathcal L$.  For technical reasons we extend~$\stops$ by the
\emph{dummy stops} $d_1$, $d_2$, $d_3$ and $d_4$ such that~$d_1 < d_2
< s < d_3 < d_4$ for any stop~$s\in \stops$. For $d_1$ and $d_2$ we
introduce the \emph{dummy switchover}~$\bot$ and for $d_3$ and $d_4$
the dummy switchover~$\top$. We define that~$\bot$ and $\top$ are
compatible to all switchovers in~$\switches$ and that
$\bot$ and $\top$ are compatible, if there is a one-sided labeling
for~$I$.  Conceptually, each dummy switchover consists of two labels
that are assigned to both sides of~$C$. Further, neither $\bot$
nor~$\top$ has any influence on the cost of a labeling. Hence,
w.l.o.g.\ we assume that they are contained in any labeling.

Similar to the one-sided case we define a directed acyclic graph
$G'=(V',E')$. This graph contains a vertex~$u$ for each
switchover~$\switches\cup\{\bot,\top\}$. Let~$\sigma_u$ denote the
switchover that belongs to the vertex $u\in V$. In particular let~$x$
denote the vertex of~$\bot$ and~$y$ denote the vertex of~$\top$.  For
each pair $u, v \in V$ the graph contains the edge~$(u,v)$ if and only
if~$\sigma_u$ and $\sigma_v$ are compatible and $\sigma_u<\sigma_v$.
The cost~$w_e$ of an edge~$e=(u,v)$ in $G'$ is $w_e=
w(\mathcal
A_{\sigma_u,\sigma_v})+w_3(\sigma_u,\sigma_v)+w_2(\ell^u_2,\firstLabel_{\mathcal
  A_{\sigma_u,\sigma_v}})$, where $\sigma_{u}=(\ell^u_1,\ell^u_2)$. In the special case that~$\sigma_u$ and
$\sigma_v$ share a stop, we set
$w_e=w_3(\sigma_u,\sigma_v)+w_2(\ell^v_1,\ell^v_2)+w_1(\ell^v_2)$,
where $\sigma_v=(\ell^v_1,\ell^v_2)$.

Let $P$ be an $x$-$y$ path in $G'$ and let~$e_1=(x=v_0,v_1),e_2=(v_1,v_2),\dots,e_{l}=(v_{l-1},v_{l}=y)$ be the edges of $P$. For a vertex~$v_i$ of~$P$ with $0\leq i \leq l$ we write~$\sigma_i$ instead of
$\sigma_{v_i}$. We denote the set $\bigcup_{i=1}^l A_{\sigma_{i-1},\sigma_i}$ by $\mathcal L_P$. 

\newcommand{\contentLemmaTwoSided}{
 a)~The graph $G'$ has an $x$-$y$ path if and only if~$I$ has a labeling.\\
 b)~Let~$P$ be a shortest $x$-$y$ path in $G'$, then $\mathcal L_P$ is an optimal labeling of~$I$.

}
\begin{lemma}\label{lem:two-sided}
 \contentLemmaTwoSided
\end{lemma}

\begin{proof}
Recall for the proof that we assume that~$I$ satisfies Assumption~\ref{assumption:separation}--\ref{assumption:weighting-function}.

  By construction of~$G'$, it directly follows that~$G'$ has
  an~$x$-$y$ path~$P$ if and only if $I$ has a labeling~$\mathcal L$.
  We first show that $\mathcal L_P$ is a labeling of~$I$ with
  $w(\mathcal L_P)=w(P)$. Afterwards we prove that for any labeling~$\mathcal L$ of
  $I$ it holds $w(\mathcal L)\geq w(\mathcal L_P)$.

  Let~$e_1=(v_0,v_1),e_2=(v_1,v_2),\dots,e_{l}=(v_{l-1},v_{l})$ be the
  edges of $P$ with~$v_0=x$ and $v_{l}=y$. For a vertex~$v_i$ with
  $0\leq i \leq l$ we write~$\sigma_i$ instead of $\sigma_{v_i}$. We
  show by induction over~$m\in \mathbb N$ with~$1\leq m \leq l$ that
  \[
   \mathcal L_m:=\bigcup^{m}_{i=1}\mathcal A_{\sigma_{i-1},\sigma_i}
  \]
  is a labeling of~$I(\sigma_0,\sigma_{m}]$ with~$w(\mathcal
  L_m)=\sum_{i=1}^{m}w_{e_i}$. Altogether this implies that $\mathcal L_P=\mathcal
  L_l$ is a labeling with $w(\mathcal L_P)=w(P)$.

  For~$m=1$ we have~$\mathcal L_1=\mathcal A_{\sigma_0,\sigma_1}$. By
  the construction of $G'$ the set $\mathcal A_{\sigma_0,\sigma_1}$ is
  a labeling. Since $\sigma_0=(\ell_0,\ell_1)$ is a dummy switchover
  we have $w_3(\sigma_0,\sigma_1)=0$ and
  $w_2(\ell_2,\firstLabel_{\mathcal A_{\sigma_0,\sigma_1}})=0$. Hence,
  it holds $w(\mathcal A_{\sigma_0,\sigma_1})=w_{e_1}$.
  
  Now, consider the set~$\mathcal L_m$ for $m>1$. We first argue
  that~$\mathcal L_m$ is a labeling of~$I(\sigma_0,\sigma_m]$. By
  induction the set~$\mathcal L_{m-1}\subseteq \mathcal L_m$ is a labeling
  of~$I(\sigma_0,\sigma_{m-1}]$. Further, by construction the set
  $\mathcal L_{m}\setminus \mathcal
  L_{m-1}=A_{\sigma_{m-1},\sigma_{m}}$ is a labeling of the
  instance~$I(\sigma_{m-1},\sigma_m]$. Since~$\sigma_{m-1}$ and
  $\sigma_m$ are compatible, no label of $\mathcal L_{m}\setminus
  \mathcal L_{m-1}$ intersects any label of~$\sigma_{m-1}$. Then by
  Lemma~\ref{lem:two-sided-trans-prop} no two labels in $\mathcal L_m$
  intersect each other.

  We now show that $w(\mathcal L_m)=\sum_{i=1}^{m}w_{e_i}$. By
  induction we have~$w(\mathcal
  L_{m-1})=\sum_{i=1}^{m-1}w_{e_i}$. Since $\mathcal L_m=\mathcal
  L_{m-1} \cup \mathcal A_{\sigma_{m-1},\sigma_{m}}$ it holds
 \begin{equation}
   w(\mathcal L_m)=w(\mathcal L_{m-1} \cup \mathcal A_{\sigma_{m-1},\sigma_{m}}).\label{equation:base}
 \end{equation}
 We distinguish two cases. First assume that~$\sigma_{m-1}$ and
 $\sigma_m$ do not have any stop in common. Let~$\sigma_{m-1}=(\ell^{m-1}_1,\ell^{m-1}_2)$, then we derive from
 Equation~(\ref{equation:base})
 \begin{align*}
   w(\mathcal L_m)=&
                   w(\mathcal L_{m-1} \cup \mathcal A_{\sigma_{m-1},\sigma_{m}})\\
                  =&w(\mathcal L_{m-1})+w(\mathcal A_{\sigma_{m-1},\sigma_{m}})+w_3(\sigma_{m-1},\sigma_m)+w_2(\ell^{m-1}_2,h_{\mathcal A_{\sigma_{m-1},\sigma_{m}}})\\
                  \stackrel{(I)}{=}& w(\mathcal L_{m-1})+w_{e_m}\stackrel{(II)}{=}\sum_{i=1}^mw_{e_i}.
 \end{align*}
 Equality~(I) holds due to the definition of $w_{e_m}$ and
 Equality~(II) is by induction true. Now assume that~$\sigma_{m-1}$
 and $\sigma_m$ have a stop in
 common. Let~$\sigma_m=(\ell^m_1,\ell^m_2)$, then we derive from
 Equation~(\ref{equation:base})

 \begin{align*}
   w(\mathcal L_m)=&
   w(\mathcal L_{m-1} \cup \mathcal A_{\sigma_{m-1},\sigma_{m}})\\
   =&w(\mathcal L_{m-1})+w_3(\sigma_{m-1},\sigma_m)+w_2(\ell^m_1,\ell^m_2)+w_1(\ell^m_2)\\
   \stackrel{(III)}{=}& w(\mathcal
   L_{m-1})+w_{e_m}\stackrel{(IV)}{=}\sum_{i=1}^mw_{e_i}.
 \end{align*}
 Equality~(III) holds due to the definition of $w_{e_m}$ and
 Equality~(IV) is by induction true. Altogether we obtain that
 $\mathcal L_l$ is a labeling with $w(\mathcal L_l)=w(P)$.

 Finally we show, that there is no other labeling~$\mathcal L$
 with~$w(\mathcal L)<w(\mathcal L_P)$. Assume for the sake of
 contradiction that there is such a labeling~$\mathcal L$.
 Let~$\bot=\sigma_0,\sigma_1,\dots,\sigma_l=\top$ be the switchovers
 in~$\mathcal L$, such that $\sigma_i < \sigma_{i+1}$ for each~$0\leq
 i < l$. We observe that two consecutive switchovers~$\sigma_i$
 and~$\sigma_{i+1}$ are compatible. Hence, for any~$i$ with~$1\leq i
 \leq l$ there is an edge~$e_i=(u,v)$ in~$G'$
 with~$\sigma_u=\sigma_{i-1}$ and $\sigma_{v}=\sigma_{i+1}$. Consequently,
 the edges~$e_1,\dots,e_l$ form an $x$-$y$ path~$Q$. By the first two claims of
 this lemma, there is a labeling~$\mathcal L_Q$ of~$I$
 with~$w(Q)=w(\mathcal L_Q)$. Since~$P$ is a shortest~$x$-$y$ path it
 holds
 \[
   w(\mathcal L) < w(\mathcal L_P)=w(P) \leq  w(Q)= w(\mathcal L_Q) 
 \]
 We now show that $w(\mathcal L_Q)\leq w(\mathcal L)$ deriving a
 contradiction. To that end recall
 that
 \[
\mathcal L_Q=\bigcup_{i=1}^l\mathcal
 A_{\sigma_{i-1},\sigma_i}.\] For $1\leq i \leq l$ the set $\mathcal
 A_{\sigma_{i-1},\sigma_i}$ is an optimal labeling of the
 instance~$I(\sigma_{i-1},\sigma_{i}]$ such that $\mathcal
 A_{\sigma_{i-1},\sigma_i} \cup \sigma_{i-1}$ is a labeling
 and~$\sigma_{i}$ is the only switchover contained in~$\mathcal
 A_{\sigma_{i-1},\sigma_i}$. Let~$\mathcal L_i \subseteq \mathcal L$ and $\mathcal
 L^Q_i\subseteq \mathcal L_Q$ be the labelings restricted
 to~$I(\sigma_{i-1},\sigma_i]$. In particular we have~$\mathcal
 L^Q_i=\mathcal A_{\sigma_{i-1},\sigma_i}$ and~$\sigma_i$ is the only
 switchover in~$\mathcal L_i$. If we had $w(\mathcal L_Q)> w(\mathcal
 L)$, there must be two consecutive switchovers~$\sigma_{i-1}$ and
 $\sigma_{i}$ such that $w(\mathcal L_i)<w(\mathcal L^Q_i)$. However,
 this contradicts the optimality of~$\mathcal
 A_{\sigma_{i-1},\sigma_i}$.  Consequently, it holds~$w(\mathcal L_Q)\leq w(\mathcal L)$
 yielding the claimed contradiction.  \qed
\end{proof}

By Lemma~\ref{lem:two-sided} a
shortest $x$-$y$ path~$P$ in~$G'$ corresponds with an optimal
labeling~$\mathcal L$ of~$C$, if this exists. Using \SP we
construct~$P$ in~$O(|V'|+|E'|)$ time.  Since~$\switches$
contains~$O(nk^2)$ switchovers, the graph $G'$ contains $O(nk^2)$
vertices and $O(n^2 k^4)$ edges.  As \SP considers each edge only
once, we compute the edges of~$G'$ on demand, which needs $O(nk^2)$
storage. We compute the costs of the incoming edges of
a vertex~$v\in V'$ utilizing the one-sided case. Proceeding naively,
we need $O(nk^2)$ time per edge, which yields $O(n^3k^6)$ time in
total. 

Reusing already computed information, we improve that result as follows.
Let~$(u_1,v),\dots,(u_k,v)$ denote the incoming edges of~$v$ such
that~$\sigma_{u_1}\leq\dots\leq\sigma_{u_k}$, i.e., the stop of
$\sigma_{u_i}$'s first label does not lie after the stop of
$\sigma_{u_j}$'s first label with $i<j$. Further, let~$\sigma_v=(\ell^1_v,\ell^2_v)$ and let~$G_i$ be
the graph for the one-sided instance~$I[s_i,s]$ considering only candidates that lie on the same side as~$\ell^1_v$, where~$s_i$ is
the stop of the second label of~$\sigma_{u_i}$ and $s$ is the
stop of~$\ell^1_v$. Let~$P_i$ be the shortest
$x_i$-$y_i$ path in $G_i$, where~$x_i$ and $y_i$ denote the source
and target of~$G_i$, respectively.  We observe that excluding the
source and target, the graph $G_i$ is a sub-graph of~$G_1$
for all $1\leq i \leq k$. Further, since a sub-path of a shortest path
is also a shortest path among all paths having the same end vertices,
we can assume without loss of optimality that when excluding~$x_i$ and
$y_i$ from~$P_i$, the path~$P_i$ is a sub path of~$P_1$ for all~$1\leq
i \leq k$. We therefore only need to compute~$G_1$ and~$P_1$ and
can use sub-paths of~$P_1$ in order to gain the costs of all
incoming edges of~$v$. Hence, we basically apply for each vertex~$v\in
V$ the algorithm for the one-sided case once using $O(nk^2)$ time per
vertex. We then can compute the costs in~$O(1)$ time per edge, which
yields the next result.

\begin{theorem}
\RMLL can be optimally solved in~$O(n^2 k^4)$ time and~$O(nk^2)$ space.
\end{theorem}

\begin{figure*}[t]
  \centering
  \includegraphics[width=\textwidth,page=1]{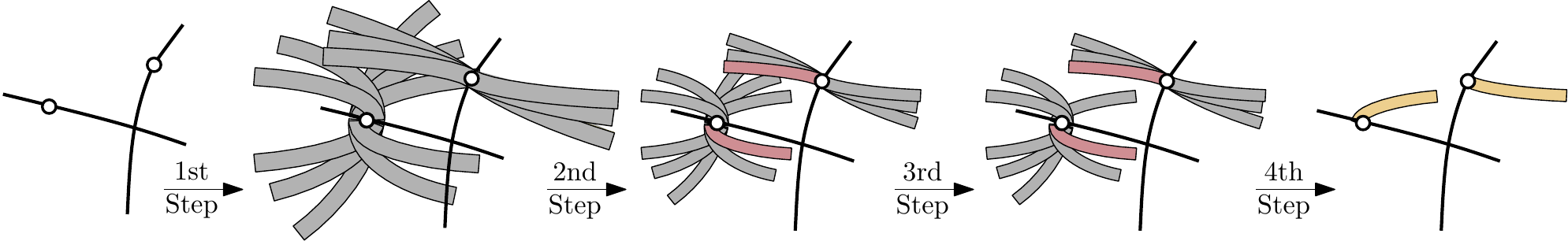}
  \caption{Schematic illustration of the presented workflow.
    \ 1st Step: generation of candidates. 2nd
    Step: scaling and creation of initial labeling (red labels). 3rd Step: pre-selection of candidates. 4th
    Step: Solving the metro lines independently.
   }
\label{fig:approach:pipeline} 
\end{figure*}

\section{Cost Function}
\label{sec:weighting-function}

In this section we motivate the cost function introduced in
Section~\ref{sec:single-line}.  For a given metro map $(\{C\},\stops)$
that consists of a single metro line~$C$, and generated
candidates~$\cands$, we rate each labeling~$\mathcal L \subseteq
\cands$ using the following cost function:
\[
w(\mathcal L)=\sum_{\ell \in \mathcal L}w_1(\ell) +
\smashoperator[r]{\sum_{(\ell_1,\ell_2)\in P_{\mathcal L}}}w_2(\ell_1,\ell_2) +
\smashoperator[r]{\sum_{(\sigma_1,\sigma_2)\in \consSwitches_{\mathcal L}}} w_3(\sigma_1,\sigma_2),
\]
where~$w_1\colon \mathcal L \to \mathbb{R}$ rates a single label,
$w_2\colon P_{\mathcal L} \to \mathbb{R}$ rates two consecutive labels
and $w_3\colon \consSwitches_{\mathcal L} \to \mathbb{R}$ rates two
consecutive switchovers; see Assumption~\ref{assumption:weighting-function}. The definition of this function relies on the following considerations, which are based on Imhof's ``general principles and requirements'' for map labeling \cite{imhof}. 

To respect that some (e.g., steep and highly curved) labels
are more difficult to read than others, we introduce a cost
$w_1(\ell)$ for each candidate label $\ell$. 
 
Imhof further notes that ``names should assist directly in revealing
spatial situation'' and exemplifies this principle with maps that show
text only while still conveying the most relevant geographic
information.  To transfer this idea to metro maps, we favor solutions
where the labels for each two consecutive stops on a metro line have
similar properties.  That is, the two labels should be placed on the
same side of the line and their slopes and curvatures should be
similar.  In a map satisfying this criterion, a user need not find the
point-text correspondence on a one-to-one basis.  Instead, the user
can identify metro lines and sequences of stops based on label groups,
which, for example, makes it easier to count the stops till a
destination.  (Of course, this is also an improvement in terms of
legibility.)  In our model, we consider the similarity of consecutive
labels by introducing a cost $w_2(\ell_1, \ell_2)$ for each pair
$(\ell_1,\ell_2)$ of candidates that belong to consecutive stops on $C$. We penalize consecutive candidates that lie on opposite
sides of the metro line, because those disturb the overall label
placement.  We add this cost to the objective value of a solution if
both candidates are selected.  Since we minimize the total costs of
the solution, the cost for a pair of candidates should be low if
they are similar. Further, if~$C$ has labels that do not
lie on the same side of~$C$, the implied switchovers should occur in
regular distances and not cluttered. Hence, for each pair $\sigma_1$
and $\sigma_2$ of two consecutive switchovers in a solution, we add a
cost~$w_3(\sigma_1,\sigma_2)$ to the objective value of the solution
that depends on the distance between $\sigma_1$ and $\sigma_2$; the
smaller the distance, the greater the cost
of~$w_3(\sigma_1,\sigma_2)$.

We now describe more precisely how we defined~$w_1$, $w_2$ and $w_3$
for our evaluation. The definitions depend on the applied labeling
style.

\textbf{Curved Labels.} Using \CurvedLabels for the labels we define
the cost functions as follows. For a label~$\ell\in \cands$
let~$p^1_\ell=(x^1_\ell,y^1_\ell)$ be the start point of~$\ell$ and
let~$p^2_\ell=(x^2_\ell,y^2_\ell)$ be the end point of~$\ell$; recall
that we derived the labels from B\'{e}zier curves. Let~$\vec v_\ell$
be the vector connecting~$p^1_\ell$ with $p^2_\ell$ and
let~$\alpha\in[0,2\pi]$ denote the angle of~$\vec v_\ell$. We define
\[
\delta_\ell =\begin{cases}
               \alpha_\ell & 0 \leq \alpha \leq \frac{\pi}{2}\\
               \pi-\alpha_\ell & \frac{\pi}{2} < \alpha \leq \pi\\
               \alpha_\ell-\pi & \pi < \alpha \leq \frac{3}{2}\pi\\
               2\pi-\alpha_\ell & \frac{3}{2}\pi < \alpha < 2\pi\\
            \end{cases}
\] 
Hence, the angle~$\delta_\ell$ is a measure for how horizontal the vector~$\vec v_\ell$ is, whereby the smaller the value of~$\delta_\ell$, the more horizontal is $\vec v_\ell$.

We defined the cost function $w_1$ rating a single label~$\ell\in
\cands$ to be $w_1=10\cdot \delta_\ell$. Hence, we penalize steep
labels. We defined~$w_2$ rating two consecutive labels~$\ell_1$ and
$\ell_2$ as follows. If~$\ell_1$ and $\ell_2$ point into different
$x$-direction, $w_2$ is 150. Else, if~$\ell_1$ and $\ell_2$ are
switchovers, the cost $w_2$ is 0. In all other cases, $w_2$ is the difference
between the angle $\alpha_{\ell_1}$ and $\alpha_{\ell_2}$. Hence, in
the latter case we penalize labels that are differently aligned.
Finally,~$w_3$ rating two switchovers $\sigma_1$ and $\sigma_2$ of the
same line is defined as $w_3=\frac{200}{d_{\sigma_1,\sigma_2}}$,
where $d_{\sigma_1,\sigma_2}$ is the number of stops in between
$\sigma_1$ and $\sigma_2$. In particular $w_3$ effects that a labeling
with equally sized sequences of labels lying on the same side of their
metro line are rated better than a labeling where the sequences are
sized irregularly.

\textbf{Octilinear Labels.} Using \OctiLabels for the labels we define the 
cost functions as follows. Recall that we use $\OctiLabels$ for 
octilinear metro maps.  For a label $\ell$ let $l_\ell$ be the segment on 
which its stop is placed.
If $\l_\ell$ is horizontal, but $\ell$ is not diagonal, we set $w_1=200$.
If $\l_\ell$ is vertical or diagonal, but $\ell$ is not horizontal, we set~$w_1=100$. In all other cases we set $w_1=0$. The functions $w_2$ and $w_3$ are defined in the same way as for \CurvedLabels.

\section{Multiple Metro Lines}\label{sec:multi-lines}
In this section we consider the problem that we are given a metro map
$(\lines,\stops)$ consisting of multiple metro lines. We present an
algorithm that creates a labeling for~$(\lines,\stops)$ in two
phases. Each phase is divided into two steps;
see Fig.~\ref{fig:approach:pipeline} for a schematic illustration. In the
first phase the algorithm creates the set~$\cands$ of label candidates
and ensures that there exists at least one labeling for the metro map.
In the second phase it then computes a labeling~$\mathcal L$ for
$(\lines,\stops,\cands)$. To that end it makes use of the labeling
algorithm for a single metro line; see Section~\ref{sec:single-line}.
In order to rate~$\mathcal L$, we extend the cost function~$w$
for a single metro line on multiple metro lines, i.e., for any
$\mathcal L\subseteq \cands$ we require
\[
w(\mathcal L)=\smashoperator[r]{\sum_{C\in \lines}}w(C) \text{ with }
w(C)=\sum_{\ell \in \mathcal L_C}w_1(\ell) +
\smashoperator[r]{\sum_{(\ell_1,\ell_2)\in P_{\mathcal L_C}}}w_2(\ell_1,\ell_2) +
\smashoperator[r]{\sum_{(\sigma_1,\sigma_2)\in \consSwitches_{\mathcal L_C}}} w_3(\sigma_1,\sigma_2),
\]
where $\mathcal L_C=\mathcal L\cap \cands_C$ is the labeling restricted to metro line~$C\in\lines$,~$w_1\colon \mathcal L_C \to \mathbb{R}$ rates a single label,
$w_2\colon P_{\mathcal L_C} \to \mathbb{R}$ rates two consecutive labels
and $w_3\colon \consSwitches_{\mathcal L_C} \to \mathbb{R}$ rates two
consecutive switchovers of~$C$. In particular~$w$ satisfies Assumption~\ref{assumption:weighting-function} for a single metro line.

Altogether, the workflow yields a heuristic that relies on the conjecture that
using optimal algorithms in single steps is sufficient to obtain good labelings. In our evaluation we call that approach \DYNALG.

\subsection{First Phase -- Candidate Generation}
\label{sec:basic-approach}
  First, we create the label
  candidates~$\cands$. We then enforce that there is a
  labeling~$\mathcal L$ for the given
  instance~$I=(\lines,\stops,\cands)$.

\paragraph{1st Step -- Candidate Creation.}
Depending on the labeling style, we generate a discrete set of
candidate labels for every stop.  Hence, we are now given the
instance~$(\lines,\stops,\cands,w)$. In particular we assume that each
candidate~$\ell\in \cands$ is assigned to one side of its metro
line~$C$, namely to the left or right side of~$C$, and, furthermore, $w$ is a cost function satisfying Assumption~\ref{assumption:weighting-function} for each metro line~$C$.

\paragraph{2nd Step -- Scaling.}
Since each stop of a metro map must be labeled, we first apply a
transformation on the given candidates to ensure that there is at
least one labeling of the metro map.
To that end we first determine for each stop~$s \in \stops$ of each
metro line~$C\in \lines$ two candidates of~$\cands_s$ that are assigned to
opposite sides of~$C$. More specifically, among all candidates
in~$\cands_s$ that are assigned to the right hand side of~$C$ and that do not
intersect any metro line of~$\lines$, we take that
candidate~$\ell_{\RR}\in \cands_s$ with minimum costs, i.e.,
$w_1(\ell)$ is minimal.  If such a candidate does not exist, because
each label of~$\cands_s$ intersects at least one metro line, we take a
pre-defined label~$\ell_{\RR}\in \cands_s$ that is assigned to the right hand
side of~$C$. In the same manner we choose a candidate~$\ell_{\LL}\in
\cands_s$ that is assigned to the left hand side of~$C$.
Let~$D_s=\{\ell_{\RR},\ell_{\LL}\}$, we now enforce that there is a
labeling~$\mathcal L$ of $I$ such that
 for each stop~$s\in \stops$ it
contains a label of~$D_{s}$.

We check whether the set~$\bigcup_{s \in \stops}D_s$ admits a
labeling~$\mathcal L$ for~$I$. Later, we describe more specifically
how to do this. If~$\mathcal L$ exists, we continue with the third
step of the algorithm using $I$ and $\mathcal L$ as input.  Otherwise,
we scale all candidates of~$\cands$ smaller by a constant factor and
repeat the described procedure.  Sampling a pre-defined scaling range
$[x_\mathrm{min},x_\mathrm{max}]$, we find in that manner a scaling
factor~$x\in[x_\mathrm{min},x_\mathrm{max}]$ for the candidates that
admits a labeling~$\mathcal L$ of~$I$. We choose $x$ as large as
possible. If we could not find $x$, e.g, because we have chosen
$[x_\mathrm{min},x_\mathrm{max}]$ or the sampling not appropriately,
we abort the algorithm, stating that the algorithm could not find a
labeling.

Next, we describe how to check whether there is a labeling for
$\bigcup_{s \in \stops}D_s$.  Since for each~$s\in \stops$ the
set~$D_s$ contains two candidates, we can make use of a 2SAT
formulation to model the labeling problem.  For each stop
$s \in \stops$ and each candidate $\ell\in D_s$, we introduce the Boolean
variable $x_{\ell}$.
Those Boolean variables induce the set
$\mathcal L=\{\ell \in \bigcup_{s\in S}D_s \mid x_\ell \text{ is true}\}$.
The following formulas are satisfiable if and only
if~$\mathcal L$ is a labeling of $\lines$. 
\begin{align*}
  \lnot x_\ell &\quad \forall s\in \stops\ \forall \ell\in D_s \text{ s.t. } \exists
C\in \lines \text{ that intersects } \ell\\
  \lnot x_\ell \vee \lnot x_{\ell'} &\quad  \forall \ell, \ell' \in
\bigcup\nolimits_{s\in S}D_s \text{ s.t. } \ell \text{ and } \ell' \text{ intersect.}\\
  x_\ell \vee x_{\ell'}, \lnot x_\ell \vee \lnot x_{\ell'}  &\quad
\forall s\in \stops\ \forall \ell, \ell'\in D_s
\end{align*}
The first formula ensures that there is no label of the solution that
intersects any metro line. The second one avoids overlaps between
labels, while the two last formulas enforce that for each stop $s\in
\stops$ there is exactly one label of~$D_s$ that is contained in the
solution. 

According to~\cite{two-sat} in linear time with respect to the number
of variables and formulas, the satisfiability can be checked. We
introduce $O(n)$ variables and instantiate the second formula $O(n^2)$
times, because each pair of candidates may overlap. The remaining
formulas are instantiated in $O(n^2)$ time. Hence, the total running time is in $O(n^2\cdot
t)$, where~$t$ denotes the number of scaling
steps.

\subsection{Second Phase -- Candidate Selection}
We assume that we are given the scaled instance
$I=(\lines,\stops,\cands)$ and the labeling $\mathcal L$ for $I$ of
the previous phase.  We apply a pre-selection on the candidates by
discarding candidates such that no two stops' candidates of different
metro lines intersect and each metro line satisfies
Assumption~\ref{assumption:separation} and
Assumption~\ref{assumption:transitivity-property}. 
We never remove a label in $\mathcal L$ from the candidate set, however, to
ensure that there is always a feasible solution.
 Finally,
considering the metro lines independently, we select for each stop a
candidate to be its label using the dynamic program described in
Section~\ref{sec:single-line}.

\paragraph{3th Step -- Candidate Pre-Selection.}
We first ensure that~$I$ satisfies
Assumption~\ref{assumption:separation} and
Assumption~\ref{assumption:transitivity-property}. If two candidates
$\ell$ and $\ell'$ of the same metro line~$C$ intersect and if they
are assigned to opposite sides of~$C$, we delete one of both labels as
follows. If $\ell \in \mathcal L$, we delete $\ell'$, and if $\ell'\in
\mathcal L$ we delete $\ell$. Otherwise, if none of both is contained
in $\mathcal L$, we delete that label with higher costs; ties are
broken arbitrarily. Afterwards, Assumption~\ref{assumption:separation}
is satisfied.  Further, for each metro line $C$ of $I$ we iterate
through the stops of~$C$ from its beginning to its end. Doing so, we
delete candidates from~$\cands$ violating
Assumption~\ref{assumption:transitivity-property} as described as
follows. Let~$s$ be the currently considered stop. For each
candidate~$\ell\in \cands_s$ we check for each stop~$s'$ with~$s'<s$
whether there is a label~$\ell'\in\cands_{s'}$ that
intersects~$\ell$. If~$\ell'$ exists, we check whether each
candidate~$\ell''\in \cands_{s''}$ of any stop~$s''$ with $s'<s''<s$
intersects $\ell$ or~$\ell'$. If this is not the case, we delete
$\ell$ if $\ell$ is not contained in~$\mathcal L$, and otherwise we
delete~$\ell'$. Note that not both can be contained in~$\mathcal
L$. By construction each metro line of instance~$I$ then satisfies
Assumption~\ref{assumption:transitivity-property}, where $\mathcal
L\subseteq \cands$.

Finally, we ensure that the metro lines in~$\lines$ become
\emph{independent} in the sense that no candidates of stops belonging
to different metro lines intersect and no candidate intersects any
metro line. Hence, after this step, the metro lines can independently
be labeled such that the resulting labelings compose to a labeling of
$I$. 

We first rank the candidates of~$\cands$ as follows. For each metro
line~$C\in \lines$ we construct a labeling~$\mathcal L_C$ using
the dynamic program for the two sided case as presented in
Section~\ref{sec:single-line}. Due to the previous step, 
those labelings exist. Note that for two metro
lines~$C,C' \in \lines$ there may be labels $\ell\in \mathcal L_C$
and~$\ell'\in \mathcal L_{C'}$ that intersect each other. For each
candidate~$\ell\in\cands$ we set~$val_\ell=1$ if $\ell\in \mathcal
L_{C}$ for a metro line~$C\in \lines$, and $val_\ell=0$ otherwise. A
candidate~$\ell\in\cands$ has a \emph{smaller rank} than a
candidate~$\ell'\in \cands$, if~$val_\ell>val_{\ell'}$ or
$val_\ell=val_{\ell'}$ 
and $w_1(\ell)\leq w_1(\ell')$; ties are broken
arbitrarily.

We now greedily remove candidates from~$\cands$ until all metro lines
are independent. We create a conflict graph~$G=(V,E)$
such that the vertices of~$G$ are the candidates and the edges model intersections between
candidates, i.e., two vertices are adjacent if and only if the
corresponding labels intersect. Then, we delete all vertices whose
corresponding labels intersect any metro line. Afterwards, starting with~$\mathcal I=\emptyset$,
we construct an independent set~$\mathcal I$ on $G$ as follows.
First, we add all vertices of~$G$
to~$\mathcal I$, whose labels are contained in~$\mathcal L$, and
delete them and their neighbors from~$G$. Since the labels in
$\mathcal L$ do not intersect, $\mathcal I$ is an
independent set of the original conflict graph. In the
increasing order of their ranks, we remove each vertex~$v$
and its neighbors from~$G$. Each time we add~$v$ to
$\mathcal I$; obviously sustaining that~$\mathcal I$ is an independent set in~$G$. 
We then update for each stop~$s\in \stops$ its
candidate set~$\cands_s$ to~$\cands_s \gets  \{\ell\in \cands_s \mid \text{vertex of $\ell$ is contained in }\mathcal I\}$.
Since all labels of~$\mathcal L$ are contained in $\mathcal I$, there is a
labeling for $(\lines, \stops, \cands,w)$ based on the
new candidate set~$\cands$.

\paragraph{4th Step -- Final Candidate Selection.}
Let $I=(\lines,\stops,\cands,w)$ be the instance after applying the
third step and $\mathcal L$ be the labeling that has been created in
the first step.  By the previous step 
the metro lines are in the sense independent that candidates of stops
belonging to different metro lines do not intersect. Further, they all
satisfy Assumption~\ref{assumption:separation} and
Assumption~\ref{assumption:transitivity-property}.
Hence, we use the dynamic programming approach of Section~\ref{sec:single-line}
in order to label them independently. The composition of those labelings
is then a labeling~$\mathcal L$ of~$I$.

\section{Alternative Approaches}
We now present the three approaches \ILPALG, \SCALEALG, \GREEDYALG,
which are adaptions of our workflow. We use those to experimentally
evaluate our approach against alternatives. While \GREEDYALG is a
simple and fast greedy algorithm, \ILPALG and \SCALEALG are based on
an ILP formulation.
\subsection{Integer Linear Programming Formulation}
To assess the impact of the second phase of our approach, we present
an integer linear programming formulation that optimally solves \MML
with respect to the required cost function. Let
$(\lines,\stops,\cands,w)$ be an instance of that problem, which we
obtain after the first phase of our approach. We first note that we
apply a specific cost function (see
Section~\ref{sec:weighting-function}). The cost function $w_3$, which
rates two consecutive switchovers $\sigma=(\ell_1,\ell_2)$ and
$\sigma'=(\ell'_1,\ell'_2)$ of a labeling~$\mathcal L$, does not rely
on the actual switchovers, but only on their positions on the
corresponding metro line. Hence, we may assume that $w_3$ expects the
stops of $\ell_1$ and $\ell'_1$. This assumption helps us reduce the
number of variables.

For each candidate~$\ell \in \cands$ we define a binary
variable~$x_\ell \in \{0,1\}$. If $x_\ell=1$, we interpret it such
that $\ell$ is selected for the labeling. We introduce the following
constraints.
\begin{align}
  x_\ell&=0 && \forall \ell \in \cands \text{ that intersect a metro line.}\label{ilp:c1}\\
  x_\ell+x_{\ell'}&\leq 1 && \forall \ell,\ell' \in \cands \text{ that intersect.}\label{ilp:c2}\\
  \sum_{\ell \in \cands_s}x_\ell&=1 && \forall s\in S \label{ilp:c3}  
\end{align}
  Moreover, for each metro line~$C\in \lines$ we define
the following variables. To that end, let $\mathcal P\in
\cands_C\times\cands_C$ be the set of all consecutive labels and let
$s_1,\ldots,s_n$ be the stops of $C$ in that particular order.
\begin{align*}
  y_{\ell,\ell'}&\in \{0,1\}, && \forall (\ell,\ell')\in \mathcal P.\\
  z_{i}&\in\{0,1\}, &&   1\leq i < n\\
 h_{i,j}&\in\{0,1\}, &&  1\leq i < n \text{ and } i < j < n 
\end{align*}
If $y_{\ell,\ell'}=1$, we interpret it such that both $\ell$ and
$\ell'$ are selected for the labeling. If $z_i=1$, we interpret it
such that the selected labels of the stops $s_i$ and $s_{i+1}$ form a
switchover. If $h_{i,j}=1$, we interpret it such that the selected
labels at $s_i$, $s_{i+1}$, $s_j$ and $s_{j+1}$ form two consecutive
switchovers.
Further, for each metro line~$C$  we introduce the following constraints.  In these constraints $L(\cands)$ denotes the set of labels that lie to the left of $C$, and $R(\cands)$ denotes the set of labels that lie to the right of~$C$. 
\begin{align}
  x_\ell + x_{\ell'}-1 &\leq y_{\ell,\ell'} && \forall (\ell,\ell') \in \mathcal P\label{ilp:c4}\\
  \smashoperator[r]{\sum_{\ell \in L(\cands_{s_{i}})}}x_\ell +\smashoperator[r]{\sum_{\ell \in R(\cands_{s_{i+1}})}}x_\ell-1 &\leq z_i && 1\leq i < n\label{ilp:c5a}\\
  \smashoperator[r]{\sum_{\ell \in R(\cands_{s_{i}})}}x_\ell +\smashoperator[r]{\sum_{\ell \in L(\cands_{s_{i+1}})}}x_\ell-1 &\leq z_i && 1\leq i < n\label{ilp:c5b}\\
  z_i+z_j-1-\smashoperator[lr]{\sum_{i<k<j}}z_k&\leq h_{i,j}& & 1\leq i < j < n\label{ilp:c6} 
\end{align}
We further define for each metro line~$C$ the following linear term.

\begin{align}
  w(C):=\sum_{\ell\in \cands_C}x_\ell\cdot w_1(\ell) +
  \smashoperator[lr]{\sum_{(\ell,\ell')\in \mathcal P}}y_{\ell,\ell'}
  \cdot w_2(\ell,\ell') +
  \smashoperator[lr]{\sum_{\stackrel{1\leq i < n}{i<j<n}}}h_{i,j}\cdot
  w_3(s_i,s_j)
\end{align}
Subject to the presented Constraints (\ref{ilp:c1})--(\ref{ilp:c6}) we then minimize 

\begin{align}
  \sum_{C\in \lines} w(C).\label{ilp:objective}
\end{align}
Consider a variable assignment that minimizes (\ref{ilp:objective}) and satisfies Constraints
(\ref{ilp:c1})--(\ref{ilp:c6}). We show that~\[\mathcal
L =\{\ell \in \cands \mid x_\ell =1 \}\] is an optimal labeling of the
given instance with respect to the given cost function.
First of all,~$\mathcal L$ is a valid labeling. Constraint~(\ref{ilp:c1})
ensures that no label in~$\mathcal L$ intersects any metro line. By
Constraint~(\ref{ilp:c2}) the labels in~$\mathcal L$ are pairwise
disjoint. Finally, by Constraint~(\ref{ilp:c3}) for each stop there is
exactly one label contained in~$\mathcal L$.
In particular, for a metro line~$C\in \lines$, the set~$\mathcal L_C =
\mathcal L\cap \cands_C$ is a valid labeling of~$C$.

We now show that~$w(\mathcal L_C)=w(C)$ for any metro line~$C\in
\lines$. Since we minimize (\ref{ilp:objective}), this implies the
optimality of~$\mathcal L$. Obviously, for a label $\ell \in \cands_C$
the cost~$w_1(\ell)$ is taken into account in~$w(C)$ if and only
if~$\ell$ belongs to~$\mathcal L$.

By Constraint~(\ref{ilp:c4}) for two consecutive labels $\ell$ and
$\ell'$ of $C$ we have $y_{\ell,\ell'}=1$ if both are contained in
$\mathcal L$. Further, by the minimality of (\ref{ilp:objective}), if
at least one of both labels does not belong to~$\mathcal L$, it holds
$y_{\ell,\ell'}=0$. Hence, $w_2(\ell,\ell')$ is taken into account
in~$w(C)$ if and only if both~$\ell$ and $\ell'$ belong
to~$\mathcal L$.

By Constraint~(\ref{ilp:c5a}) and Constraint~(\ref{ilp:c5b}) it holds
$z_{i,j}=1$ if the labels $\ell_i\in \mathcal L$ and $\ell_{i+1}\in
\mathcal L$ of the consecutive stops~$s_i$ and $s_{i+1}$ form a
switchover in $\mathcal L$. Hence, by Constraint~(\ref{ilp:c6}) it
further holds $h_{i,j}=1$ if the labels of $s_i$ and $s_{i+1}$ as well
as $s_j$ and $s_{j+1}$ form switchovers~$\sigma_i$ and $\sigma_j$ in
$\mathcal L$, and, furthermore, there is no other switchover in
between $\sigma_i$ and $\sigma_j$, i.e., both switchovers are consecutive. On
the other hand, by the minimality of~$w(C)$ in all other cases it
holds~$h_{i,j}=0$. Hence, $w_3(s_i,s_j)$ is taken
into account in~$w(C)$ if and only if~$\sigma_i$ and $\sigma_j$ are
consecutive switchovers in~$\mathcal L$.  Altogether we obtain the
following theorem.

\begin{theorem}
  Given an optimal variable assignment for the presented ILP
  formulation, the set~$\mathcal L =\{\ell \in \cands \mid x_\ell =1
  \}$ is an optimal labeling of $(\lines,\stops,\cands)$ with respect
  to~$w$.
\end{theorem}

The approach~\ILPALG simply replaces the second phase of \DYNALG by
that ILP formulation. Hence, it solves the second phase optimally. The
approach \SCALEALG samples a predefined scaling range
$[x_\textit{min},x_\textit{max}]$, which is also used by \DYNALG. For
each scale $x$ it scales the candidates correspondingly. Using the ILP
formulation it then checks whether the candidates admit a
labeling. Hence, we approximately
obtain the greatest scaling factor that admits a labeling.

\subsection{Greedy-Algorithm}
The algorithm \GREEDYALG replaces the dynamic programming approach in
our workflow as follows.  Starting with the solution~$\mathcal L$
enforced by the 1st step, the greedy algorithm iterates once through
the stops of~$C$. For each stop~$s$ of~$C$ it selects the
candidate~$\ell \in \cands_s$ that minimizes
$w_1(\ell)+w_2(\ell_p,\ell)+w_2(\ell,\ell_s)$ among all valid
candidates in~$\cands_s$, where~$\ell_p\in \mathcal L$ is the
candidate selected for the previous stop~$s_p$ and $\ell_s\in\mathcal
L$ is the candidate for the successive stop. It replaces the candidate
of~$s$ in $\mathcal L$ with~$\ell$.

\section{Evaluation}\label{sec:evaluation}
To evaluate our approach presented in Section~\ref{sec:multi-lines},
we did a case study on the metro systems of Sydney (173 stops) and
Vienna (84 stops), which have been used as benchmarks before
\cite{fink,metroMap1,wang2011}.  For Sydney we took the curved layout
from~\cite[Fig.~1a]{fink} and the octilinear layouts
from~\cite[Fig.~9a,b]{metroMap1}\cite[Fig.~10.]{wang2011}, while for
Vienna we took the curved layout from~\cite[Fig.~8c]{fink} and the
ocitlinear layouts from~\cite[Fig.~13a,b]{metroMap1}. See also
Table~\ref{table:instances} for an overview of the instances. Since
the metro lines of Sydney are not only paths, we disassembled those
metro lines into single paths.  We did this by hand and tried to
extract as long paths as possible. Hence, the instances of Sydney
decompose into 12 lines and the instances of Vienna into 5 lines.  We
took the positions of the stops as presented in the corresponding
papers. In the curved layout of Sydney we removed the stops
\emph{Tempe} and \emph{Martin Place} (in
Fig.~\ref{fig:labeligns:paper:sydney-curved-dp} marked as red dots),
because both stops are tightly enclosed by metro lines such that only
the placement of very small labels is possible. This is not so much a
problem of our approach, but of the given layout. In a semi-automatic
approach the designer would then need to change the layout. For the
curvilinear layouts we used labels of \CurvedLabels and for the
octilinear layouts we used labels of \OctiLabels. For the layouts of
\emph{Sydney2}, \emph{Sydney3} and \emph{Vienna2} the authors present
labelings; see Fig.~\ref{fig:labeling:sydneyB},
Fig.~\ref{fig:labeling:comparison:complete} and
Fig.~\ref{fig:labeling:comparison:vienna}. For any other layout they
do not present labelings.

\begin{table}
  \centering
  \caption{Overview of considered instances. \emph{Style:} Style of map and applied labels. $\frac{s}{s_\textit{max}}$: Ratio of applied
scale factors.  The scale~$s_\textit{max}$ is a lower bound for the
largest possible scaling factor (obtained by \SCALEALG), and $s$ is the
scale computed by the first phase of our workflow. }
  \label{table:instances}
  \begin{tabular}{cccccccc}
    \toprule
    \textbf{Instance}\quad &\quad \emph{Sydney1}\quad&\quad \emph{Sydney2} \quad &\quad \emph{Sydney3} \quad &\quad \emph{Sydney4}\quad & \quad \emph{Vienna1} \quad &\quad \emph{Vienna2}\quad &\quad \emph{Vienna3}\quad\\
    \midrule
    
     \textit{Style} & Octi. & Octi. & Octi. & Curved & Octi. & Octi. & Curved.\\    
     $\frac{s}{s_\textit{max}}$& 0.69 & 0.57 & 0.57 & 0.67 & 0.59 & 0.81 & 0.54\\
\textit{Reference} & \cite[Fig. 9a]{metroMap1} & \cite[Fig. 9b]{metroMap1} &\cite[Fig.~10.]{wang2011} & \cite[Fig.~1a]{fink} & \cite[Fig. 13a]{metroMap1} & \cite[Fig. 13b]{metroMap1} & \cite[Fig.~8c]{fink} \\
    \bottomrule
  \end{tabular}
\end{table}

\newcommand{\SKIP}{$\star$}
\begin{table}[]
\centering
\caption{Running times in seconds of the workflow broken down into its two phases and their single steps (if applicable). \emph{Algo.}: The applied algorithm. Times less than 0.01 seconds are marked with~\SKIP.  }
\label{table:running-time}
\begin{tabular}{ccccccccccc}
\toprule
         &      &Algo.  & \multicolumn{3}{c}{\textit{Phase 1: Creation}}                    &        \multicolumn{3}{c}{\textit{Phase 2: Selection}} &  \\ \cmidrule(lr){4-6}\cmidrule(lr){7-9}
\textit{Instance} &Layout && \textit{1st}                  & \textit{2nd}                  & $\sum$                & \textit{1st}           & \textit{2nd}           & $\sum$           & \quad&  \textbf{Total}     \\\midrule
\multirow{4}{*}{Sydney1} & \multirow{4}{*}{Octi.} & \DYNALG & 0.03 & 0.18 & 0.21 & 0.09 & 0.1 & 0.18 & & \textbf{0.39} \\
 & & \GREEDYALG & 0.02 & 0.18 & 0.2 & \SKIP & \SKIP & \SKIP & & \textbf{0.2} \\
 & & \SCALEALG & 0.02 & -- & 0.02 & -- & -- & 69.09 & & \textbf{69.11} \\
 & & \ILPALG & 0.02 & 0.17 & 0.19 & -- & -- & 25.43 & & \textbf{25.62} \\
\midrule
\multirow{4}{*}{Sydney2} & \multirow{4}{*}{Octi.} & \DYNALG & 0.03 & 0.28 & 0.31 & 0.07 & 0.09 & 0.15 & & \textbf{0.46} \\
 & & \GREEDYALG & 0.02 & 0.28 & 0.3 & \SKIP & \SKIP & \SKIP & & \textbf{0.3} \\
 & & \SCALEALG & 0.02 & -- & 0.02 & -- & -- & 17.29 & & \textbf{17.31} \\
 & & \ILPALG & 0.02 & 0.28 & 0.29 & -- & -- & 15.17 & & \textbf{15.46} \\
\midrule
\multirow{4}{*}{Sydney3} & \multirow{4}{*}{Octi.} & \DYNALG & 0.03 & 0.17 & 0.2 & 0.07 & 0.09 & 0.16 & & \textbf{0.36} \\
 & & \GREEDYALG & 0.02 & 0.16 & 0.17 & \SKIP & \SKIP & \SKIP & & \textbf{0.17} \\
 & & \SCALEALG & 0.02 & -- & 0.02 & -- & -- & 33.61 & & \textbf{33.63} \\
 & & \ILPALG & 0.02 & 0.15 & 0.17 & -- & -- & 21.56 & & \textbf{21.73} \\
\midrule
\multirow{4}{*}{Sydney4} & \multirow{4}{*}{Curved} & \DYNALG & 0.03 & 0.33 & 0.36 & 0.1 & 0.1 & 0.2 & & \textbf{0.56} \\
 & & \GREEDYALG & 0.03 & 0.33 & 0.35 & \SKIP & \SKIP & \SKIP & & \textbf{0.35} \\
 & & \SCALEALG & 0.03 & -- & 0.05 & -- & -- & 243.58 & & \textbf{243.63} \\
 & & \ILPALG & 0.02 & 0.34 & 0.36 & -- & -- & 98.81 & & \textbf{99.17} \\
\midrule
\multirow{4}{*}{Vienna1} & \multirow{4}{*}{Octi.} & \DYNALG & 0.01 & 0.06 & 0.07 & 0.01 & \SKIP & 0.02 & & \textbf{0.09} \\
 & & \GREEDYALG & 0.01 & 0.06 & 0.07 & \SKIP & \SKIP & \SKIP & & \textbf{0.07} \\
 & & \SCALEALG & 0.01 & -- & 0.01 & -- & -- & 3.31 & & \textbf{3.32} \\
 & & \ILPALG & 0.01 & 0.06 & 0.08 & -- & -- & 0.55 & & \textbf{0.63} \\
\midrule
\multirow{4}{*}{Vienna2} & \multirow{4}{*}{Octi.} & \DYNALG & 0.01 & 0.09 & 0.1 & \SKIP & \SKIP & \SKIP & & \textbf{0.11} \\
 & & \GREEDYALG & 0.01 & 0.1 & 0.11 & \SKIP & \SKIP & \SKIP & & \textbf{0.11} \\
 & & \SCALEALG & 0.01 & -- & 0.01 & -- & -- & 4.65 & & \textbf{4.66} \\
 & & \ILPALG & 0.01 & 0.08 & 0.09 & -- & -- & 0.36 & & \textbf{0.45} \\
\midrule
\multirow{4}{*}{Vienna3} & \multirow{4}{*}{Curved} & \DYNALG & 0.02 & 0.15 & 0.17 & 0.02 & 0.04 & 0.06 & & \textbf{0.23} \\
 & & \GREEDYALG & 0.02 & 0.14 & 0.16 & \SKIP & \SKIP & \SKIP & & \textbf{0.16} \\
 & & \SCALEALG & 0.02 & -- & 0.02 & -- & -- & 7.27 & & \textbf{7.29} \\
 & & \ILPALG & 0.01 & 0.13 & 0.15 & -- & -- & 0.49 & & \textbf{0.64} \\
\bottomrule
\end{tabular}
\end{table}

\begin{table}[]
\centering
\caption{Experimental results for Sydney and Vienna. \emph{Algo.}: The applied algorithm. \emph{Candidates}: Values concerning candidates; No.\ of candidates after the first and third step, \emph{A\ref{assumption:separation}\ref{assumption:transitivity-property}}= No.\  of labels removed to establish Assumption~\ref{assumption:separation} and Assumption~\ref{assumption:transitivity-property}, SO= No.\ of switchovers. \emph{Cost}: Ratio of costs; $\mathcal L_{A}=$ labeling obtained by procedure $A\in\{\DYNALG,\GREEDYALG,\ILPALG,\SCALEALG\}$,~$\mathcal L=$ labeling obtained by procedure \ILPALG. \emph{Sequence}: Values concerning sequences of labels lying on the same side of their metro line; \emph{min}=length of shortest sequence, \emph{max}=length of longest sequence, \emph{avg.}=average length of sequences.
}
\label{table:quality}
\begin{tabular}{cclccclccccccc}
\toprule
         &        &      & \multicolumn{4}{c}{Candidates}                         & \multicolumn{4}{c}{Cost}                                                                                                                                                        & \multicolumn{3}{c}{Sequence} \\ \cmidrule(lr){4-7} \cmidrule(lr){8-11} \cmidrule(lr){12-14}
Instance & Layout & Algo. & 1st & 3rd                  & A12                  & SO & $\frac{w_1(\mathcal L_A)}{w_1(\mathcal L)}$ & $\frac{w_2(\mathcal L_A)}{w_2(\mathcal L)}$ & $\frac{w_3(\mathcal L_A)}{w_3(\mathcal L)}$ & $\frac{w(\mathcal L_A)}{w(\mathcal L)}$ & min     & max     & avg.     \\\midrule
\multirow{4}{*}{Sydney1} & \multirow{4}{*}{Octi.} & \DYNALG & 1164 & 1004 & 9 & 4 & 1.73 & 1.12 & 0.63 & 1.16 & 3 & 41 & 9.53\\
 & & \GREEDYALG & 1164 & 1005 & 9 & 15 & 1.73 & 1.87 & 7.62 & 2.64 & 1 & 24 & 5.81\\
 & & \SCALEALG & 1164 & -- & -- & 13 & 2.46 & 1.67 & 2.91 & 1.97 & 3 & 20 & 6.97\\
 & & \ILPALG & 1164 & -- & -- & 7 & 1 & 1 & 1 & 1 & 3 & 20 & 8.1\\
\midrule
\multirow{4}{*}{Sydney2} & \multirow{4}{*}{Octi.} & \DYNALG & 1104 & 963 & 3 & 6 & 1.36 & 1.08 & 3.25 & 1.33 & 1 & 26 & 9.34\\
 & & \GREEDYALG & 1104 & 964 & 3 & 19 & 2.46 & 1.83 & 19.52 & 3.54 & 1 & 35 & 5.94\\
 & & \SCALEALG & 1104 & -- & -- & 15 & 3.91 & 1.92 & 8.88 & 2.91 & 1 & 19 & 6.76\\
 & & \ILPALG & 1104 & -- & -- & 5 & 1 & 1 & 1 & 1 & 3 & 25 & 9.48\\
\midrule
\multirow{4}{*}{Sydney3} & \multirow{4}{*}{Octi.} & \DYNALG & 1132 & 1013 & 0 & 6 & 1 & 1 & 1 & 1 & 3 & 25 & 8.76\\
 & & \GREEDYALG & 1132 & 1013 & 0 & 14 & 2.15 & 1.79 & 5.99 & 2.61 & 1 & 23 & 6.9\\
 & & \SCALEALG & 1132 & -- & -- & 12 & 2.72 & 2.68 & 3.06 & 2.76 & 2 & 21 & 7.02\\
 & & \ILPALG & 1132 & -- & -- & 6 & 1 & 1 & 1 & 1 & 3 & 25 & 8.76\\
\midrule
\multirow{4}{*}{Sydney4} & \multirow{4}{*}{Curved} & \DYNALG & 1275 & 1020 & 2 & 8 & 1 & 1.13 & 1.15 & 1.01 & 3 & 27 & 7.38\\
 & & \GREEDYALG & 1275 & 1021 & 2 & 28 & 1 & 1.08 & 12.81 & 1.16 & 1 & 24 & 4.82\\
 & & \SCALEALG & 1275 & -- & -- & 13 & 1.05 & 1.69 & 2.28 & 1.1 & 2 & 20 & 6.21\\
 & & \ILPALG & 1275 & -- & -- & 7 & 1 & 1 & 1 & 1 & 3 & 27 & 7.92\\
\midrule
\multirow{4}{*}{Vienna1} & \multirow{4}{*}{Octi.} & \DYNALG & 466 & 403 & 0 & 2 & 0.74 & 1.29 & 0.99 & 1.07 & 8 & 21 & 12.8\\
 & & \GREEDYALG & 466 & 403 & 0 & 7 & 1.26 & 2.43 & 8.59 & 2.48 & 1 & 21 & 9.41\\
 & & \SCALEALG & 466 & -- & -- & 6 & 2.04 & 3.14 & 6.22 & 2.99 & 1 & 19 & 10.05\\
 & & \ILPALG & 466 & -- & -- & 2 & 1 & 1 & 1 & 1 & 8 & 21 & 12.8\\
\midrule
\multirow{4}{*}{Vienna2} & \multirow{4}{*}{Octi.} & \DYNALG & 468 & 324 & 1 & 8 & 2.63 & 1.03 & 4.01 & 1.52 & 1 & 16 & 8.12\\
 & & \GREEDYALG & 468 & 324 & 1 & 18 & 4.79 & 1.72 & 14.23 & 3.36 & 1 & 12 & 4.31\\
 & & \SCALEALG & 468 & -- & -- & 6 & 2.63 & 1.07 & 2.55 & 1.4 & 1 & 12 & 7.87\\
 & & \ILPALG & 468 & -- & -- & 4 & 1 & 1 & 1 & 1 & 5 & 16 & 10\\
\midrule
\multirow{4}{*}{Vienna3} & \multirow{4}{*}{Curved} & \DYNALG & 606 & 535 & 0 & 7 & 1 & 1 & 1 & 1 & 2 & 22 & 10.7\\
 & & \GREEDYALG & 606 & 535 & 0 & 11 & 0.98 & 1.47 & 1.92 & 1.06 & 1 & 17 & 6.6\\
 & & \SCALEALG & 606 & -- & -- & 13 & 1.18 & 3.2 & 2.18 & 1.34 & 1 & 16 & 7.43\\
 & & \ILPALG & 606 & -- & -- & 7 & 1 & 1 & 1 & 1 & 2 & 22 & 10.7\\
\bottomrule       
\end{tabular}
\end{table}

The experiments were performed on a single core of an Intel(R)
Core(TM) i7-3520M CPU processor. The machine is clocked at 2.9 GHz,
and has 4096 MB RAM. Our implementation is written in Java. For each
instance and each algorithm we conducted 100 runs and took the average
running times. Each time before we started the 100 runs, we performed
50 runs without measuring the running time in order to \emph{warm up}
the virtual machine (Java(TM) SE Runtime Environment, build
1.8.0$\_$60-b27, Oracle). We did this in order to measure the actual
running times of our algorithms and not to measure the time that the
virtual machine of Java spends for loading classes and optimizing byte
code.

\newcommand{\scaleLabelings}{0.47}
\begin{figure*}[t]
  \centering

  \subfigure[\emph{Sydney4}: \DYNALG, \CurvedLabels\newline {Layout from \cite[Fig.~1a.]{fink}.}]{
    \label{fig:labeligns:paper:sydney-curved-dp} \includegraphics[width=\scaleLabelings\textwidth]{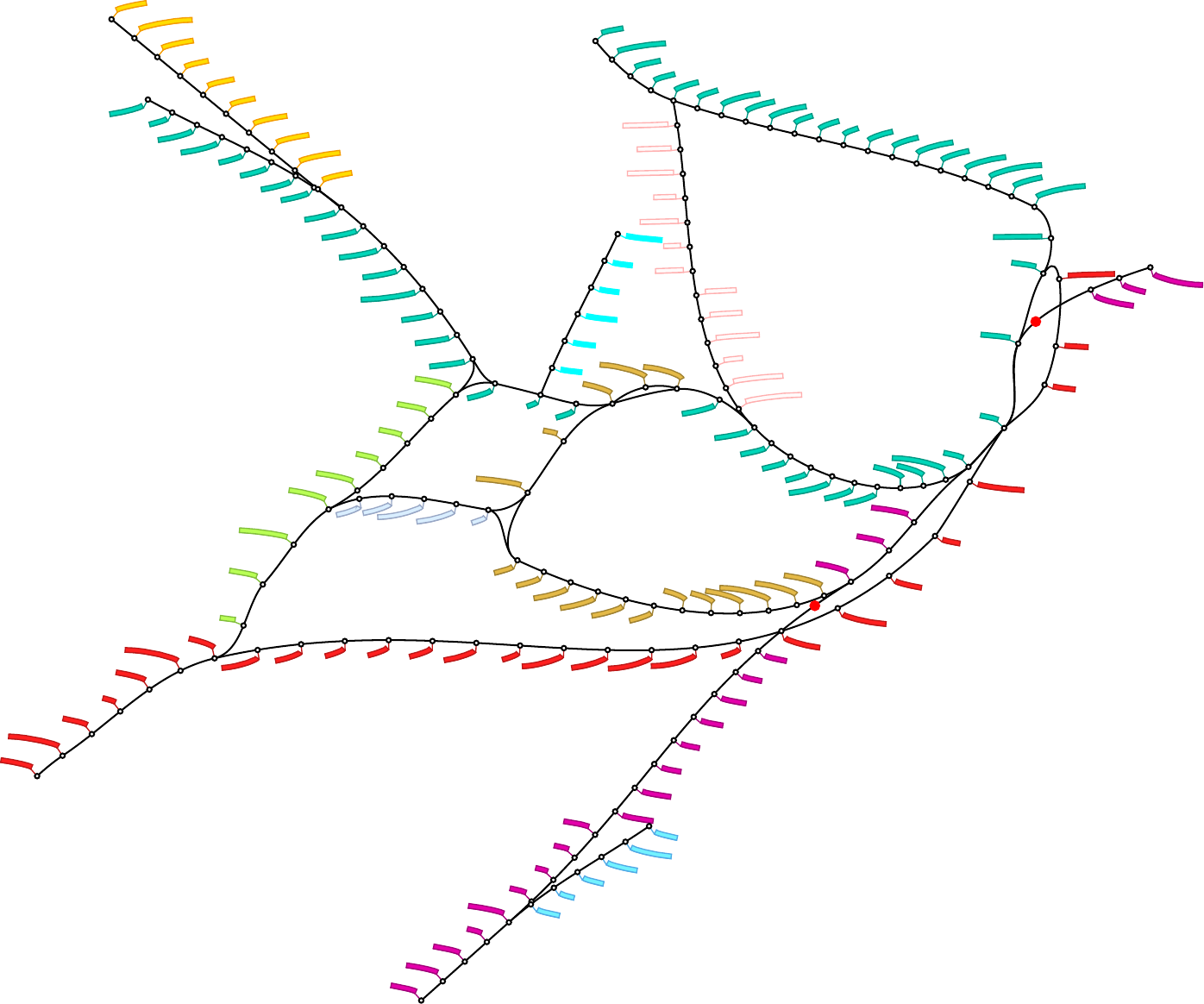}
  }
  \subfigure[\emph{Sydney4}: \GREEDYALG, \CurvedLabels\newline {Layout from \cite[Fig. 1a]{fink}.}]{
   \label{fig:labeligns:paper:sydney-curved-greedy} \includegraphics[width=\scaleLabelings\textwidth]{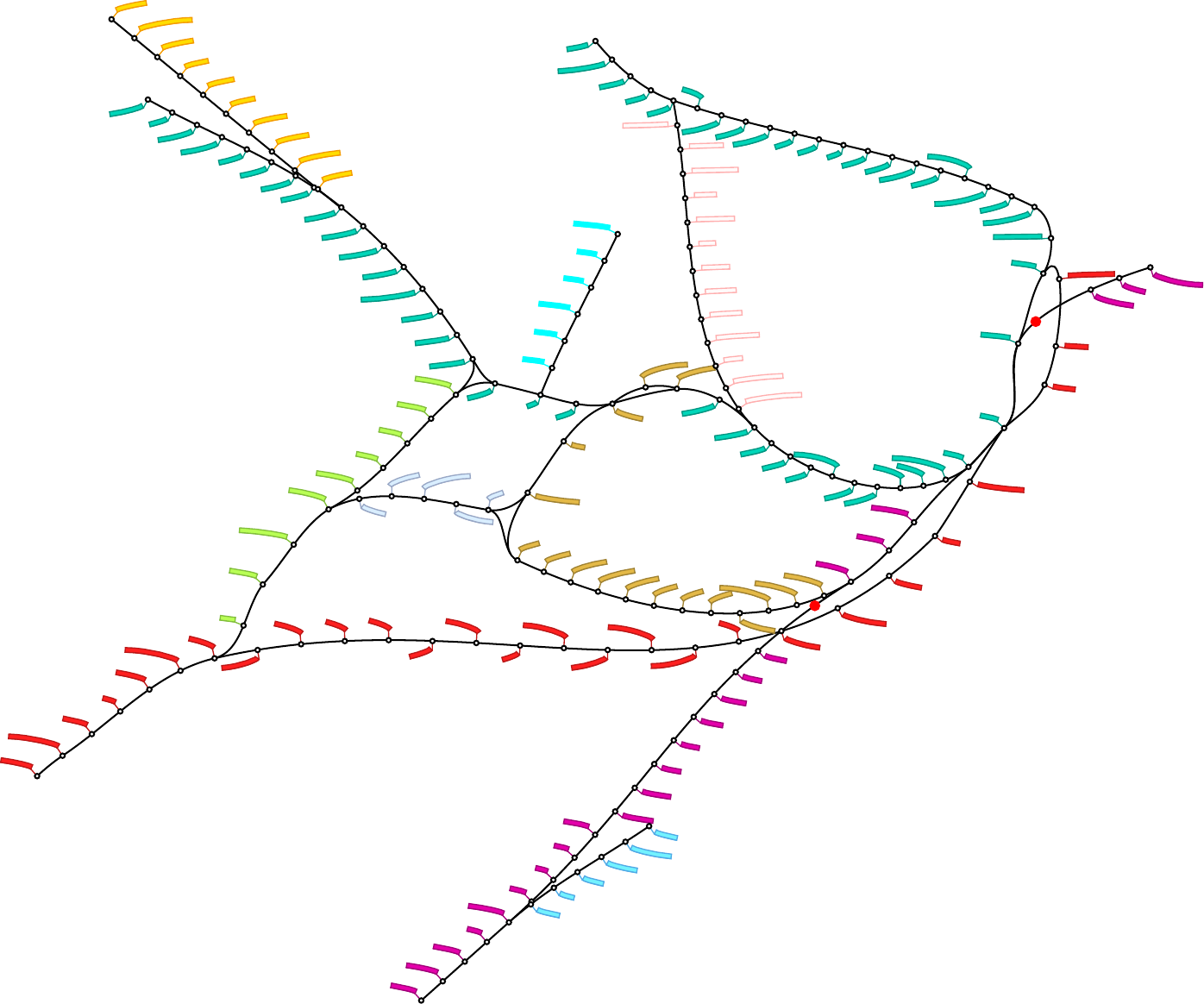}
  }

  \caption{Labelings for Sydney using \DYNALG and \GREEDYALG.}
  \label{fig:labeling:sydneyA}
\end{figure*}

Table~\ref{table:running-time} and Table~\ref{table:quality} present
our quantitative results for the considered instances. For
\emph{Sydney4} labelings are found in
Fig.~\ref{fig:labeling:sydneyA} and for \emph{Sydney2}, \emph{Sydney3} and \emph{Vienna2} labelings
created by \DYNALG are found in Fig.~\ref{fig:labeling:sydneyB}, Fig.~\ref{fig:labeling:comparison:complete}, Fig.~\ref{fig:labeling:comparison:vienna}, respectively.  The
labelings of all instances are found in the appendix.

We first note that with respect to the total number of created
candidates only few labels are removed for enforcing
Assumption~\ref{assumption:separation} and
Assumption~\ref{assumption:transitivity-property}; see
Table~\ref{table:quality},
\emph{A\ref{assumption:separation}\ref{assumption:transitivity-property}}. This
indicates that requiring those assumptions is not a real restriction
on a realistic set of candidates, even though they seem to be
artificial.
\begin{figure*}[t]
  \centering
  \subfigure[Original labeling by \cite{metroMap1}]{
    \includegraphics[width=\scaleLabelings\textwidth]{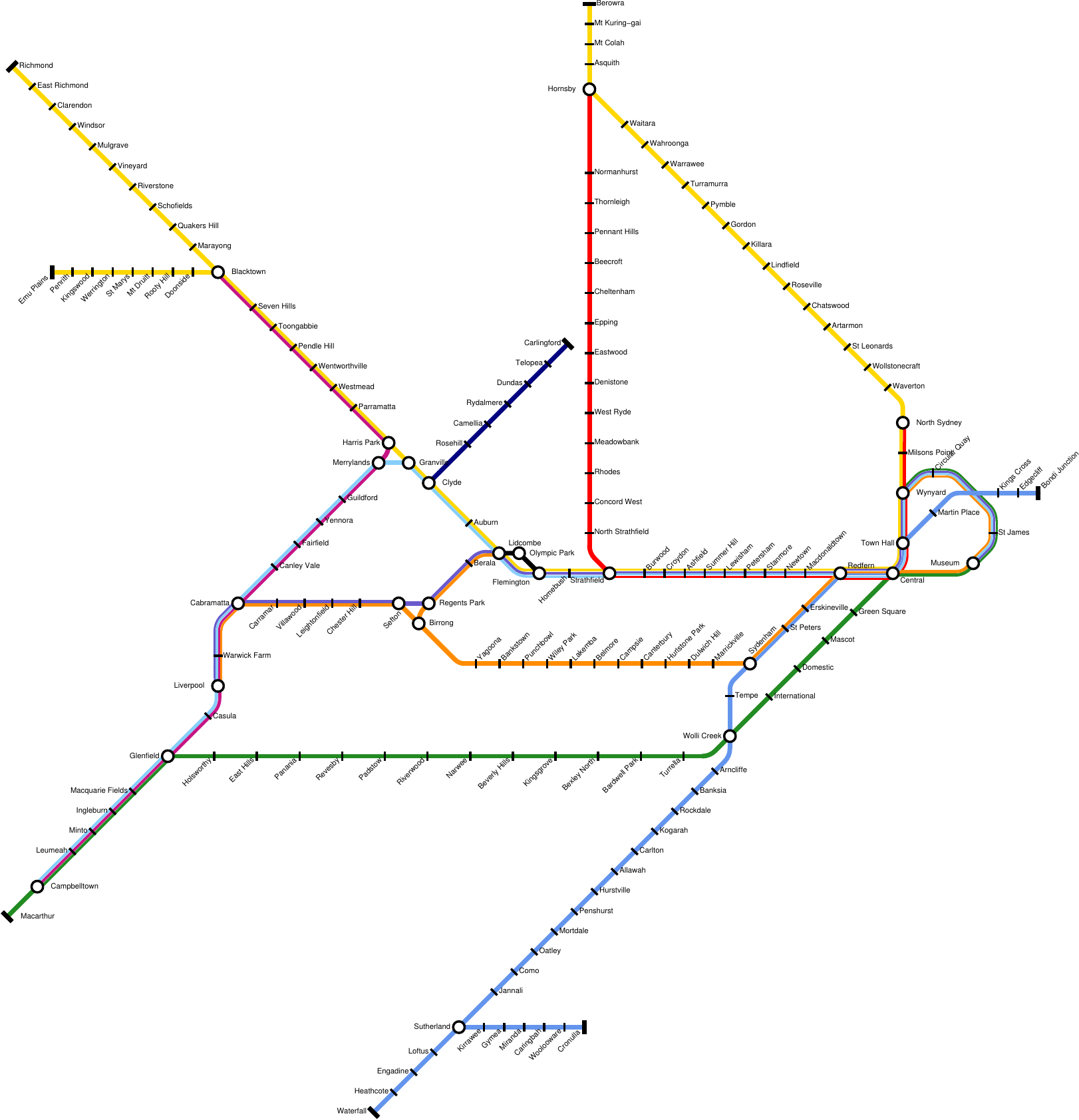}
    \label{fig:labeling:paper:sydney_original_labeling}
  }
  \subfigure[\emph{Sydney2}: \DYNALG, \OctiLabels\newline {Layout from \cite[Fig.~9b.]{metroMap1}.}]{
    \includegraphics[width=\scaleLabelings\textwidth]{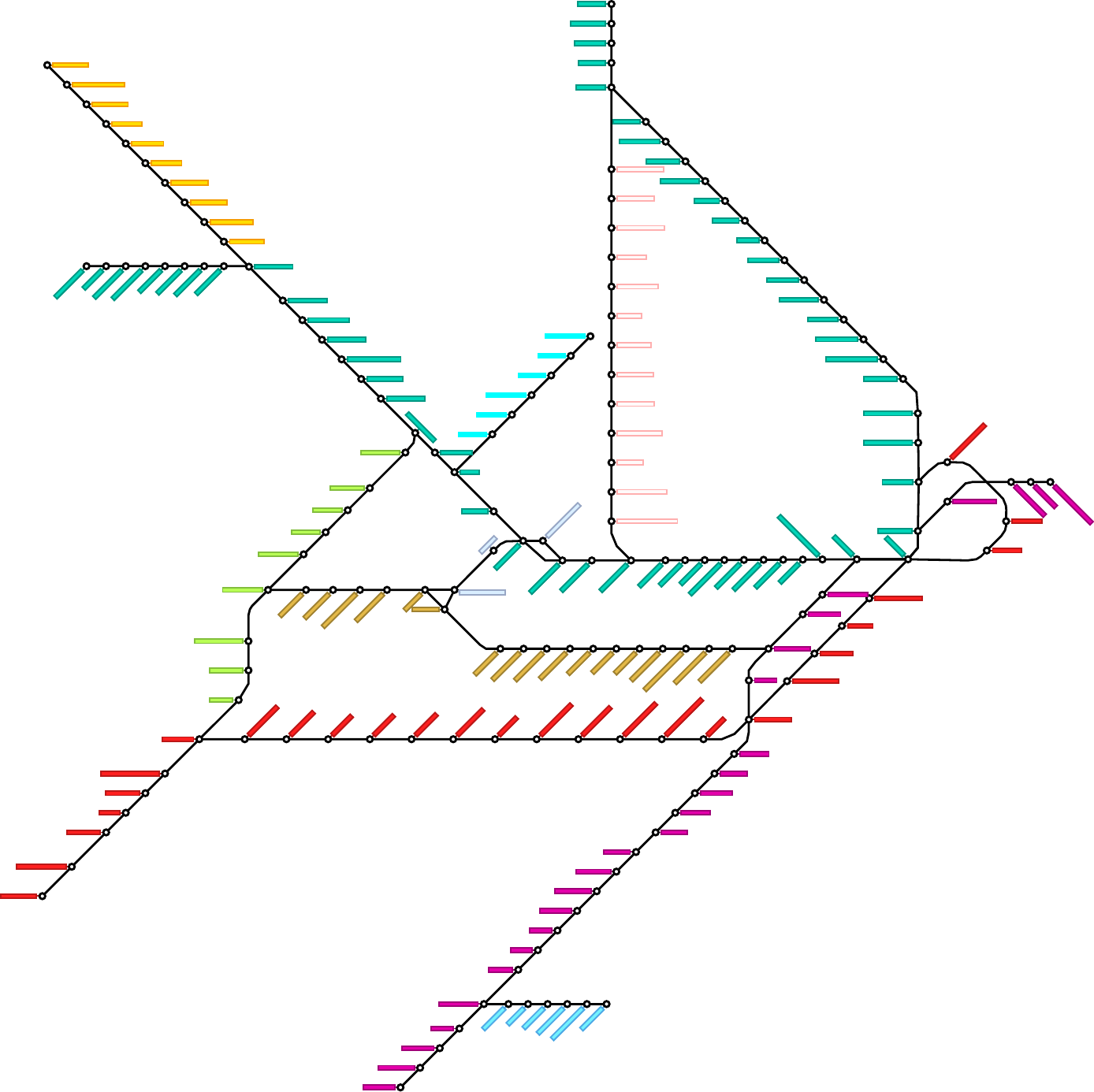}
    \label{fig:labeling:paper:sydney-curved-greedy}
  }
  \caption{Labelings for Sydney.}
  \label{fig:labeling:sydneyB}
\end{figure*}

\textbf{Running time.}  Even for large networks as Sydney, our
algorithm \DYNALG needs less than 0.6 seconds; see
Table~\ref{table:running-time}. This shows that our approach is
applicable for scenarios in which the map designer wants to adapt the
layout and its labeling interactively. In particular in those
scenarios not every of the four steps must be repeated each time,
which improves computing time. For example, after once applying the
scaling step (1st phase, 2nd step -- the most time consuming step),
the instance does not need to be rescaled again, but the relation
between label size and map size is determined.  Further,~\DYNALG is
only moderately slower than \GREEDYALG; $0.21$ seconds in maximum, see
Table~\ref{table:running-time}, \emph{Sydney4}. On the other hand, the
approaches \ILPALG and \SCALEALG are not alternatives, because their running times are
much worse; over 1 minute in maximum; see
Table~\ref{table:running-time}, \emph{Sydney4}.

\textbf{Quality.}  We observe that in all labelings created by \DYNALG
there are only few switchovers, namely 4--8; see
Table~\ref{table:quality}, column SO. Hence, there are long
\emph{sequences} of consecutive labels that lie on the same side of
their metro line; see corresponding figures and
Table~\ref{table:quality}, column Sequence. Together with the ILP
based approach \ILPALG, it yields the solution with the longest
sequences in average. In particular the switchovers are placed such
that those sequences are regularly sized. The labels of a single
sequence are mostly directed into the same~$x$-direction and in
particular they are similarly shaped so that those sequences of labels
form regular patterns as desired. The alignment of the labels is
chosen so that they blend in with the alignment of their adjacent
labels. In comparison with the solution of \ILPALG, the costs of
\DYNALG never exceed a factor of 1.52; see Table~\ref{table:quality},
column $\frac{w(\mathcal L_\DYNALG)}{w(\mathcal L)}$. For the
instances \emph{Sydney4} and \emph{Vienna3} it even obtains a solution
with the same costs. For the other instances, \DYNALG basically spends
its additional costs on the choice of the single labels
($\frac{w_1(\mathcal L_\DYNALG)}{w_1(\mathcal L)}$) and the distance
of switchovers ($\frac{w_3(\mathcal L_\DYNALG)}{w_3(\mathcal L)}$).

\begin{figure*}[t]
  \centering \subfigure[]{
    \includegraphics[scale=0.73]{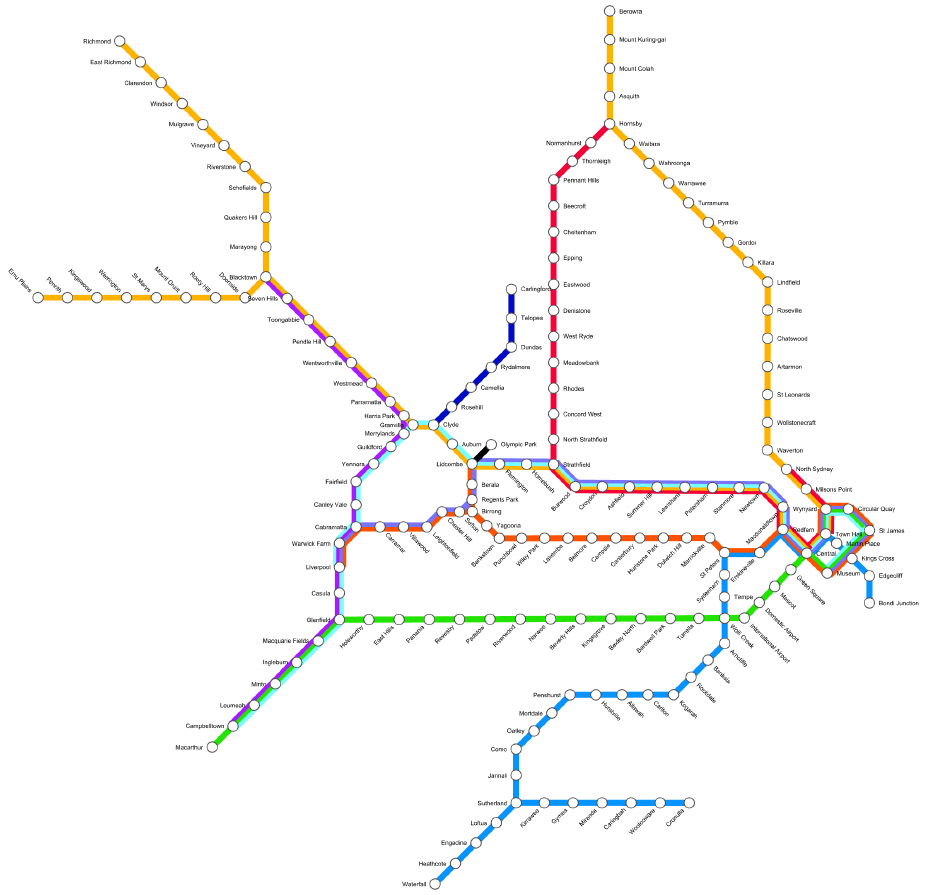}
    \label{fig:labeling:comparison:wang:complete}
  }
  \subfigure[]{
    \includegraphics[scale=0.48]{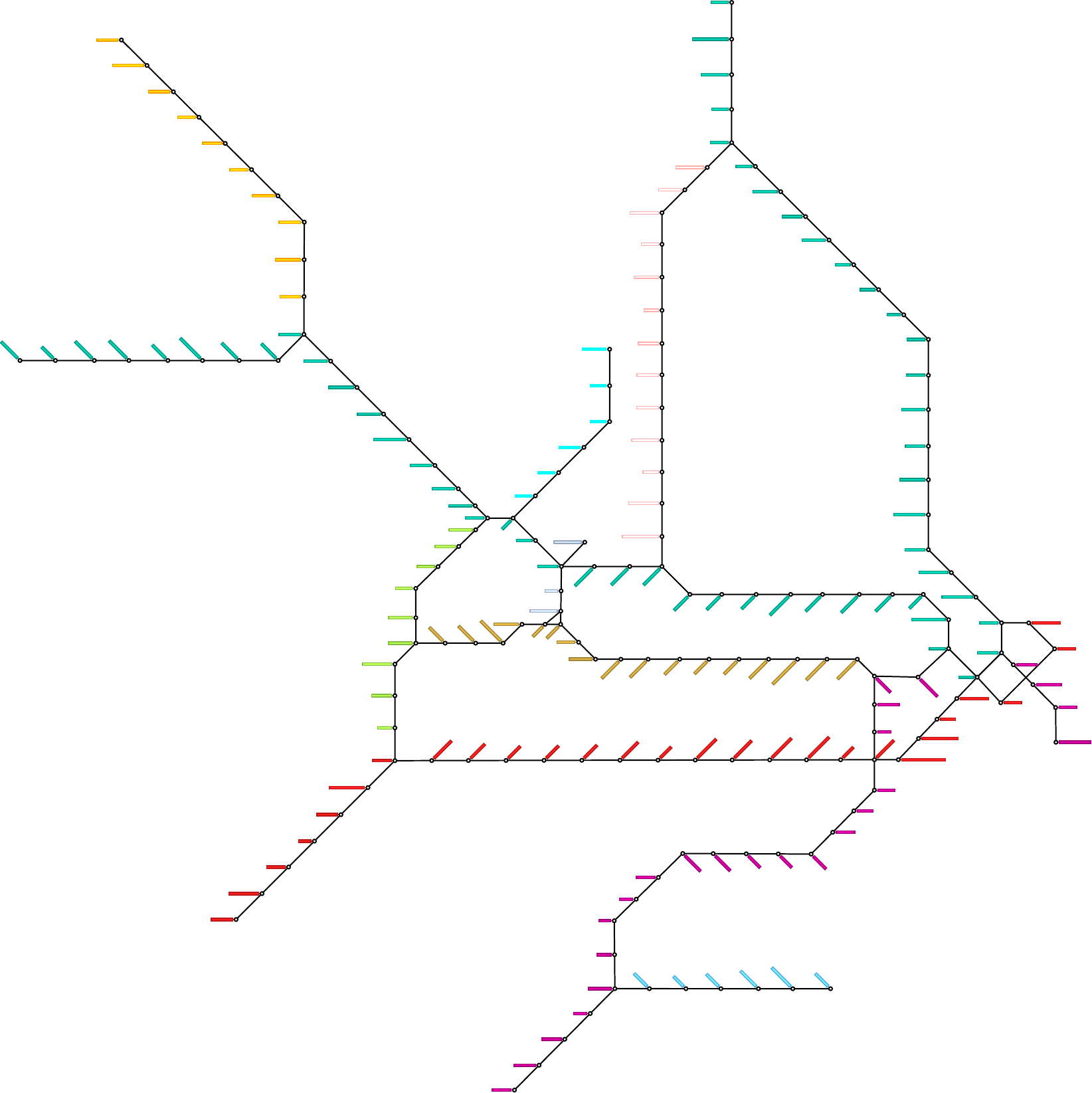}
    \label{fig:labeling:comparison:dp:complete}
  }
  \caption{Labelings for Sydney. \subref{fig:labeling:comparison:wang:complete} Original Labeling presented by Wang and Chi~\cite{wang2011} \subref{fig:labeling:comparison:dp:complete} \emph{Sydney3}: \DYNALG.}
\label{fig:labeling:comparison:complete}
\end{figure*}

\begin{figure}[!ht]
  \centering \subfigure[]{
   \includegraphics[scale=0.72]{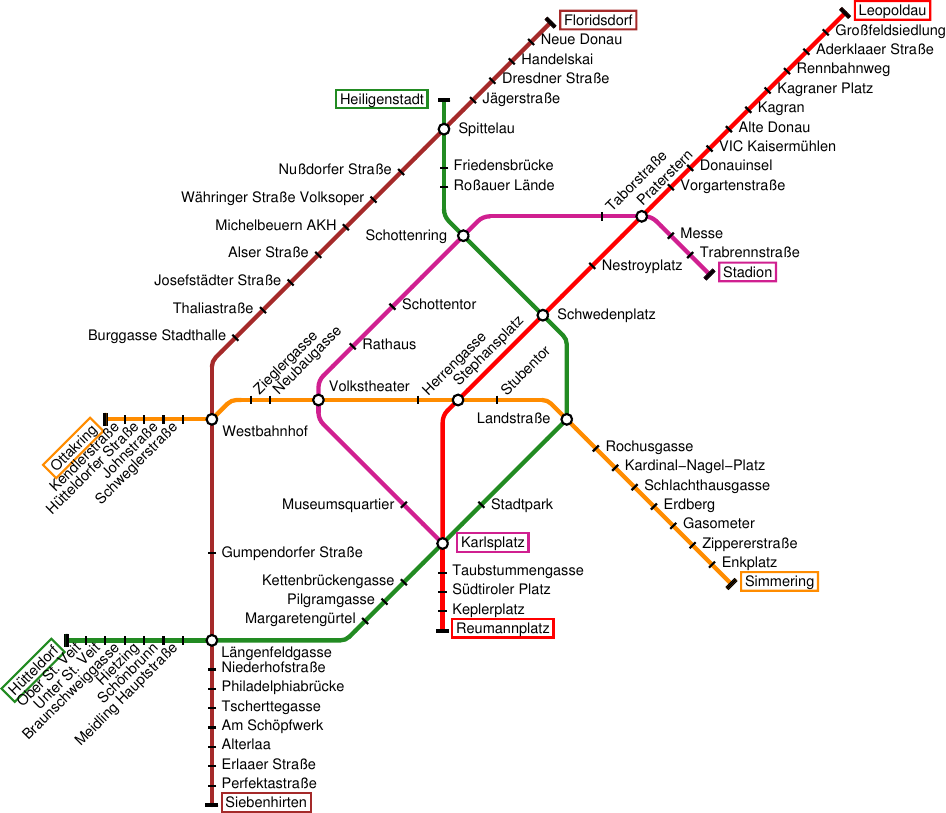}
    \label{fig:labeling:comparison:vienna:original}
  }
  \subfigure[]{
    \includegraphics[scale=0.5]{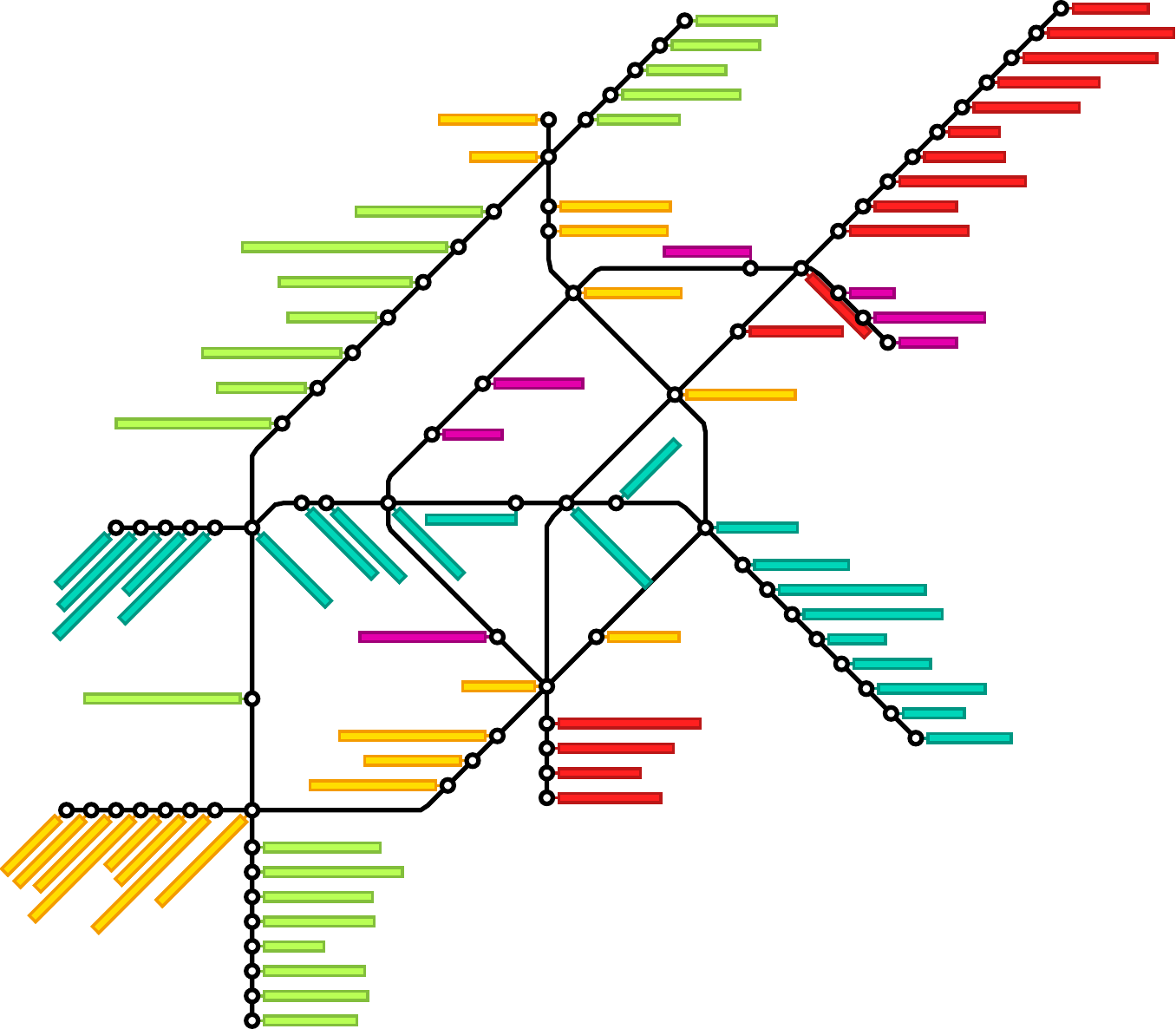}
    \label{fig:labeling:comparison:vienna:dp}
  }
  \caption{Labelings for Vienna. \subref{fig:labeling:comparison:vienna:original} Original Labeling presented by Nöllenburg and Wolff~\cite{metroMap1} \subref{fig:labeling:comparison:vienna:dp} \emph{Vienna2}: \DYNALG.}
\label{fig:labeling:comparison:vienna}
\end{figure}

In contrast, \GREEDYALG yields significantly more switchovers; in
maximum 20 switchovers more than \DYNALG, see \emph{Sydney4} . Consequently, there
are many distracting switches of labels from one side to the other of
the metro line; e.g. see Fig.~\ref{fig:labeling:sydneyA}. Although the
sequences of consecutive labels lying on the same side may be longer
in maximum compared to \DYNALG, they are much shorter in average; see
Table~\ref{table:quality}, column Sequence. Further, several adjacent
labels point in opposite $x$-directions, which results in distracting
effects; see corresponding figures. Altogether, the labelings that are
obtained by \GREEDYALG do not look regular, but cluttered. \DYNALG
solves those problems since it considers the metro line
\emph{globally} yielding an optimal labeling for a single line.  This
observation is also reflected in Table~\ref{table:quality}, column
$\frac{w(\mathcal L_\DYNALG)}{w(\mathcal L)}$, which shows that the
costs computed by \GREEDYALG are significantly larger than the costs
computed by~\DYNALG. In particular costs for positioning the
switchovers are much worse; \emph{Sydney2}: $\nicefrac{w_3(\mathcal
  L_\GREEDYALG)}{w_3(\mathcal L)}=19.52$, $\nicefrac{w_3(\mathcal
  L_\DYNALG)}{w_3(\mathcal L)}=3.25$ and \emph{Vienna2}:
$\nicefrac{w_3(\mathcal L_\GREEDYALG)}{w_3(\mathcal L)}=14.23$,
$\nicefrac{w_3(\mathcal L_\DYNALG)}{w_3(\mathcal L)}=4.01$.  Hence, the
better quality of \DYNALG prevails the slightly better running time of
\GREEDYALG.

Concerning the computed scale factor in the first phase of \DYNALG,
the labels are smaller than those produced by \SCALEALG by a factor of
0.54--0.81; see Table~\ref{table:instances}. While this seems to be a
drawback on the first sight, the smaller size provides necessary space
that is used to obtain a labeling of higher quality with respect to
the number and the placement of switchovers. Hence, the solutions
of \SCALEALG have more switchovers (except for \emph{Vienna2}) and
shorter sequences of labels lying on the same side in average than
\DYNALG; see Table~\ref{table:quality}, column SO and Sequence.

We observe that both Nöllenburg and Wolff's and our labelings of
Sydney look quite similar, whereas our labeling has less switchovers;
see Fig.~\ref{fig:labeling:sydneyB}. The same applies for the
labelings of the layout of Vienna; see
Fig.~\ref{fig:labeling:comparison:vienna}.  Recall that their approach
needed more than 10 hours to compute a labeled metro map of
Sydney. Since they need only up to 23 minutes to compute the layout
without labeling, it lends itself to first apply their approach to
gain a layout and then to apply our approach to construct a
corresponding labeling.

\begin{figure}[t]
  \centering \subfigure[]{
    \includegraphics[width=0.5\textwidth]{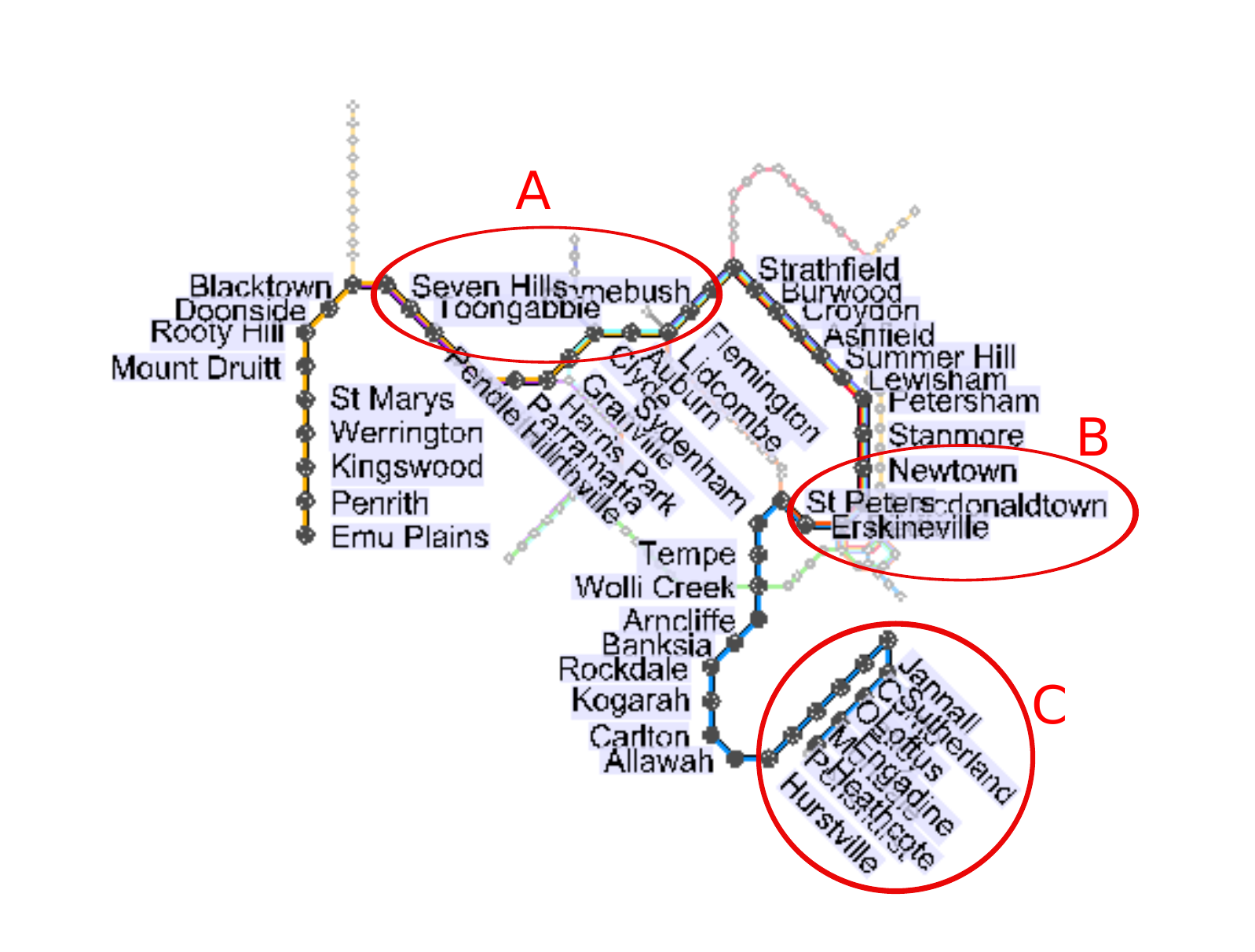}
    \label{fig:labeling:comparison:wang}
  }
  \subfigure[]{
\includegraphics[width=0.4\textwidth]{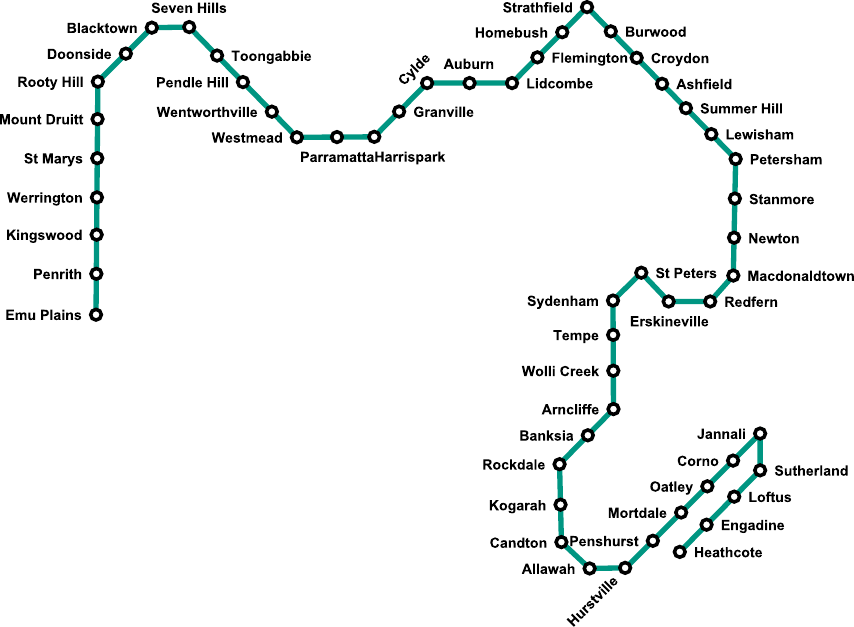}
    \label{fig:labeling:comparison:dp}
  }
  \caption{Comparison of two labelings for the same line. \subref{fig:labeling:comparison:wang} Labeling is created by the tool presented by Wang and Chi\cite{wang2011}. \subref{fig:labeling:comparison:dp} Labeling is created by \DYNALG.}
\label{fig:labeling:comparison}
\end{figure}

Wang and Chi present in their paper~\cite{wang2011} an approach that
is divided into two phases. In the first phase they compute the layout
of the metro map and then in the second phase they create a labeling
for that layout. For both steps they formulate energy functions
expressing their desired objectives, which then are locally
optimized. Figure~\ref{fig:labeling:comparison:wang:complete} shows
the metro map of Sydney created by their approach. In comparison,
Fig.~\ref{fig:labeling:comparison:dp:complete} shows the same layout
with a labeling created by our approach. Both labelings look quite
similar. While our approach needed 0.26s (see
Table~\ref{table:running-time}, \emph{Sydney3}), their approach needed
less than 0.1s on their machine. However, their approach does not
guarantee that the labels are occlusion-free, but labels may overlap
with metro lines and other labels. This may result in illegible
drawings of metro maps.  For example
Fig.~\ref{fig:labeling:comparison} shows two labelings of a metro line
of Sydney that has been laid out by the tool of Wang and
Chi. Figure~\ref{fig:labeling:comparison:wang} shows a labeling that
has been created by their tool, while
Fig.~\ref{fig:labeling:comparison:dp} shows a labeling that has been
created by our approach. The labeling of Wang and Chi has several
serious defects that makes the map hardly readable.
The marked regions A, B and C show labels that overlap each
other. Hence, some of the labels are obscured partly, while some of
the labels are completely covered by other labels. For example in
region B the label \emph{St. Peters} and the label \emph{Erskinville}
overlap the label \emph{Macdonaldtown} such that it is hardly
viewable. Further, region~C contains two diagonal rows of stops
aligned parallel. While the upper row is visible, the lower row is
almost completely covered by labels. Further, the labels of the upper
row obscure the labels of the lower row. In contrast our approach
yields an occlusion-free labeling, such that each label and each stop
is easily legible. We therefore think that our approach is a
reasonable alternative for the labeling step of Wang and Chi's
approach. In particular, we think that the better quality of our
approach prevails the better running time of Wang and Chi's approach.

In conclusion our workflow is a reasonable alternative and improvement
for the approaches presented both by Nöllenburg and Wolff, and by Wang
and Chi. In the former case, our approach is significantly faster,
while in contrast to the latter case we can guarantee occlusion-free
labelings.

\paragraph{Acknowledgment.} We sincerely thank Herman Haverkort,
Arlind Nocaj, Aidan Slingsby and Jo Wood for helpful and interesting
discussions.

\bibliographystyle{abbrv}

\newpage
\appendix

\section{Labelings}

 \newcommand{\scaleA}{0.48}
  \begin{figure}[!h]
    \centering 
    \subfigure[ \DYNALG]{
\includegraphics[width=\scaleA\textwidth]{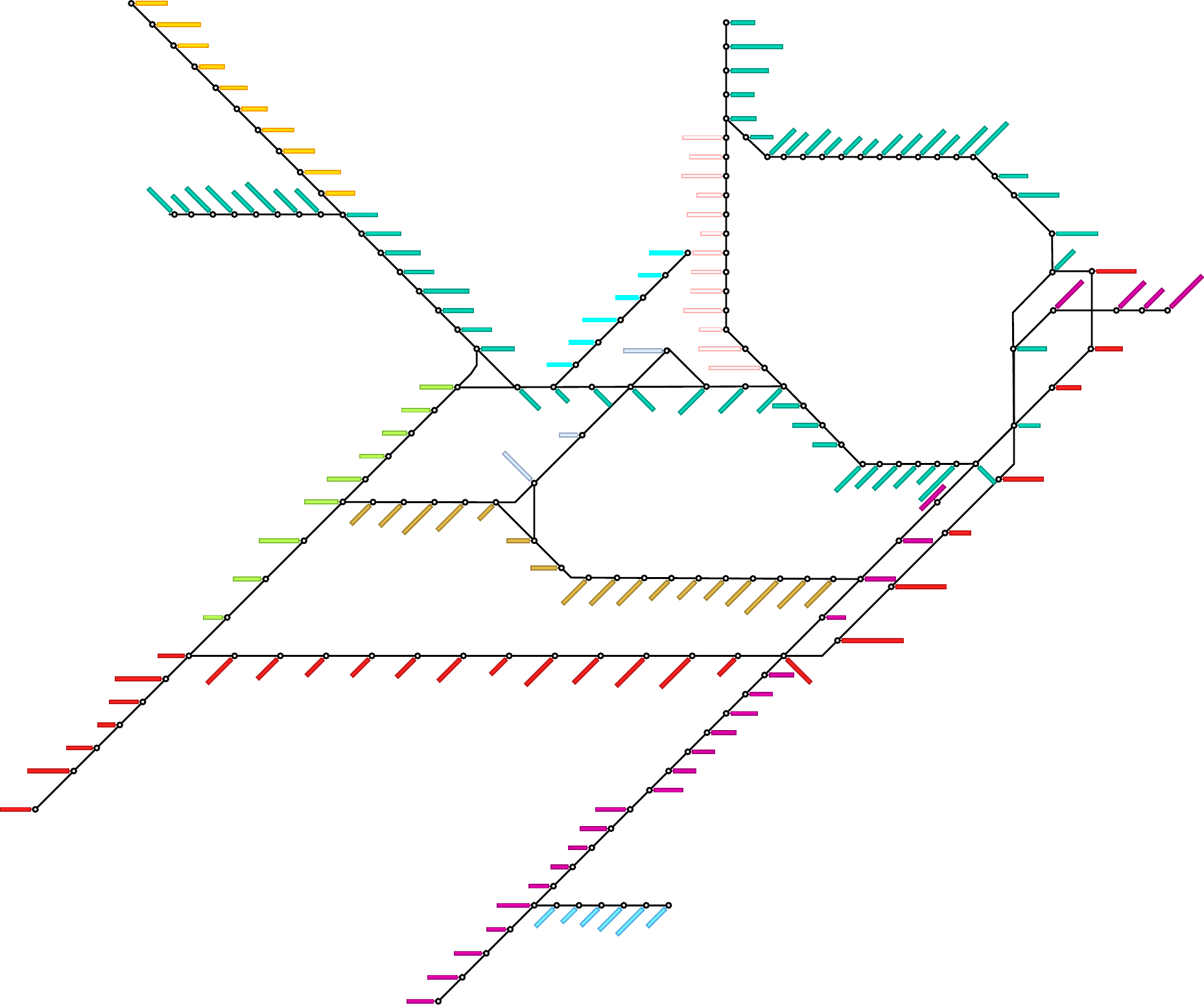}
      \label{fig:labelings:sydney1_dp}
   }
    \subfigure[ \GREEDYALG]{
\includegraphics[width=\scaleA\textwidth]{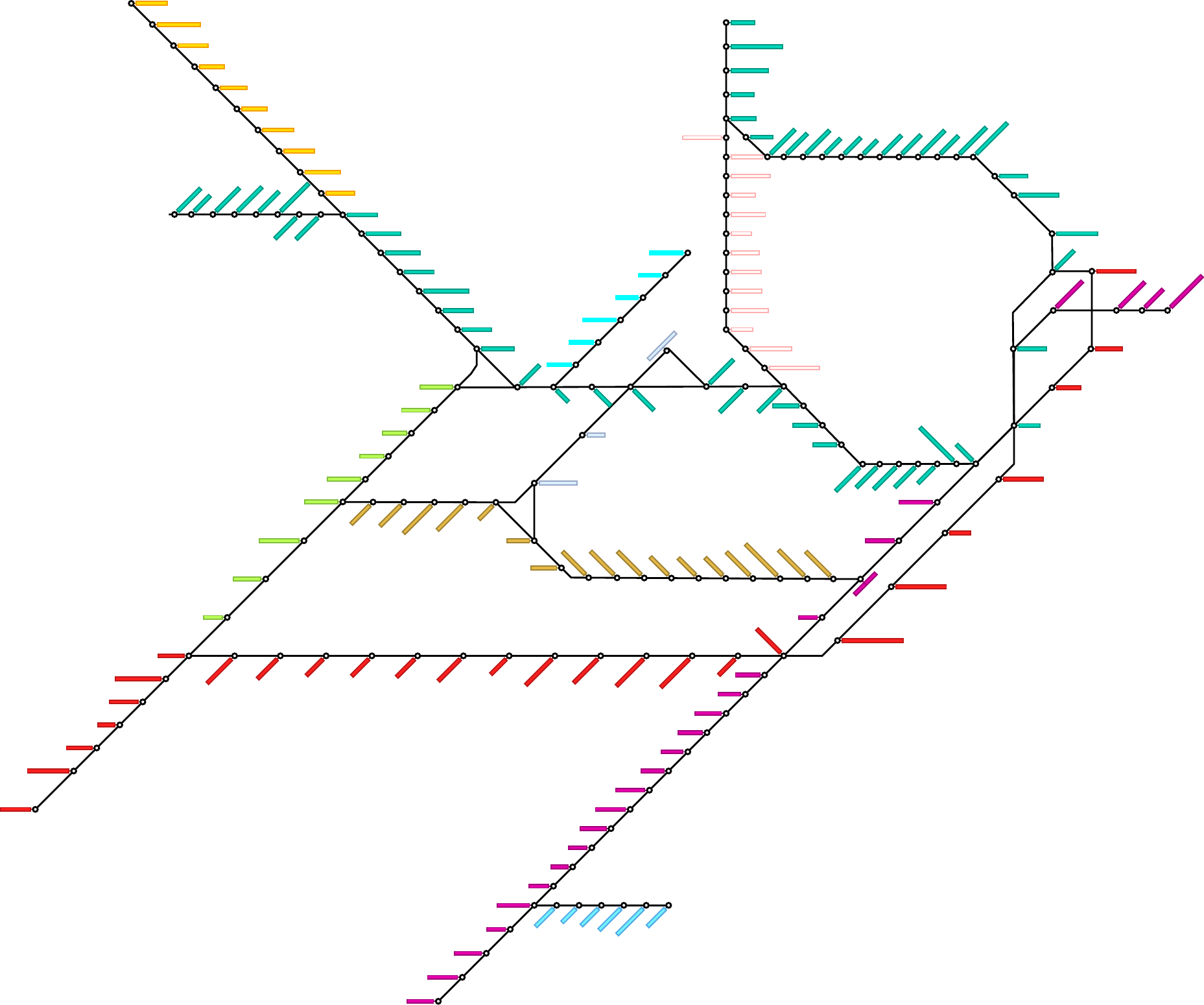}
      \label{fig:labelings:sydney1_g}
   } 

    \subfigure[ \SCALEALG]{
\includegraphics[width=\scaleA\textwidth]{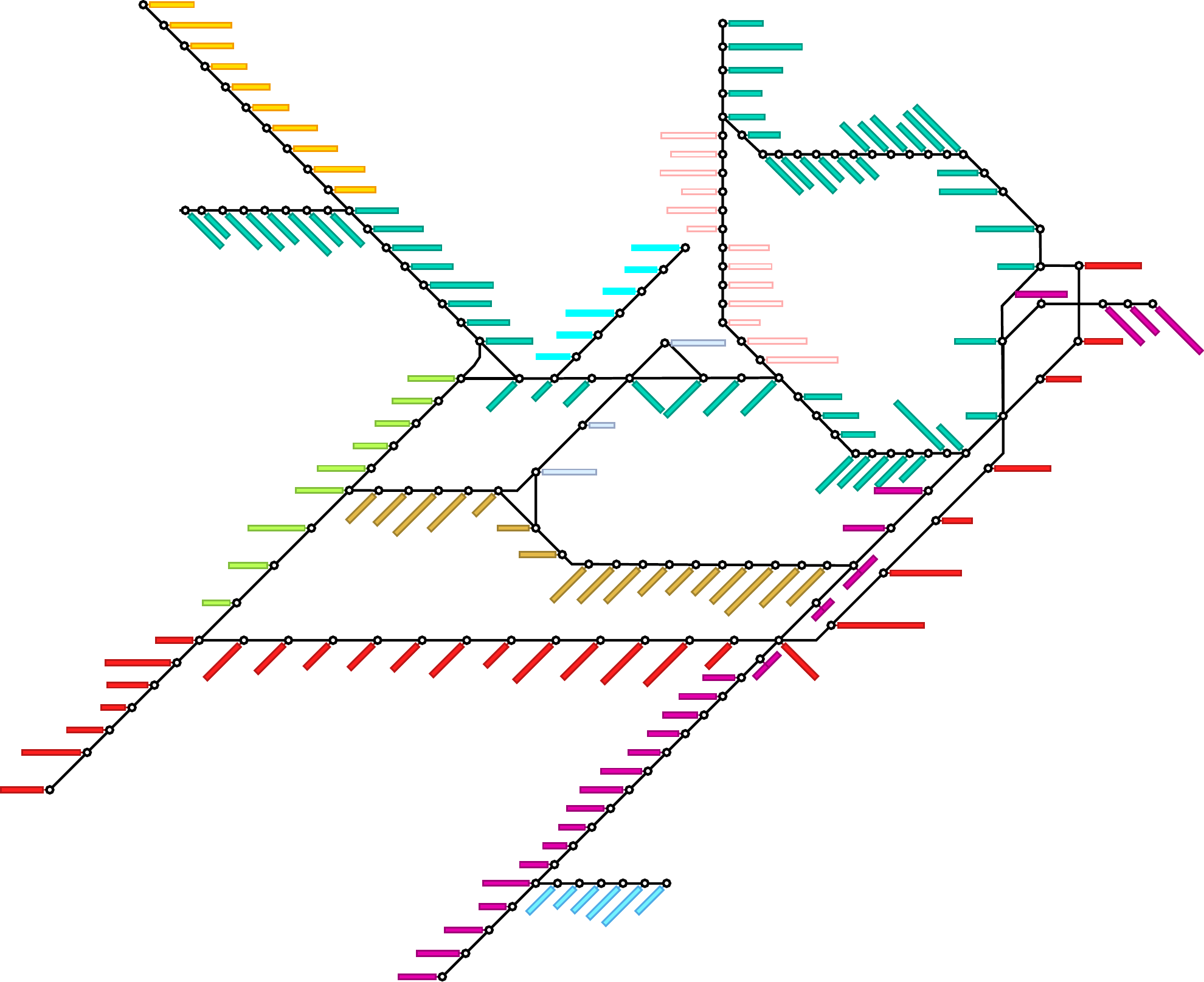}
      \label{fig:labelings:sydney1_scale}
   }
    \subfigure[ \ILPALG]{
\includegraphics[width=\scaleA\textwidth]{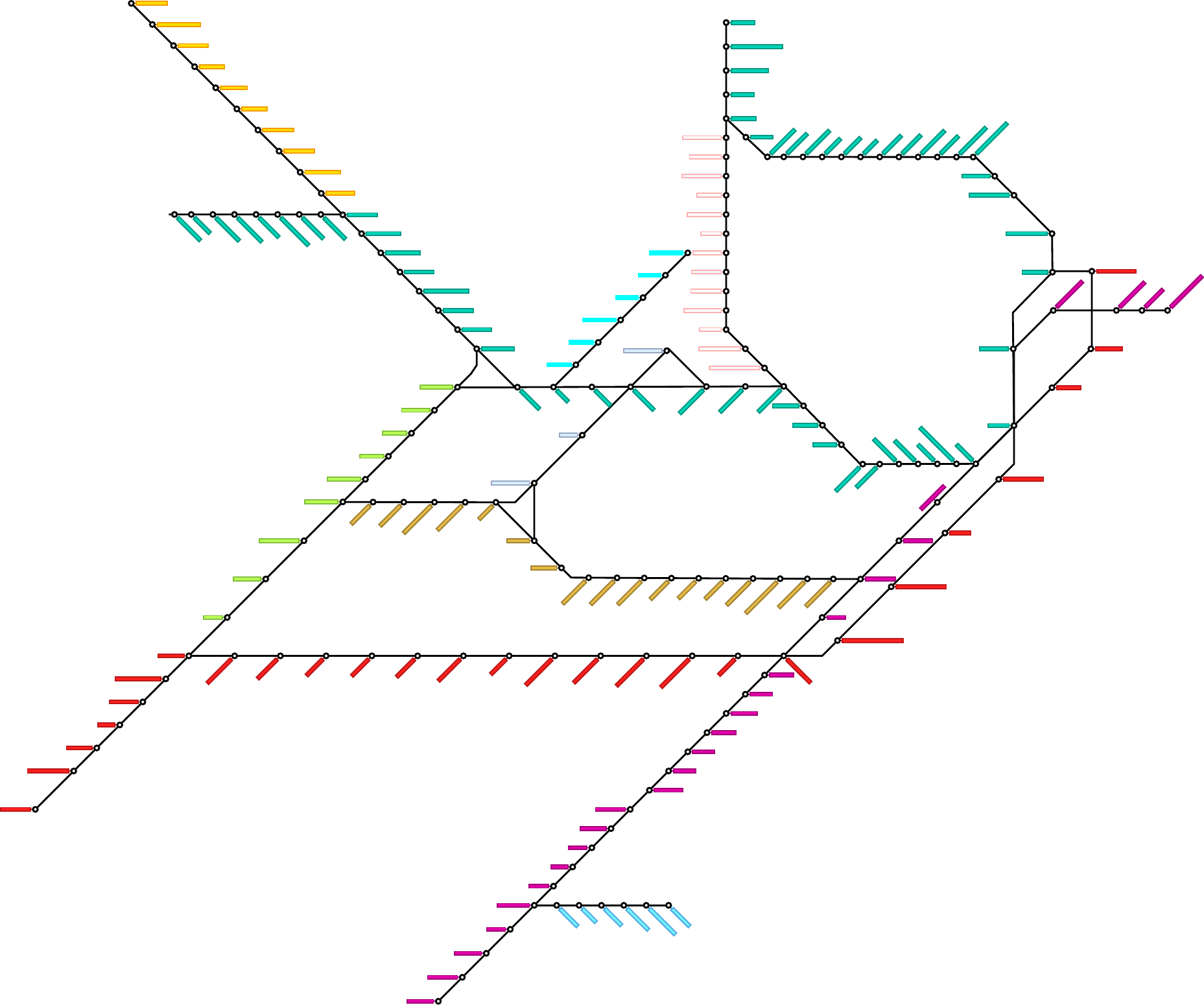}
      \label{fig:labelings:sydney1_ilp}
   } 
 \caption{Labelings for instance \emph{Sydney1}.}
 \end{figure}

 \newcommand{\scaleB}{0.48}
  \begin{figure}[!h]
    \centering 
    \subfigure[ \DYNALG]{
\includegraphics[width=\scaleB\textwidth]{octi_sydney_nw_fig9b_dp}
      \label{fig:labelings:sydney2_dp}
   }
    \subfigure[ \GREEDYALG]{
\includegraphics[width=\scaleB\textwidth]{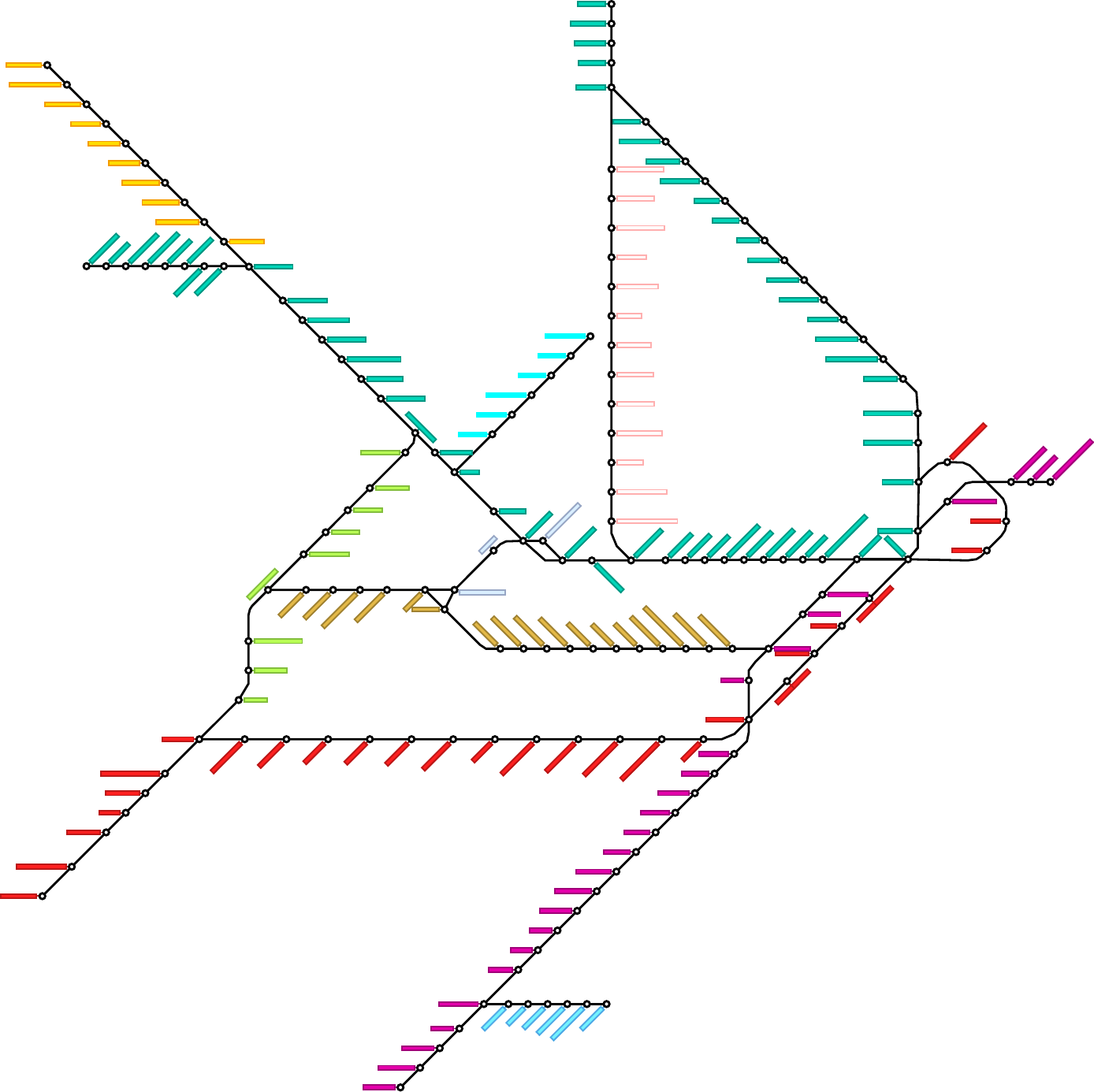}
      \label{fig:labelings:sydney2_g}
   } 

    \subfigure[ \SCALEALG]{
\includegraphics[width=\scaleB\textwidth]{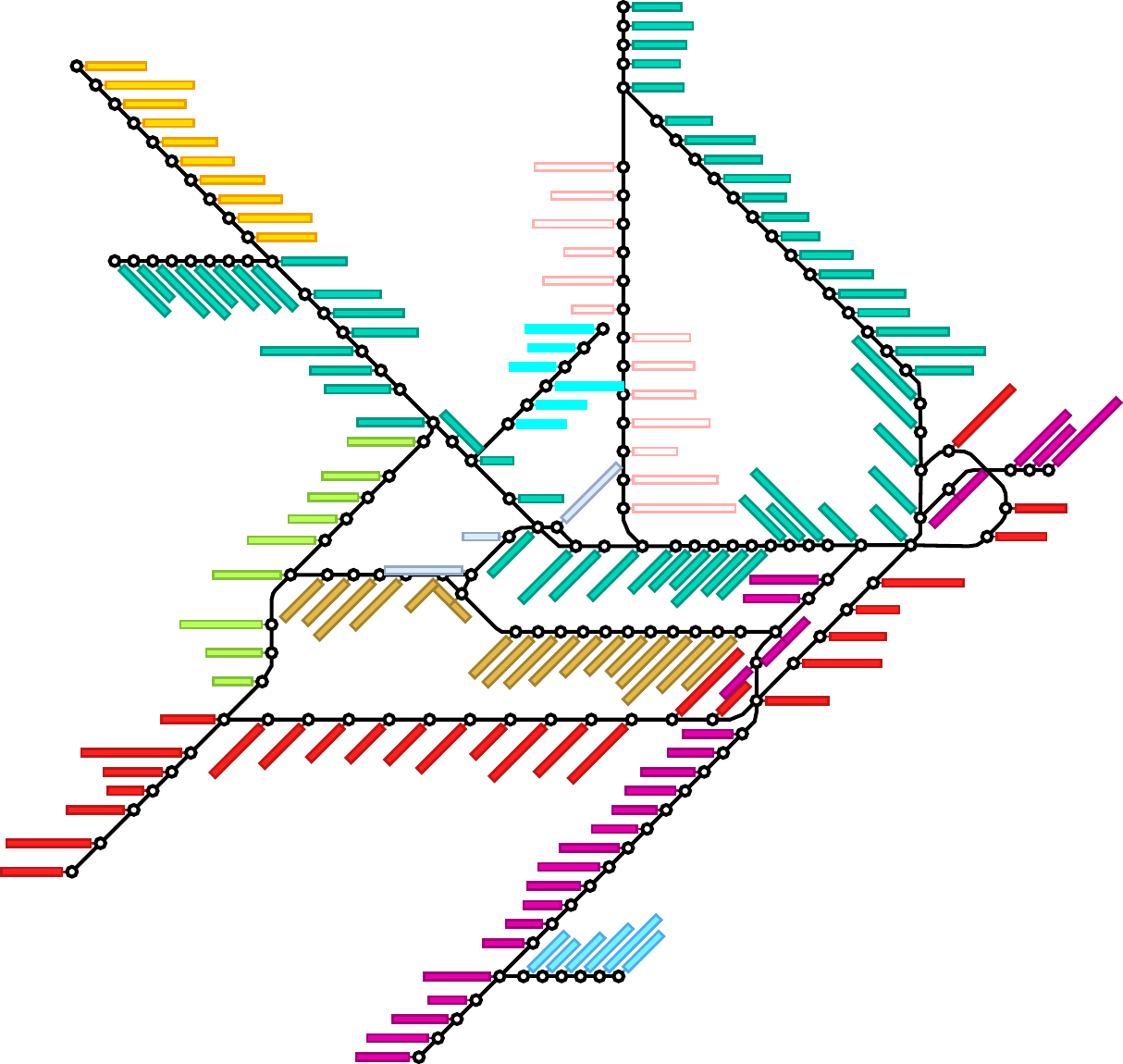}
      \label{fig:labelings:sydney2_scale}
   }
    \subfigure[ \ILPALG]{
\includegraphics[width=\scaleB\textwidth]{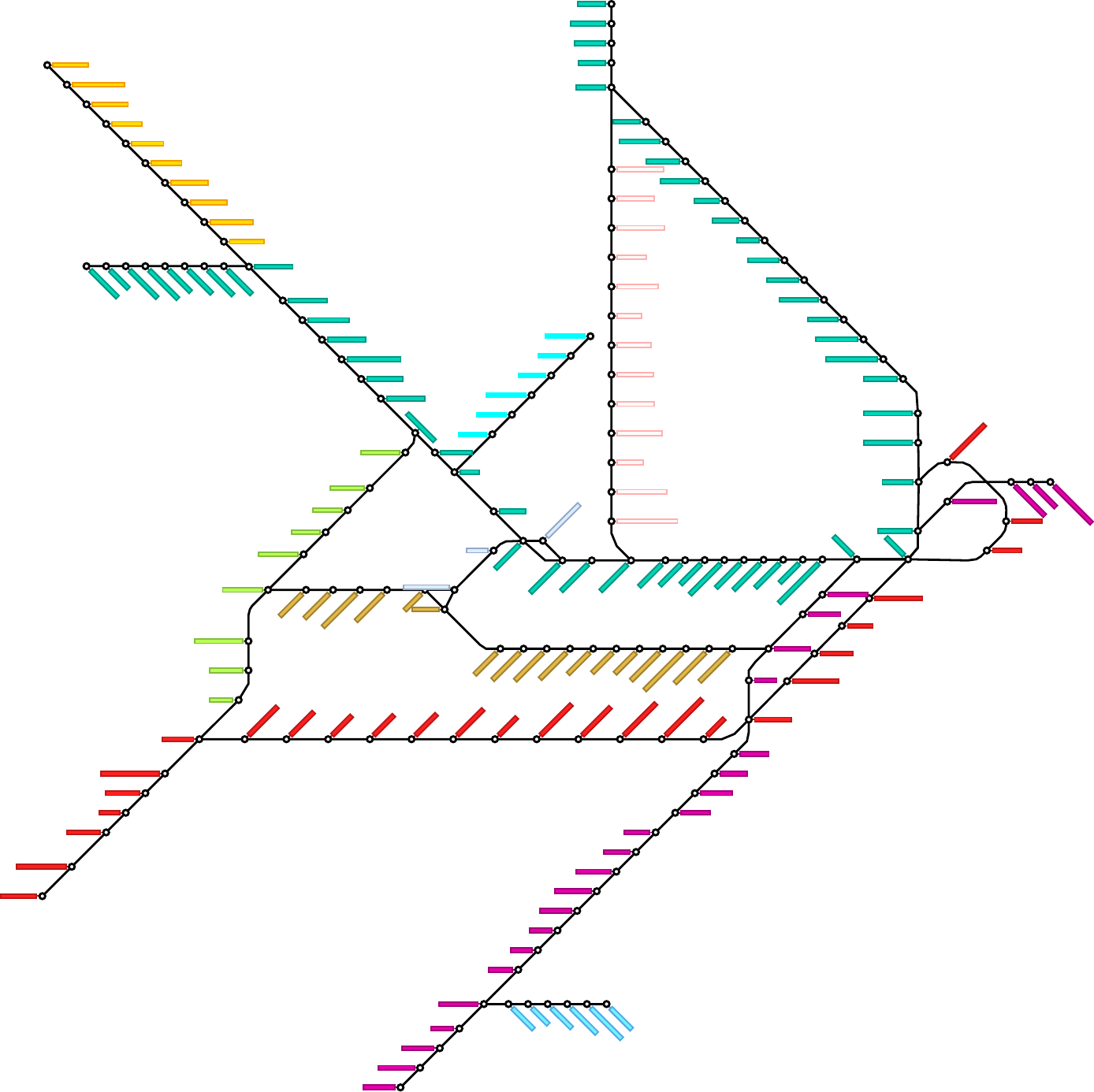}
      \label{fig:labelings:sydney2_ilp}
   } 
 \caption{Labelings for instance \emph{Sydney2}.}
 \end{figure}

 \newcommand{\scaleC}{0.48}
  \begin{figure}[!h]
    \centering 
    \subfigure[ \DYNALG]{
\includegraphics[width=\scaleC\textwidth]{octi_sydney_wang_dp}
      \label{fig:labelings:sydney3_dp}
   }
    \subfigure[ \GREEDYALG]{
\includegraphics[width=\scaleC\textwidth]{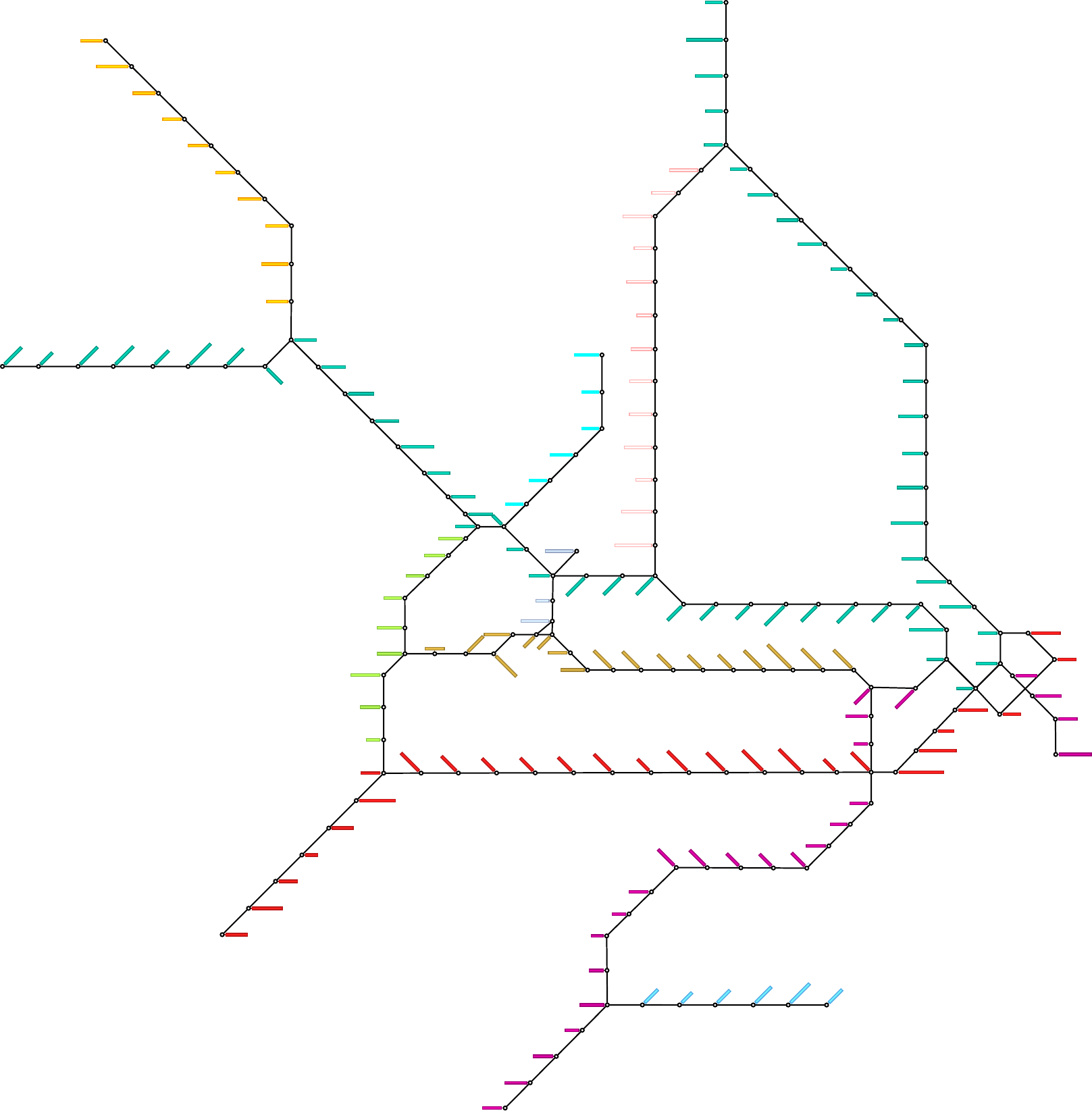}
      \label{fig:labelings:sydney3_g}
   } 

    \subfigure[ \SCALEALG]{
\includegraphics[width=\scaleC\textwidth]{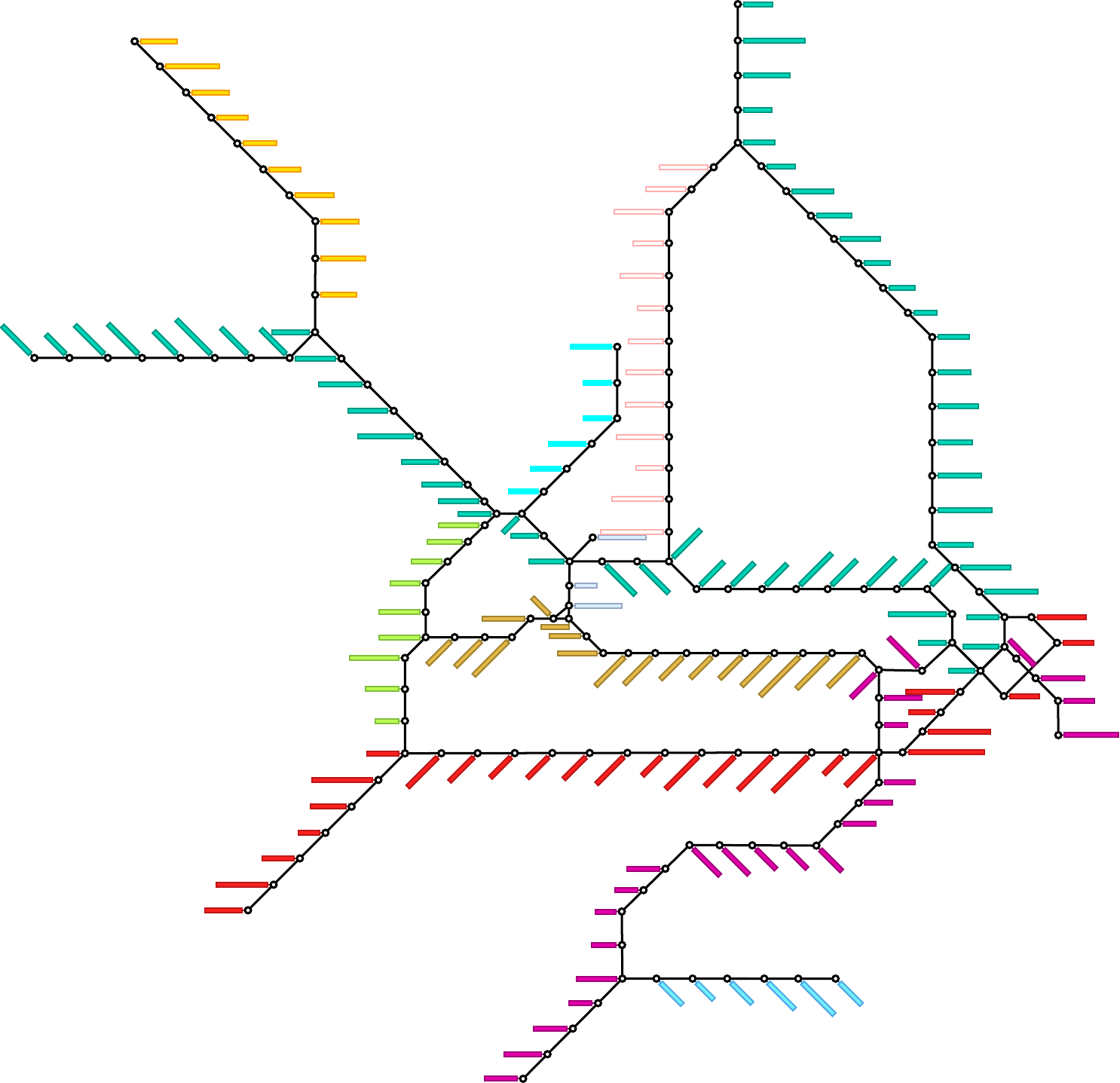}
      \label{fig:labelings:sydney3_scale}
   }
    \subfigure[ \ILPALG]{
\includegraphics[width=\scaleC\textwidth]{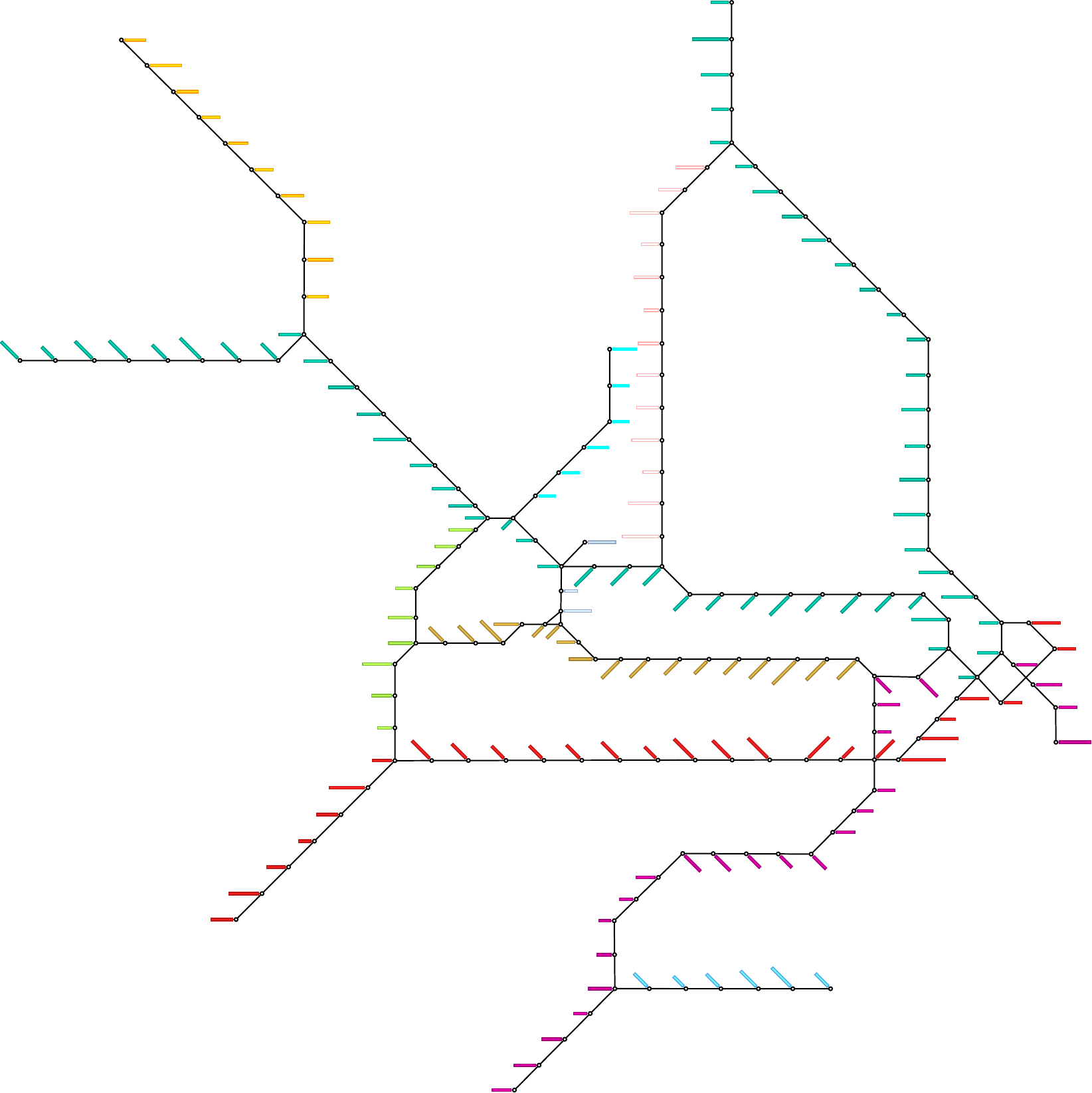}
      \label{fig:labelings:sydney3_ilp}
   } 
 \caption{Labelings for instance \emph{Sydney3}.}
 \end{figure}

 \newcommand{\scaleD}{0.48}
  \begin{figure}[!h]
    \centering 
    \subfigure[ \DYNALG]{
\includegraphics[width=\scaleD\textwidth]{curved_sydney_fink_dp}
      \label{fig:labelings:sydney4_dp}
   }
    \subfigure[ \GREEDYALG]{
\includegraphics[width=\scaleD\textwidth]{curved_sydney_fink_g}
      \label{fig:labelings:sydney4_g}
   } 

    \subfigure[ \SCALEALG]{
\includegraphics[width=\scaleD\textwidth]{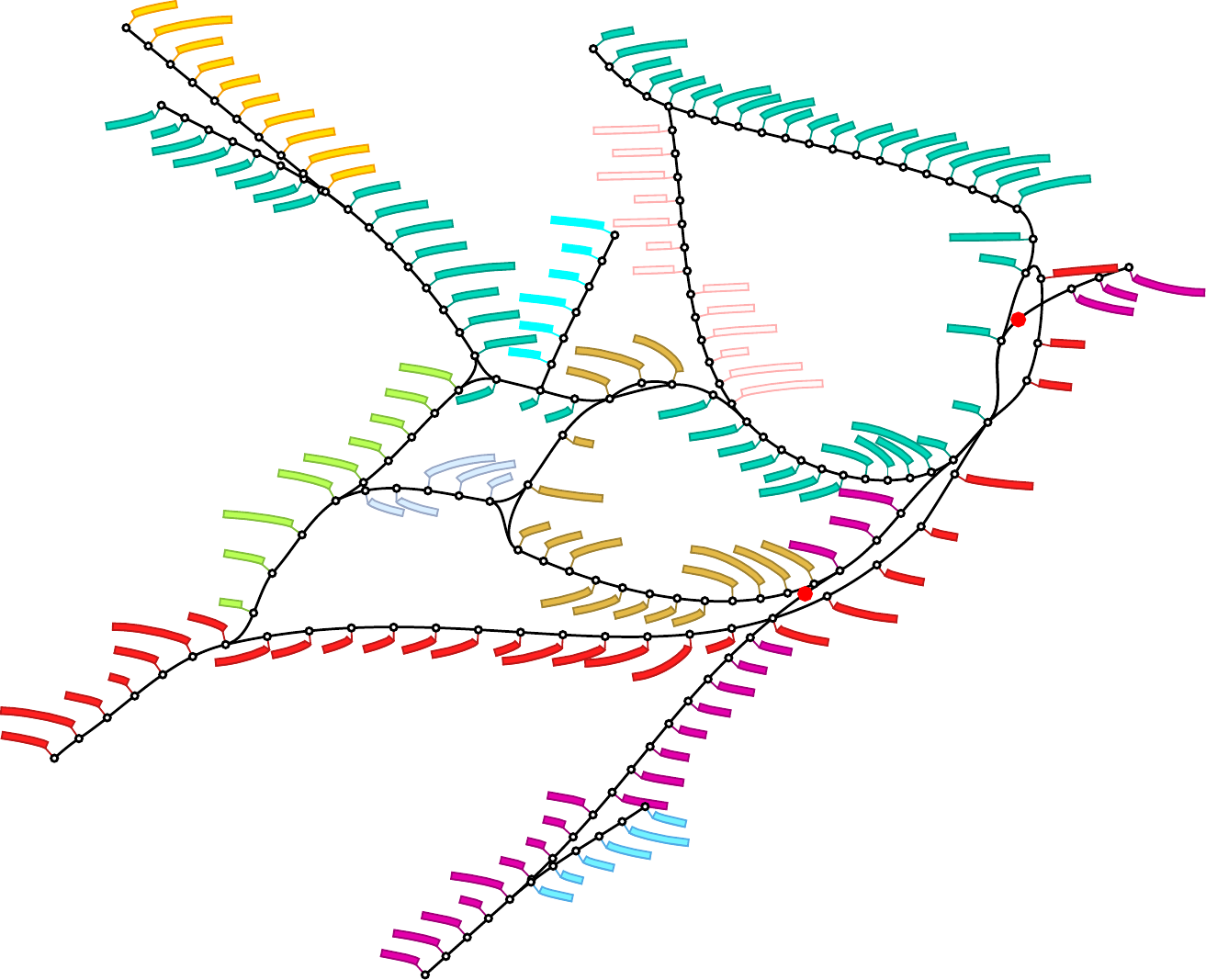}
      \label{fig:labelings:sydney4_scale}
   }
    \subfigure[ \ILPALG]{
\includegraphics[width=\scaleD\textwidth]{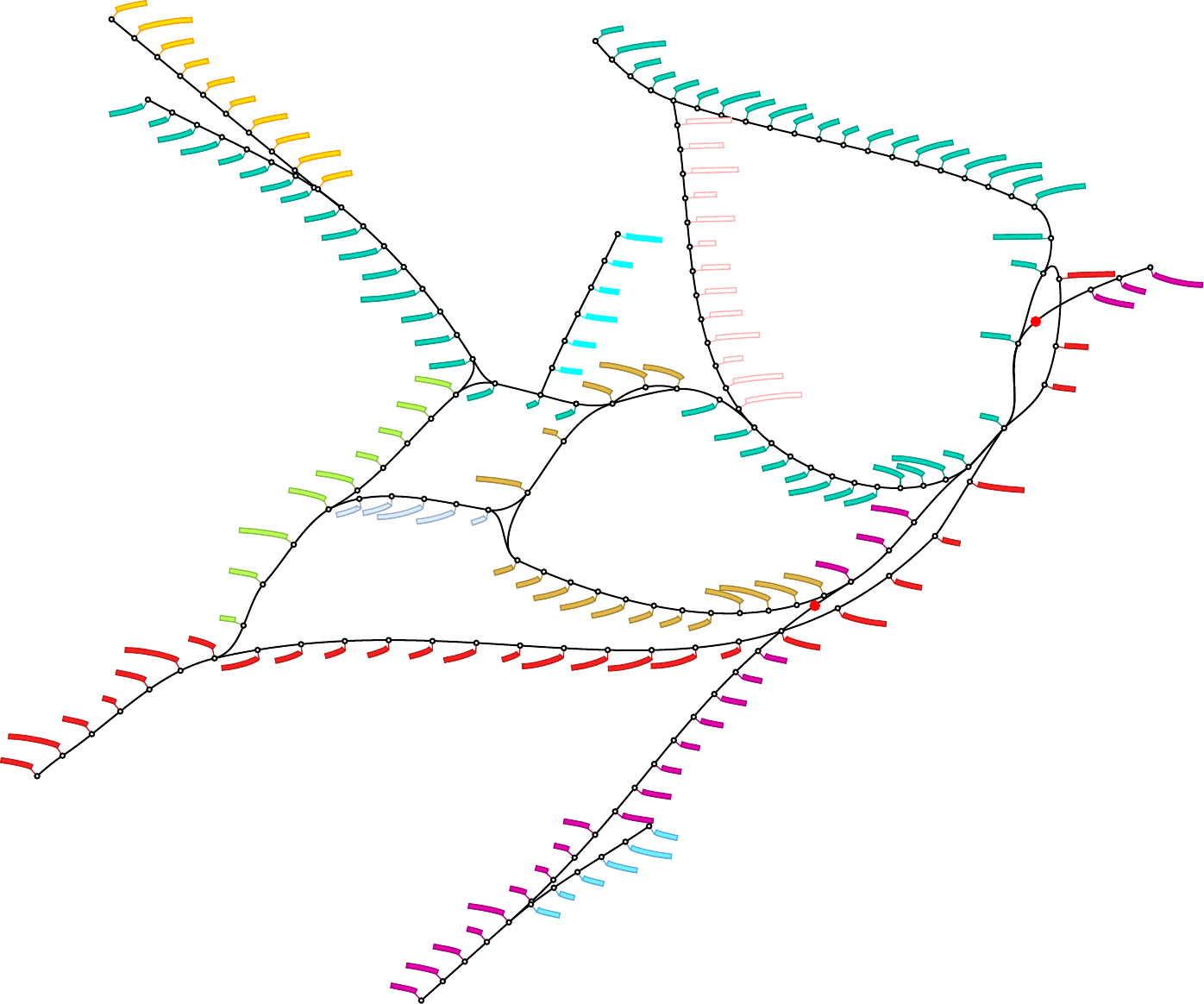}
      \label{fig:labelings:sydney4_ilp}
   } 
 \caption{Labelings for instance \emph{Sydney4}.}
 \end{figure}

 \newcommand{\scaleE}{0.48}
  \begin{figure}[!h]
    \centering 
    \subfigure[\DYNALG]{
\includegraphics[width=\scaleE\textwidth]{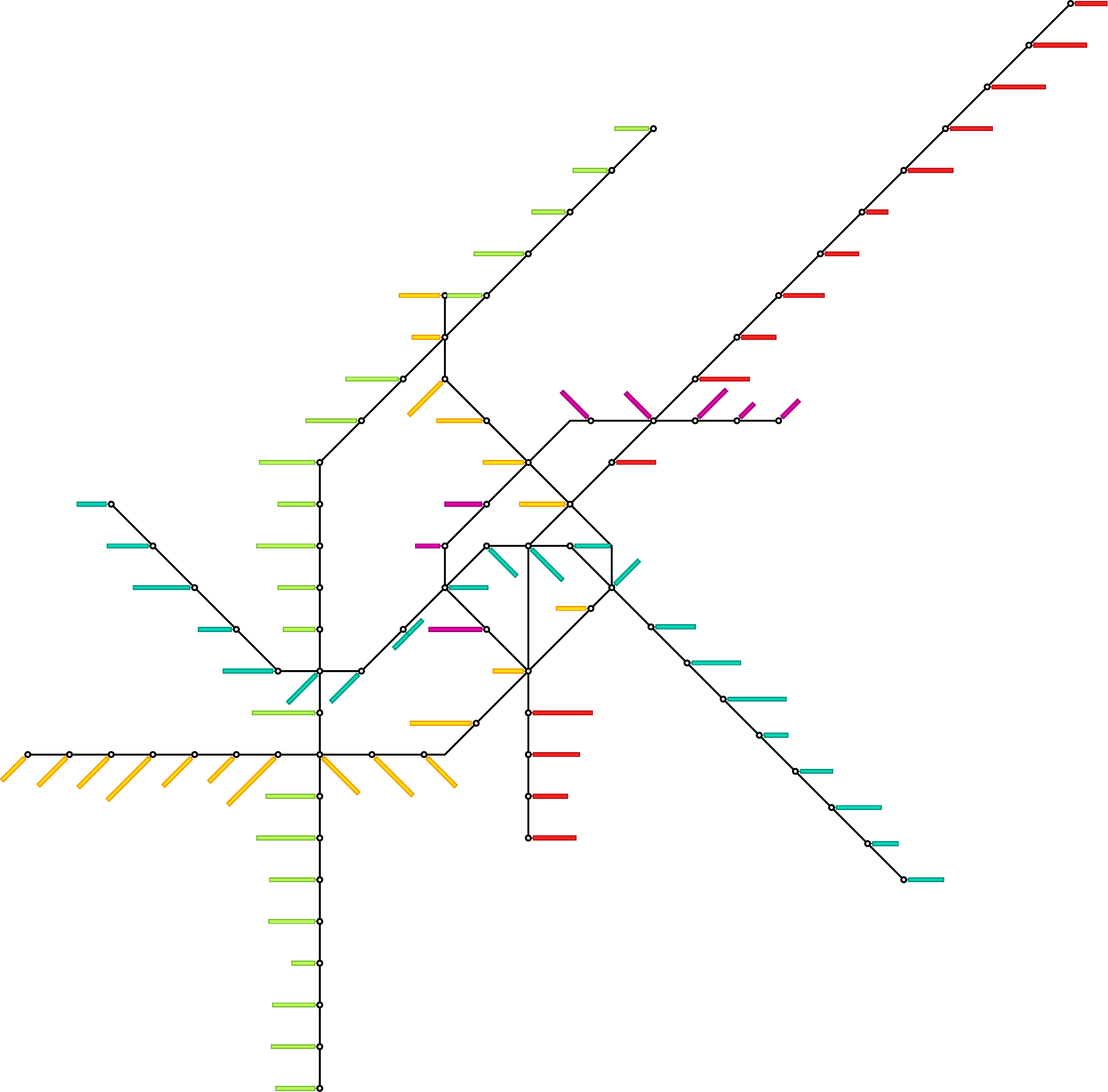}
      \label{fig:labelings:vienna1_dp}
   }
    \subfigure[ \GREEDYALG]{
\includegraphics[width=\scaleE\textwidth]{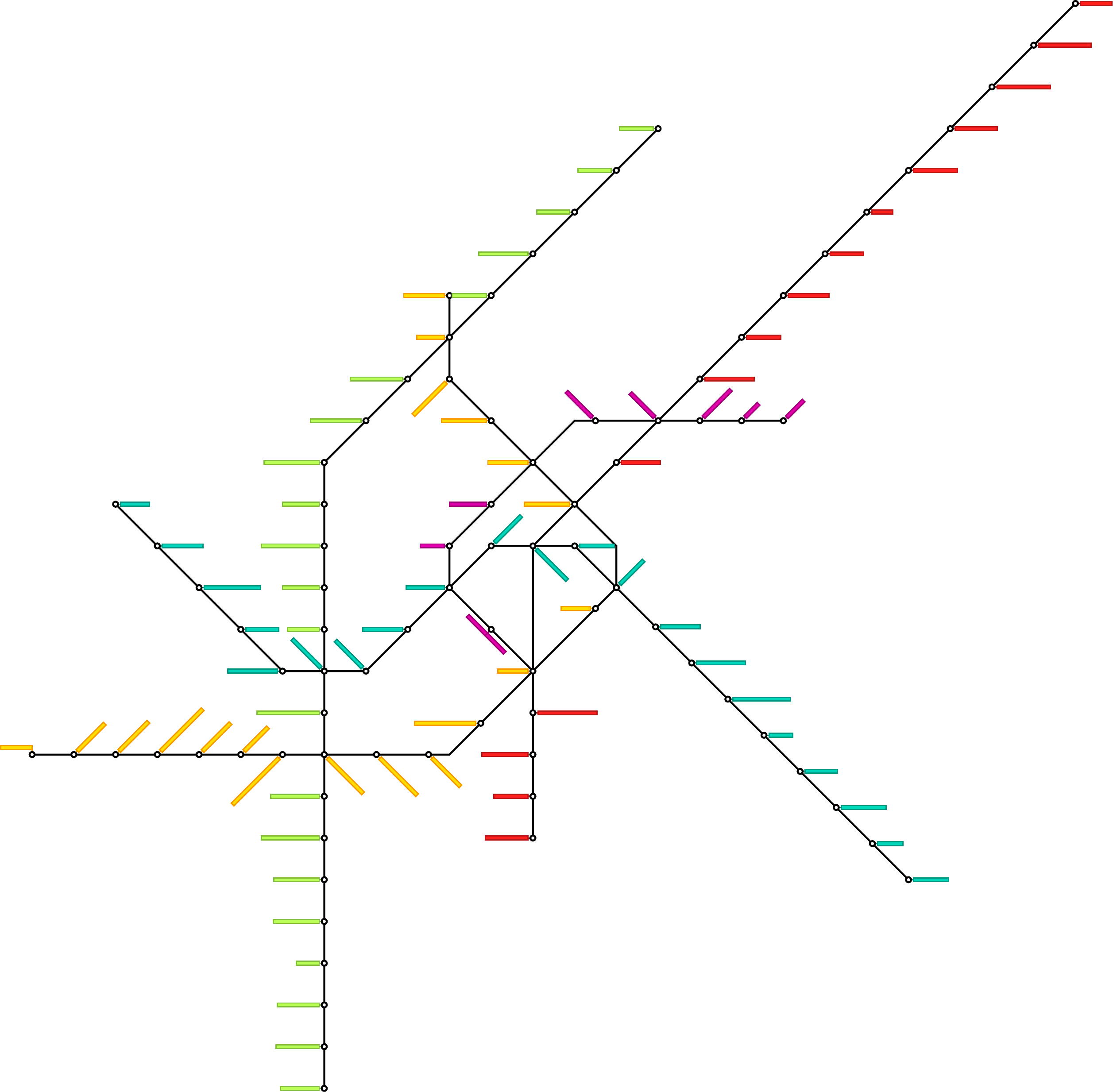}
      \label{fig:labelings:vienna1_g}
   } 

    \subfigure[ \SCALEALG]{
\includegraphics[width=\scaleE\textwidth]{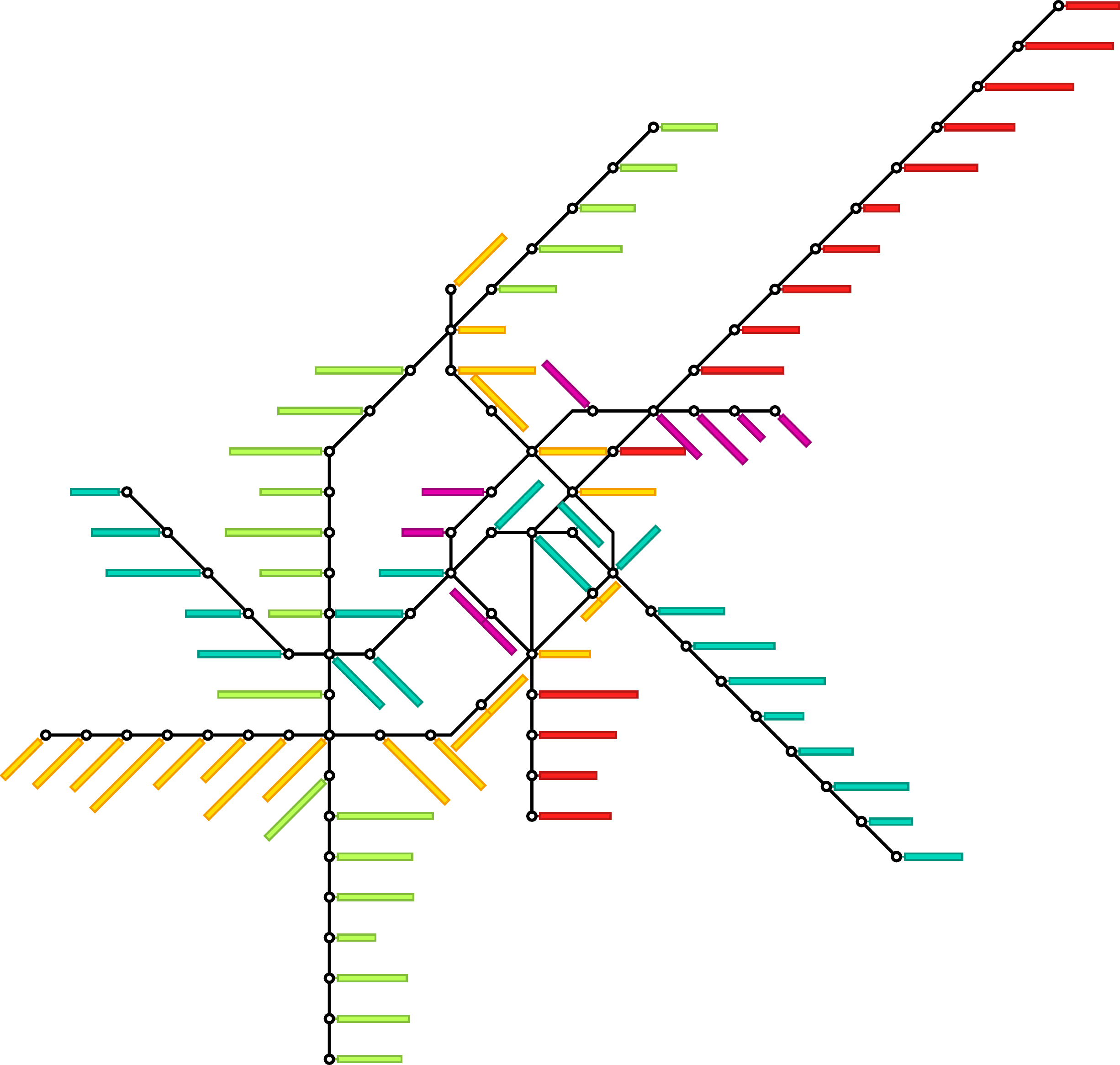}
      \label{fig:labelings:vienna1_scale}
   }
    \subfigure[\ILPALG]{
\includegraphics[width=\scaleE\textwidth]{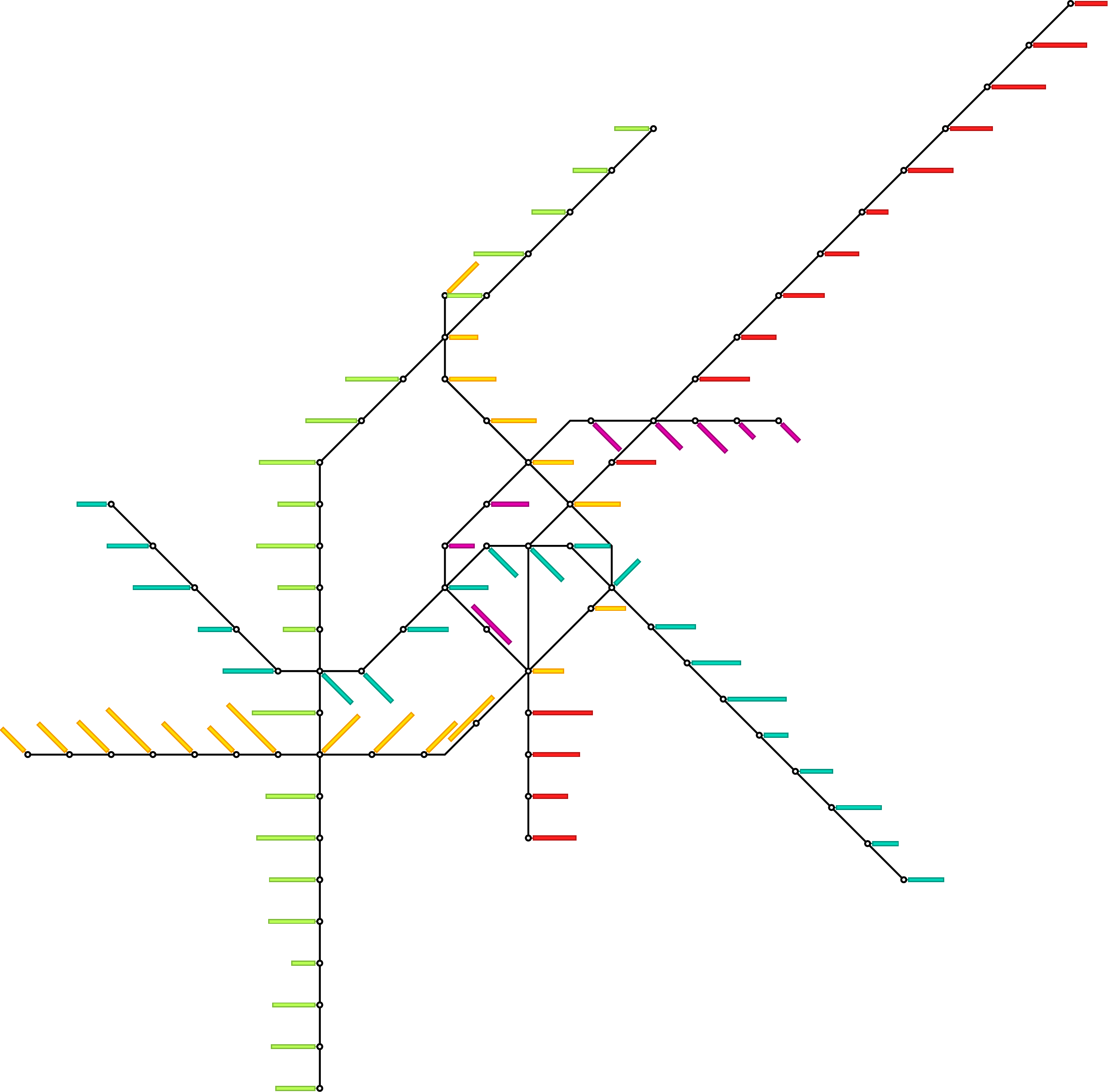}
      \label{fig:labelings:vienna1_ilp}
   } 
 \caption{Labelings for instance \emph{Vienna1}.}
 \end{figure}

 \newcommand{\scaleF}{0.48}
  \begin{figure}[!h]
    \centering 
    \subfigure[\DYNALG]{
\includegraphics[width=\scaleF\textwidth]{octi_vienna_nw_fig13b_dp}
      \label{fig:labelings:vienna2_dp}
   }
    \subfigure[ \GREEDYALG]{
\includegraphics[width=\scaleE\textwidth]{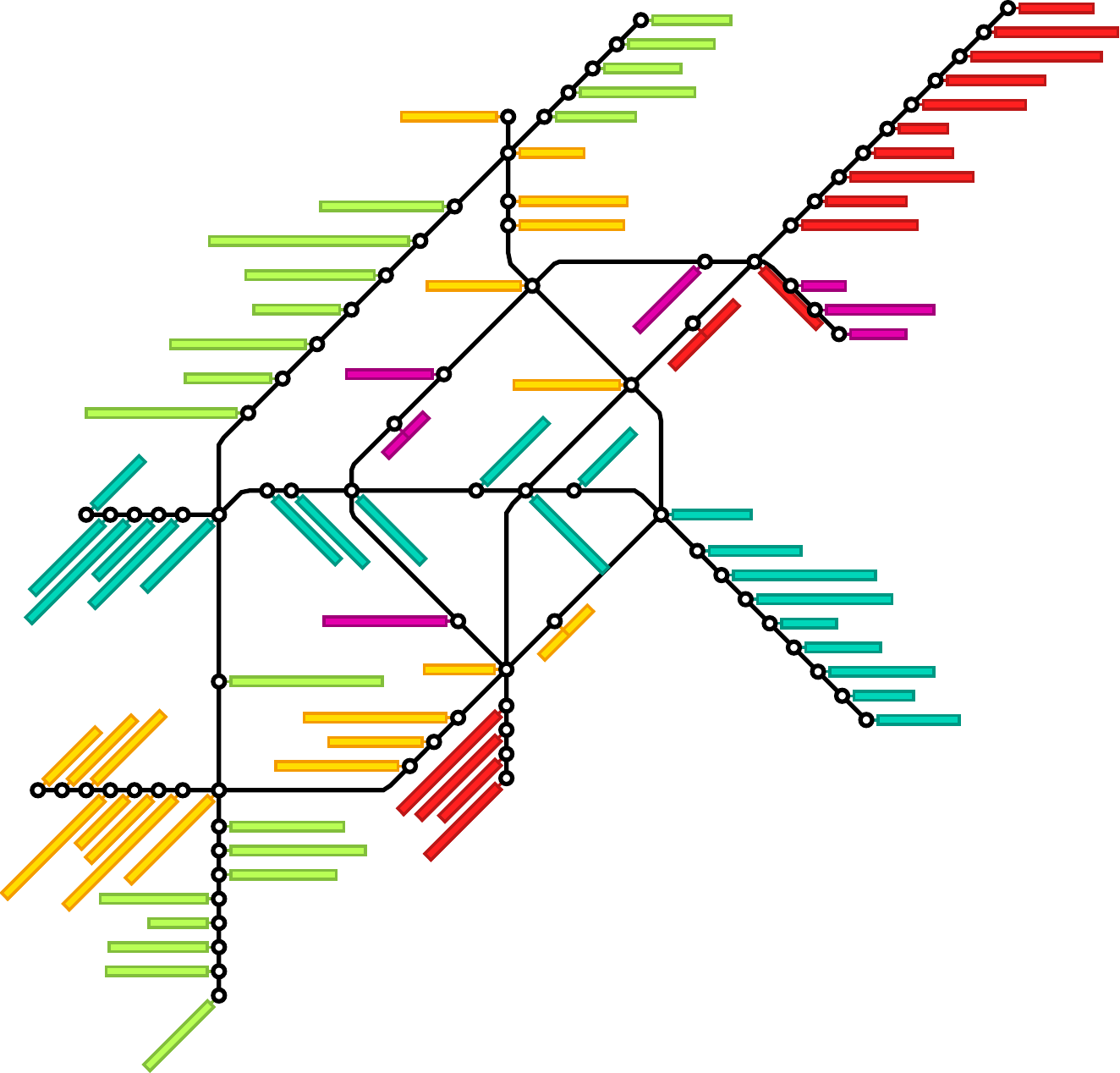}
      \label{fig:labelings:vienna2_g}
   } 

    \subfigure[ \SCALEALG]{
\includegraphics[width=\scaleE\textwidth]{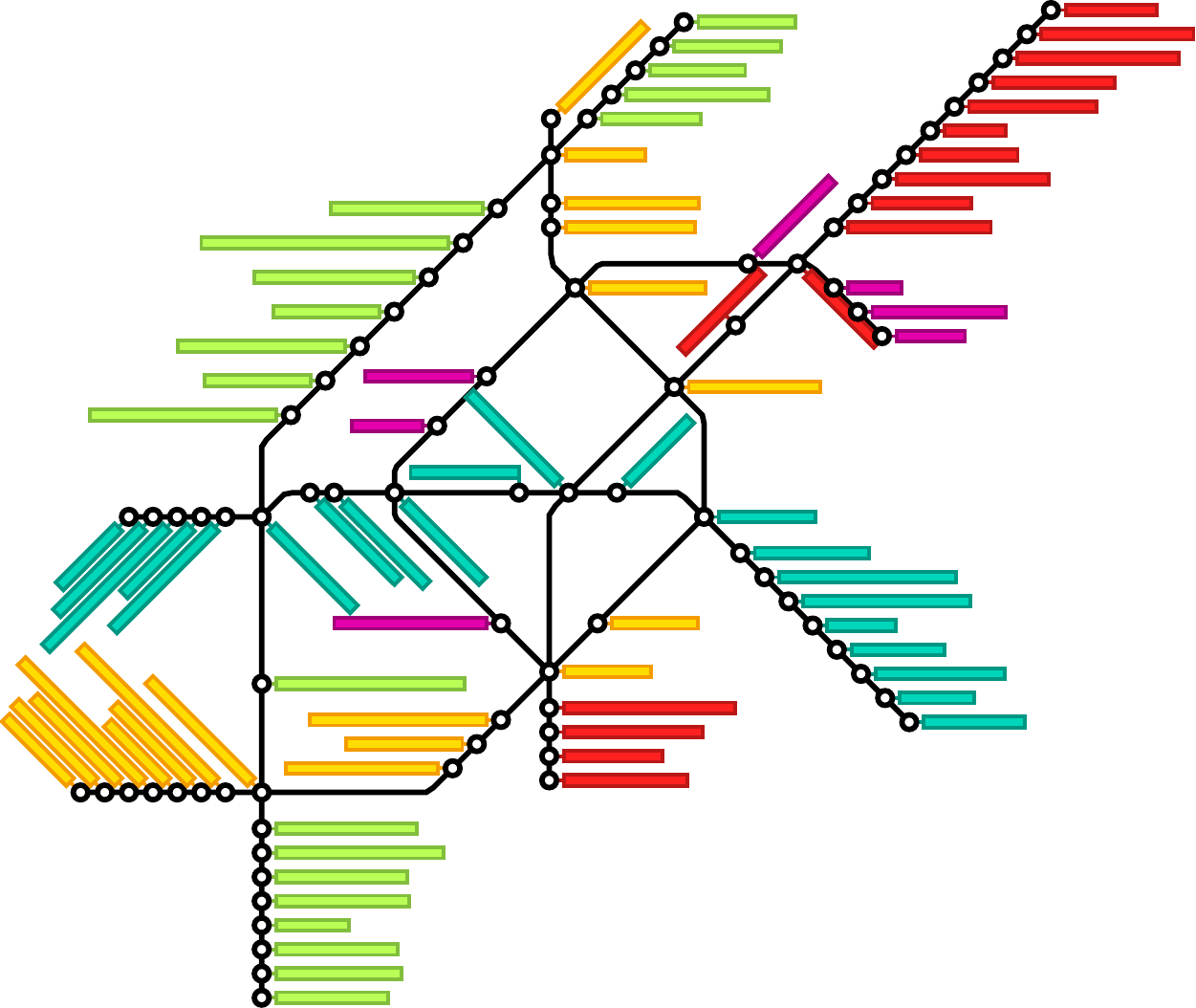}
      \label{fig:labelings:vienna2_scale}
   }
    \subfigure[\ILPALG]{
\includegraphics[width=\scaleE\textwidth]{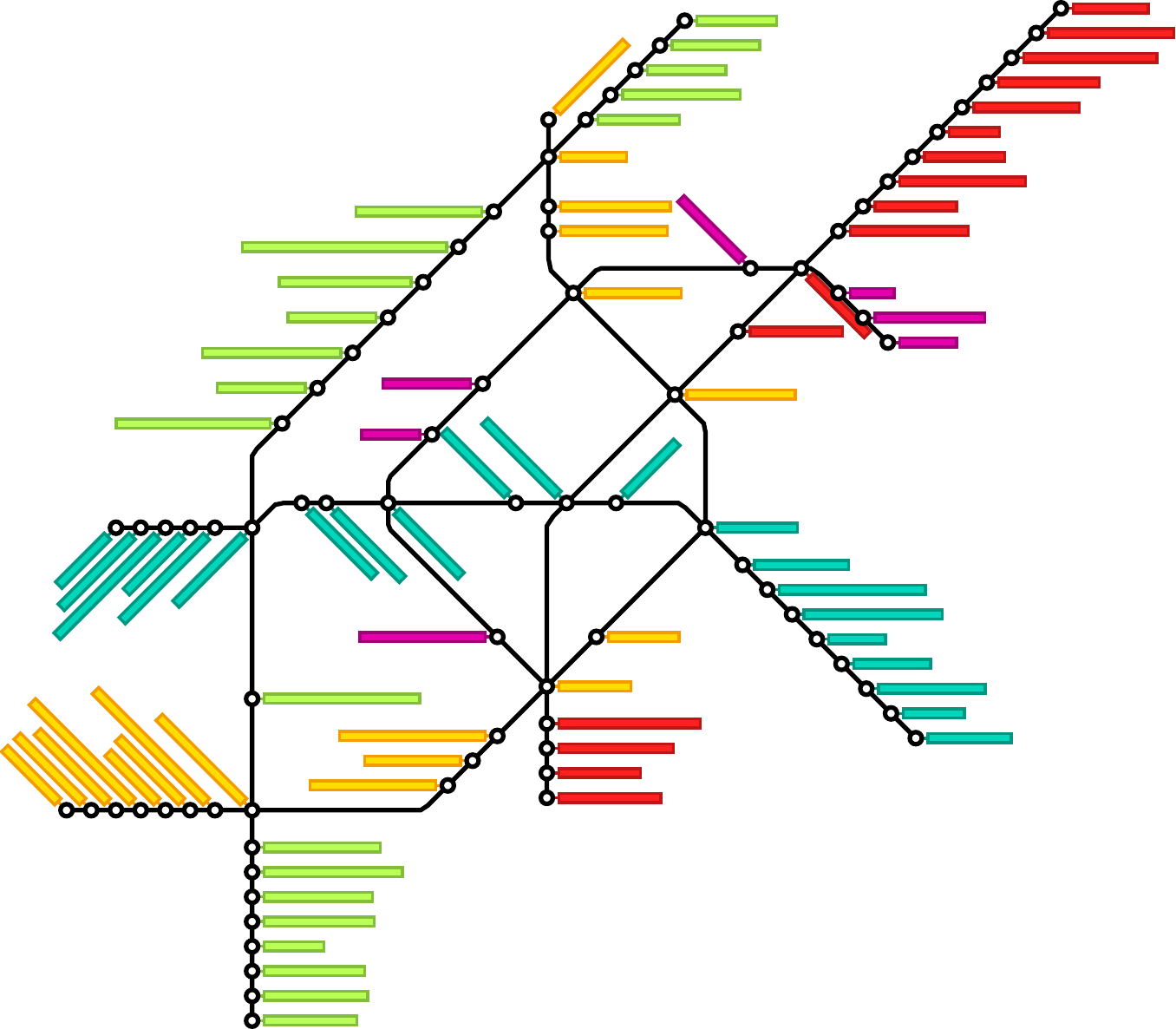}
      \label{fig:labelings:vienna2_ilp}
   } 
 \caption{Labelings for instance \emph{Vienna2}.}
 \end{figure}

 \newcommand{\scaleG}{0.48}
  \begin{figure}[!h]
    \centering 
    \subfigure[\DYNALG]{
\includegraphics[width=\scaleG\textwidth]{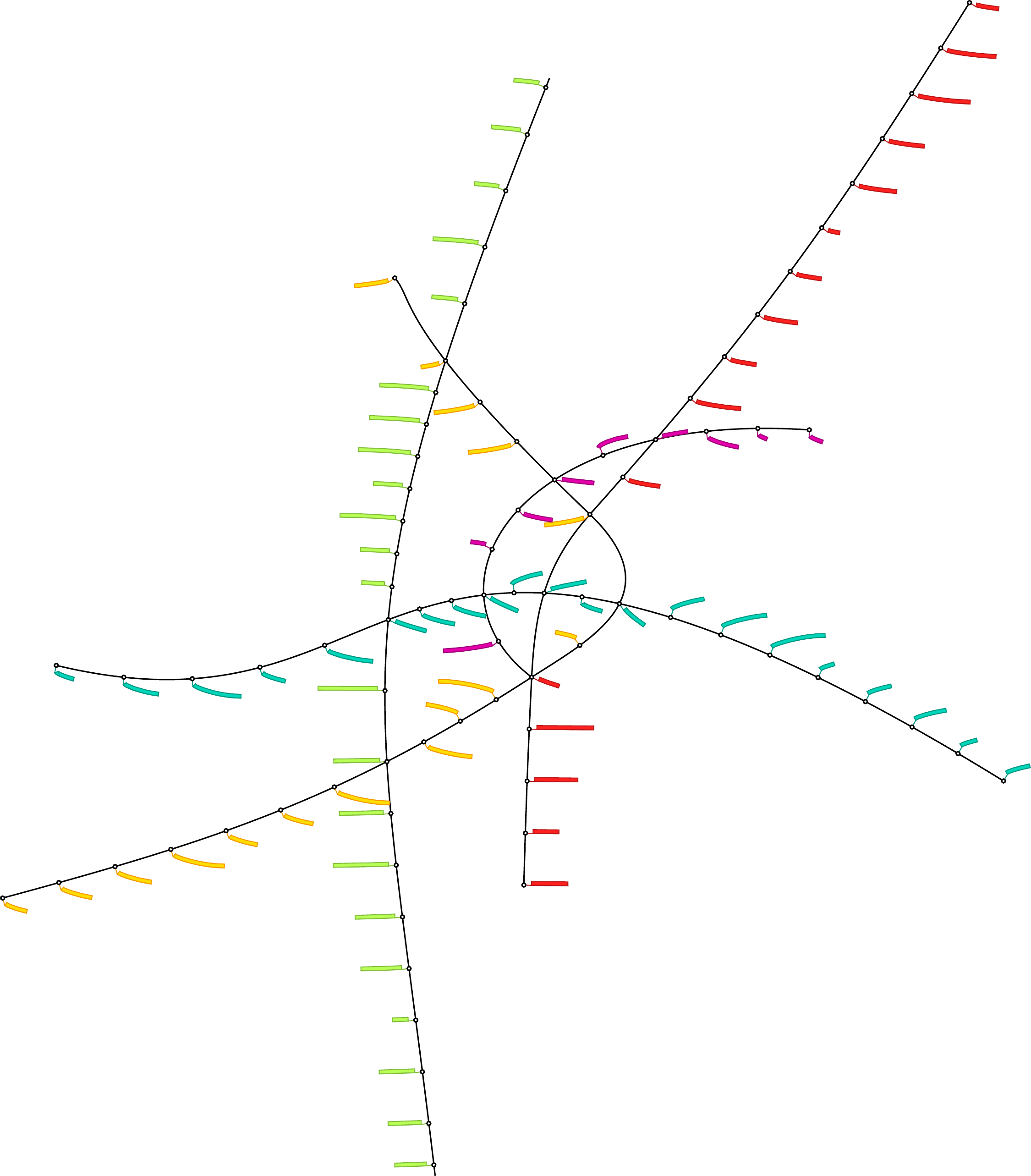}
      \label{fig:labelings:vienna3_dp}
   }
    \subfigure[ \GREEDYALG]{
\includegraphics[width=\scaleG\textwidth]{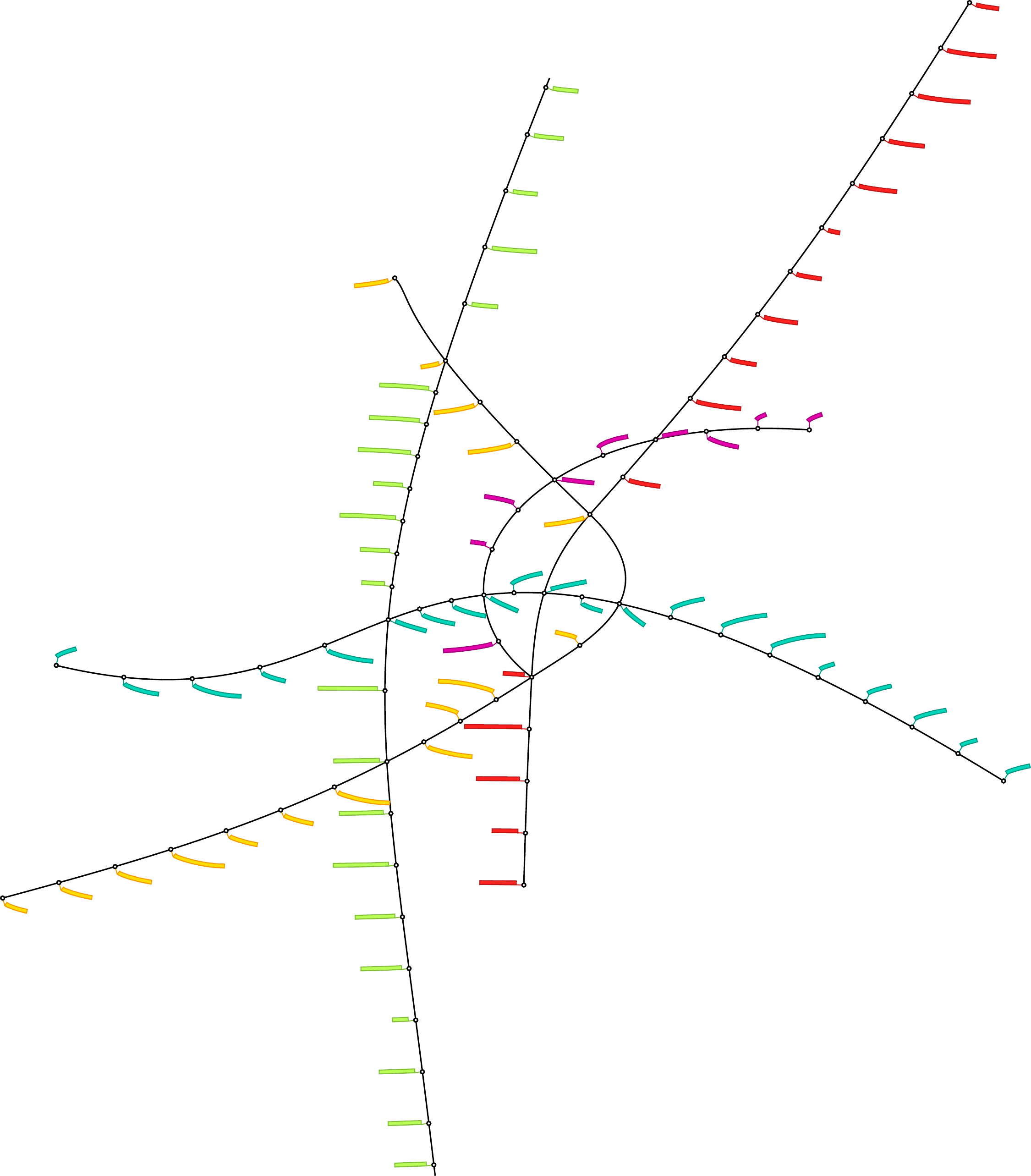}
      \label{fig:labelings:vienna3_g}
   } 

    \subfigure[ \SCALEALG]{
\includegraphics[width=\scaleG\textwidth]{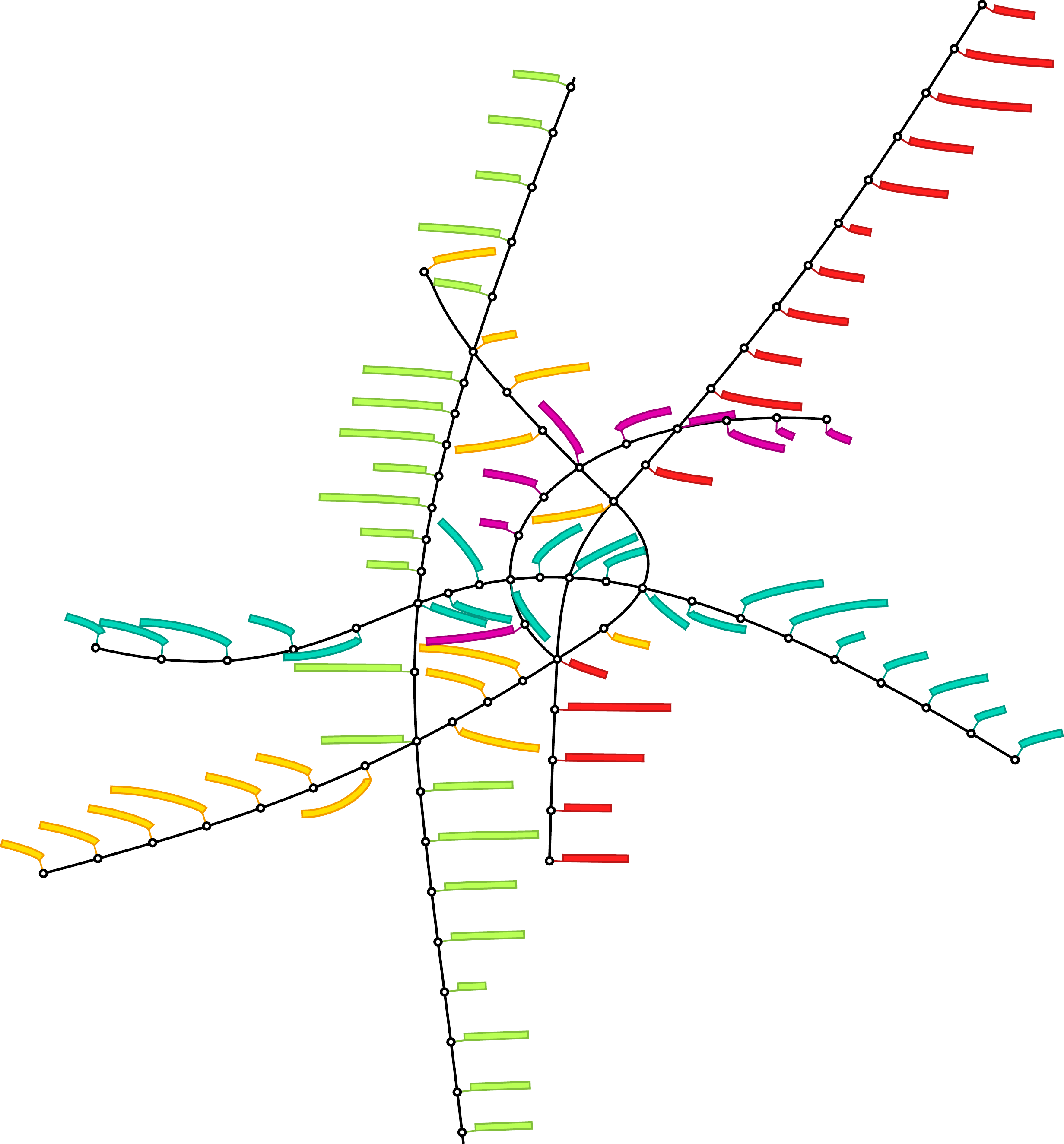}
      \label{fig:labelings:vienna3_scale}
   }
    \subfigure[\ILPALG]{
\includegraphics[width=\scaleG\textwidth]{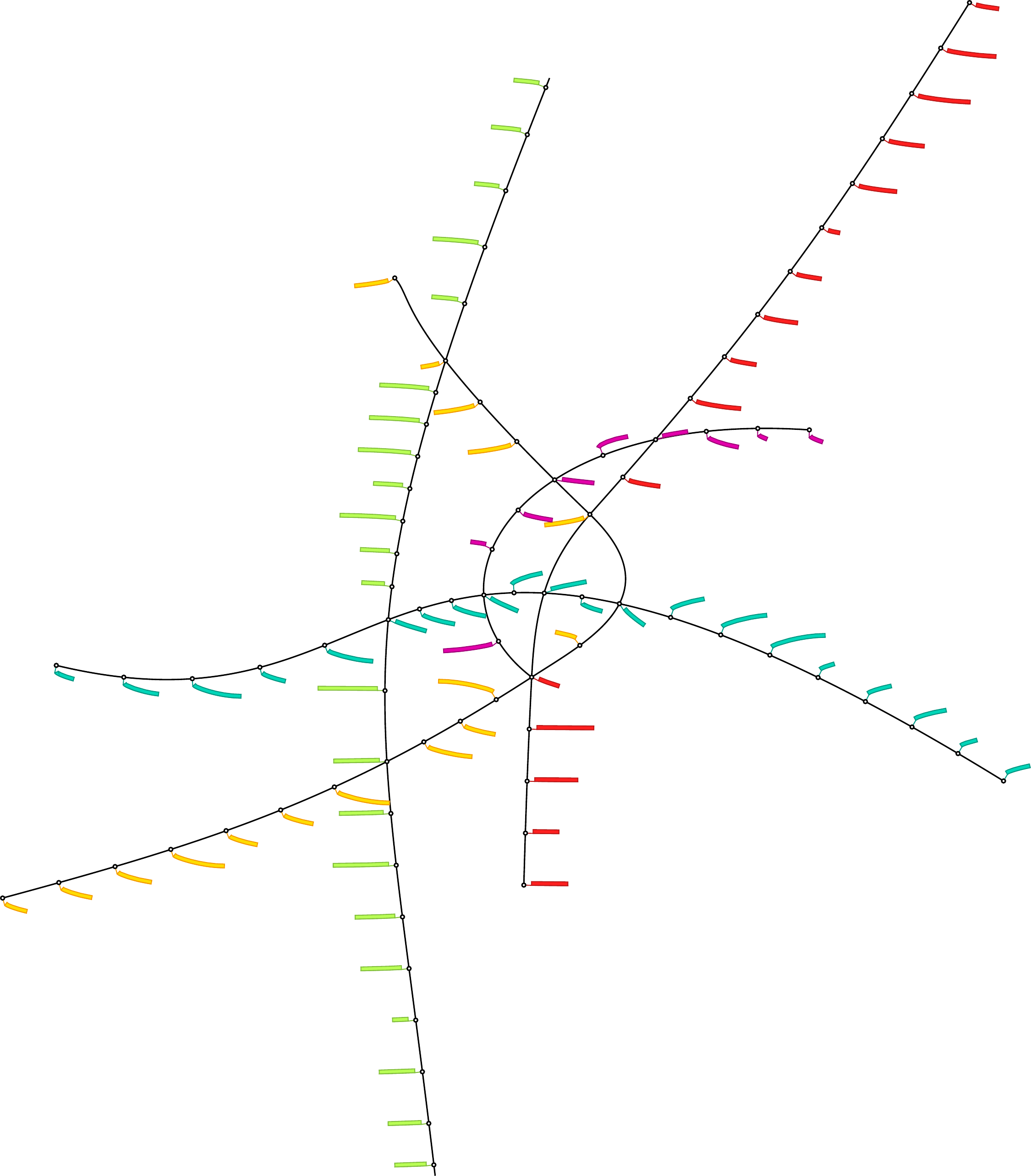}
      \label{fig:labelings:vienna3_ilp}
   } 
 \caption{Labelings for instance \emph{Vienna3}.}
 \end{figure}

\end{document}